\documentclass[%
superscriptaddress,
nofootinbib,
 amsmath,amssymb,
 aps,
pra,
]{revtex4-2}

\usepackage{graphicx}
\usepackage{dcolumn}
\usepackage{bm}

\usepackage{revsymb, amsmath, amsfonts, amssymb, enumerate, fullpage, amsthm, graphicx, braket, relsize, bbm, mathrsfs, mathtools}
\usepackage[normalem]{ulem}
\DeclarePairedDelimiter{\ev}{\langle}{\rangle}
\usepackage{graphicx,color}
\usepackage{paralist}
\usepackage[caption=false]{subfig}
\usepackage{xfrac}
\usepackage{tikzit}
\usepackage{circuitikz} 
\usepackage[linktocpage=true, colorlinks=true, linkcolor=red, urlcolor=blue, citecolor=blue]{hyperref}

\usepackage{keycommand}

\newkeycommand{\pointmap}[style=point][1]{\,\tikz{\node[style=\commandkey{style}] (x) at (0,-0.3) {$#1$};\node [style=none] (3) at (0, 1.0) {};\node [style=none] (3) at (0, -1.0) {};\draw (x)--(0,0.7);}\,}

\newkeycommand{\typedkpoint}[style=kpoint,edgestyle=][2]{%
\,\begin{tikzpicture}
  \begin{pgfonlayer}{nodelayer}
    \node [style=none] (0) at (0, 1) {};
    \node [style=\commandkey{style}] (1) at (0, -0.75) {$#2$};
    \node [style=right label] (2) at (0.25, 0.5) {$#1$};
  \end{pgfonlayer}
  \begin{pgfonlayer}{edgelayer}
    \draw [\commandkey{edgestyle}] (1) to (0.center);
  \end{pgfonlayer}
\end{tikzpicture}\,}

\newkeycommand{\typedkpointdag}[style=kpoint adjoint,edgestyle=][2]{%
\,\begin{tikzpicture}
  \begin{pgfonlayer}{nodelayer}
    \node [style=none] (0) at (0, -1) {};
    \node [style=\commandkey{style}] (1) at (0, 0.75) {$#2$};
    \node [style=right label] (2) at (0.25, -0.5) {$#1$};
  \end{pgfonlayer}
  \begin{pgfonlayer}{edgelayer}
    \draw [\commandkey{edgestyle}] (0.center) to (1);
  \end{pgfonlayer}
\end{tikzpicture}\,}

\newcommand{\bistateadj}[1]{\,\begin{tikzpicture}[yshift=-3mm]
  \begin{pgfonlayer}{nodelayer}
    \node [style=kpointadj, minimum width=1 cm, inner sep=2pt] (0) at (0, 0.5) {$\psi$};
    \node [style=none] (1) at (-0.75, 0.25) {};
    \node [style=none] (2) at (0.75, 0.25) {};
    \node [style=none] (3) at (-0.75, -0.5) {};
    \node [style=none] (4) at (0.75, -0.5) {};
  \end{pgfonlayer}
  \begin{pgfonlayer}{edgelayer}
    \draw (1.center) to (3.center);
    \draw (2.center) to (4.center);
  \end{pgfonlayer}
\end{tikzpicture}\,}

\newkeycommand{\pointketbra}[style1=copoint,style2=point][2]{\,%
\begin{tikzpicture}
  \begin{pgfonlayer}{nodelayer}
    \node [style=none] (0) at (0, -1.5) {};
    \node [style=\commandkey{style1}] (1) at (0, -0.75) {$#1$};
    \node [style=\commandkey{style2}] (2) at (0, 0.75) {$#2$};
    \node [style=none] (3) at (0, 1.5) {};
  \end{pgfonlayer}
  \begin{pgfonlayer}{edgelayer}
    \draw (0.center) to (1);
    \draw (2) to (3.center);
  \end{pgfonlayer}
\end{tikzpicture}\,}

\newkeycommand{\twocopointketbra}[style1=copoint,style2=copoint,style3=point][3]{\,%
\begin{tikzpicture}
  \begin{pgfonlayer}{nodelayer}
    \node [style=none] (0) at (-0.75, -1.5) {};
    \node [style=none] (1) at (0.75, -1.5) {};
    \node [style=\commandkey{style1}] (2) at (-0.75, -0.75) {$#1$};
    \node [style=\commandkey{style2}] (3) at (0.75, -0.75) {$#2$};
    \node [style=\commandkey{style3}] (4) at (0, 0.75) {$#3$};
    \node [style=none] (5) at (0, 1.5) {};
  \end{pgfonlayer}
  \begin{pgfonlayer}{edgelayer}
    \draw (0.center) to (2);
    \draw (1.center) to (3);
    \draw (4) to (5.center);
  \end{pgfonlayer}
\end{tikzpicture}\,}

\newkeycommand{\twopointketbra}[style1=copoint,style2=point,style3=point][3]{\,%
\begin{tikzpicture}
  \begin{pgfonlayer}{nodelayer}
    \node [style=none] (0) at (0, -1.5) {};
    \node [style=none] (1) at (-0.75, 1.5) {};
    \node [style=\commandkey{style1}] (2) at (0, -0.75) {$#1$};
    \node [style=\commandkey{style2}] (3) at (-0.75, 0.75) {$#2$};
    \node [style=\commandkey{style3}] (4) at (0.75, 0.75) {$#3$};
    \node [style=none] (5) at (0.75, 1.5) {};
  \end{pgfonlayer}
  \begin{pgfonlayer}{edgelayer}
    \draw (0.center) to (2);
    \draw (1.center) to (3);
    \draw (4) to (5.center);
  \end{pgfonlayer}
\end{tikzpicture}\,}

\newkeycommand{\pointbraket}[style1=point,style2=copoint][2]{\,%
\begin{tikzpicture}
  \begin{pgfonlayer}{nodelayer}
    \node [style=\commandkey{style1}] (0) at (0, -0.5) {$#1$};
    \node [style=\commandkey{style2}] (1) at (0, 0.5) {$#2$};
  \end{pgfonlayer}
  \begin{pgfonlayer}{edgelayer}
    \draw (0) to (1);
  \end{pgfonlayer}
\end{tikzpicture}\,}

\newkeycommand{\kpointbraket}[style1=kpoint,style2=kpoint adjoint][2]{\,%
\begin{tikzpicture}
  \begin{pgfonlayer}{nodelayer}
    \node [style=\commandkey{style1}] (0) at (0, -0.75) {$#1$};
    \node [style=\commandkey{style2}] (1) at (0, 0.75) {$#2$};
  \end{pgfonlayer}
  \begin{pgfonlayer}{edgelayer}
    \draw (0) to (1);
  \end{pgfonlayer}
\end{tikzpicture}\,}

\newkeycommand{\kpointmap}[style1=kpoint,style2=map][2]{\,%
\begin{tikzpicture}
  \begin{pgfonlayer}{nodelayer}
    \node [style=\commandkey{style1}] (0) at (0, -1.1) {$#1$};
    \node [style=\commandkey{style2}] (1) at (0, 0.2) {$#2$};
    \node [style=none] (2) at (0, 1.5) {};
  \end{pgfonlayer}
  \begin{pgfonlayer}{edgelayer}
    \draw (0) to (1);
    \draw (1) to (2.center);
  \end{pgfonlayer}
\end{tikzpicture}%
\,}

\newkeycommand{\sandwichtwo}[style1=point,style2=map,style3=map,style4=copoint][4]{\,%
\begin{tikzpicture}
  \begin{pgfonlayer}{nodelayer}
    \node [style=\commandkey{style2}] (0) at (0, -0.75) {$#2$};
    \node [style=\commandkey{style3}] (1) at (0, 0.75) {$#3$};
    \node [style=\commandkey{style1}] (2) at (0, -2) {$#1$};
    \node [style=\commandkey{style4}] (3) at (0, 2) {$#4$};
  \end{pgfonlayer}
  \begin{pgfonlayer}{edgelayer}
    \draw (2) to (0);
    \draw (0) to (1);
    \draw (1) to (3);
  \end{pgfonlayer}
\end{tikzpicture}\,}

\newkeycommand{\longbraket}[style1=point,style2=copoint][2]{\,%
\begin{tikzpicture}
  \begin{pgfonlayer}{nodelayer}
    \node [style=\commandkey{style1}] (0) at (0, -2) {$#1$};
    \node [style=\commandkey{style2}] (1) at (0, 2) {$#2$};
  \end{pgfonlayer}
  \begin{pgfonlayer}{edgelayer}
    \draw (0) to (1);
  \end{pgfonlayer}
\end{tikzpicture}\,}

\newkeycommand{\onb}[style=point]{\ensuremath{\left\{\tikz{\node[style=\commandkey{style}] (x) at (0,-0.3) {$j$};\draw (x)--(0,0.7);}\right\}}\xspace}

\newkeycommand{\copointmap}[style=copoint][1]{\,\tikz{\node[style=\commandkey{style}] (x) at (0,0.3) {$#1$};\draw (x)--(0,-0.7);}\,}

\tikzstyle{cdiag}=[matrix of math nodes, row sep=3em, column sep=3em, text height=1.5ex, text depth=0.25ex,inner sep=0.5em]
\tikzstyle{arrow above}=[transform canvas={yshift=0.5ex}]
\tikzstyle{arrow below}=[transform canvas={yshift=-0.5ex}]

\usetikzlibrary{shapes.multipart}

\tikzstyle{infpoint}=[regular polygon,regular polygon sides=3,draw,scale=0.75,inner sep=-0.5pt,minimum width=9mm,fill=white,regular polygon rotate=90]
\tikzstyle{infcopoint}=[regular polygon,regular polygon sides=3,draw,scale=0.75,inner sep=-0.5pt,minimum width=9mm,fill=white,regular polygon rotate=270]
\tikzstyle{fMeas}=[regular polygon,regular polygon sides=3,draw,scale=0.75,inner sep=-0.5pt,minimum width=4mm,fill=white,regular polygon rotate=0,line width=1pt]
\tikzstyle{fPrep}=[regular polygon,regular polygon sides=3,draw,scale=0.75,inner sep=-0.5pt,minimum width=4mm,fill=white,regular polygon rotate=180,line width=1pt]
\tikzstyle{fTransition}=[fill=white,draw, line width = 1pt,inner sep=0.6mm,font=\footnotesize,minimum height=2mm,minimum width=2mm]
 \tikzstyle{rightground}=[circuit ee IEC,thick,ground,rotate=0,xscale=2.5,yscale=2]

\tikzstyle{bwSpider}=[
       rectangle split,
       rectangle split parts=2,
       rectangle split part fill={black,white},
 minimum size=3.6 mm, inner sep=-2mm, draw=black,scale=0.5,rounded corners=0.8 mm
       ]
 \tikzstyle{wbSpider}=[
       rectangle split,
       rectangle split parts=2,
       rectangle split part fill={white,black},
 minimum size=3.6 mm, inner sep=-2mm, draw=black,scale=0.5,rounded corners=0.8 mm
       ]
\tikzstyle{qWire}=[line width = 1pt, color=black]
\tikzstyle{pWire}=[line width = 1pt, color=red!60!black]
\tikzstyle{gWire}=[line width = 1.5pt, color=green!60!black]
\tikzstyle{cWire}=[color=gray,thin]
\tikzstyle{env}=[copoint,regular polygon rotate=0,minimum width=0.2cm, fill=black]

\tikzstyle{probs}=[shape=semicircle,fill=white,draw=black,shape border rotate=180,minimum width=1.2cm]

\tikzstyle{every picture}=[baseline=-0.25em,scale=0.5]
\tikzstyle{dotpic}=[] 
\tikzstyle{diredges}=[every to/.style={diredge}]
\tikzstyle{math matrix}=[matrix of math nodes,left delimiter=(,right delimiter=),inner sep=2pt,column sep=1em,row sep=0.5em,nodes={inner sep=0pt},text height=1.5ex, text depth=0.25ex]

\tikzstyle{inline text}=[text height=1.5ex, text depth=0.25ex,yshift=0.5mm]
\tikzstyle{label}=[font=\footnotesize,text height=1.5ex, text depth=0.25ex,yshift=0.5mm]
\tikzstyle{left label}=[label,anchor=east,xshift=1.5mm]
\tikzstyle{right label}=[label,anchor=west,xshift=-1mm]
\tikzstyle{up label}=[label,anchor=south,yshift=-1mm]

\tikzstyle{braceedge}=[decorate,decoration={brace,amplitude=2mm,raise=-1mm}]
\tikzstyle{small braceedge}=[decorate,decoration={brace,amplitude=1mm,raise=-1mm}]

\tikzstyle{doubled}=[line width=1.6pt] 
\tikzstyle{boldedge}=[doubled,shorten <=-0.17mm,shorten >=-0.17mm]
\tikzstyle{boldedgegray}=[doubled,gray,shorten <=-0.17mm,shorten >=-0.17mm]
\tikzstyle{singleedgegray}=[gray]

\tikzstyle{semidoubled}=[line width=1.4pt] 
\tikzstyle{semiboldedgegray}=[semidoubled,gray,shorten <=-0.17mm,shorten >=-0.17mm]

\tikzstyle{boxedge}=[semiboldedgegray]

\tikzstyle{boldedgedashed}=[very thick,dashed,shorten <=-0.17mm,shorten >=-0.17mm]
\tikzstyle{vboldedgedashed}=[doubled,dashed,shorten <=-0.17mm,shorten >=-0.17mm]
\tikzstyle{left hook arrow}=[left hook-latex]
\tikzstyle{right hook arrow}=[right hook-latex]
\tikzstyle{sembracket}=[line width=0.5pt,shorten <=-0.07mm,shorten >=-0.07mm]

\tikzstyle{causal edge}=[->,thick,gray]
\tikzstyle{causal nondir}=[thick,gray]
\tikzstyle{timeline}=[thick,gray, dashed]

\tikzstyle{cedge}=[<->,thick,gray!70!white]

\tikzstyle{empty diagram}=[draw=gray!40!white,dashed,shape=rectangle,minimum width=1cm,minimum height=1cm]
\tikzstyle{empty diagram small}=[draw=gray!50!white,dashed,shape=rectangle,minimum width=0.6cm,minimum height=0.5cm]

\tikzstyle{dot}=[inner sep=0mm,minimum width=2mm,minimum height=2mm,draw,shape=circle]
\tikzstyle{leak}=[white dot, shape=regular polygon, minimum size=3.3 mm, regular polygon sides=3, outer sep=-0.2mm, regular polygon rotate=270]
\tikzstyle{proj}=[draw,fill=white,chamfered rectangle,chamfered rectangle angle=30, minimum width=2mm, minimum height= 1mm,scale=0.5, outer sep=-0.2mm]
\tikzstyle{wide proj}=[draw,fill=white,chamfered rectangle,chamfered rectangle angle=30, minimum width=15mm,minimum height=1mm,scale=0.5, outer sep=-0.2mm]
\tikzstyle{very wide proj}=[draw,fill=white,chamfered rectangle,chamfered rectangle angle=30, minimum width=25mm,minimum height=1mm,scale=0.5, outer sep=-0.2mm]
\tikzstyle{very very wide proj}=[draw,fill=white,chamfered rectangle,chamfered rectangle angle=30, minimum width=35mm,minimum height=1mm,scale=0.5, outer sep=-0.2mm]
\tikzstyle{preleak}=[proj]

\tikzstyle{split proj out}=[regular polygon,regular polygon sides=3,draw,scale=0.75,inner sep=-0.5pt,minimum width=3.3mm,fill=white,regular polygon rotate=180]
\tikzstyle{split proj in}=[regular polygon,regular polygon sides=3,draw,scale=0.75,inner sep=-0.5pt,minimum width=3.3mm,fill=white]
\tikzstyle{Vleak}=[white dot, shape=regular polygon, minimum size=3.3 mm, regular polygon sides=3, outer sep=-0.2mm, regular polygon rotate=90]
\tikzstyle{dleak}=[white dot, line width=1.6pt, shape=regular polygon, minimum size=3.3 mm, regular polygon sides=3, outer sep=-0.2mm, regular polygon rotate=270]

\tikzstyle{Wsquare}=[white dot, shape=regular polygon, rounded corners=0.8 mm, minimum size=3.3 mm, regular polygon sides=3, outer sep=-0.2mm]
\tikzstyle{Wsquareadj}=[white dot, shape=regular polygon, rounded corners=0.8 mm, minimum size=3.3 mm, regular polygon sides=3, outer sep=-0.2mm, regular polygon rotate=180]
\tikzstyle{ddot}=[inner sep=0mm, doubled, minimum width=2.5mm,minimum height=2.5mm,draw,shape=circle]

\tikzstyle{black dot}=[dot,fill=black]
\tikzstyle{white dot}=[dot,fill=white,,text depth=-0.2mm]
\tikzstyle{white Wsquare}=[Wsquare,fill=gray,,text depth=-0.2mm]
\tikzstyle{white Wsquareadj}=[Wsquareadj,fill=white,,text depth=-0.2mm]
\tikzstyle{green dot}=[white dot] 
\tikzstyle{gray dot}=[dot,fill=gray!40!white,,text depth=-0.2mm]
\tikzstyle{red dot}=[gray dot] 

\tikzstyle{black ddot}=[ddot,fill=black]
\tikzstyle{white ddot}=[ddot,fill=white]
\tikzstyle{gray ddot}=[ddot,fill=gray!40!white]

\tikzstyle{gray edge}=[gray!60!white]

\tikzstyle{small dot}=[inner sep=0.5mm,minimum width=0pt,minimum height=0pt,draw,shape=circle]

\tikzstyle{small black dot}=[small dot,fill=black]
\tikzstyle{small white dot}=[small dot,fill=white]
\tikzstyle{small gray dot}=[small dot,fill=gray!40!white]

\tikzstyle{causal dot}=[inner sep=0.4mm,minimum width=0pt,minimum height=0pt,draw=white,shape=circle,fill=gray!40!white]

\tikzstyle{phase dimensions}=[minimum size=5mm,font=\footnotesize,rectangle,rounded corners=2.5mm,inner sep=0.2mm,outer sep=-2mm]
\tikzstyle{dphase dimensions}=[minimum size=5mm,font=\footnotesize,rectangle,rounded corners=2.5mm,inner sep=0.2mm,outer sep=-2mm]

\tikzstyle{white phase dot}=[dot,fill=white,phase dimensions]
\tikzstyle{white phase ddot}=[ddot,fill=white,dphase dimensions]

\tikzstyle{white rect ddot}=[draw=black,fill=white,doubled,minimum size=5mm,font=\footnotesize,rectangle,rounded corners=2.5mm,inner sep=0.2mm]
\tikzstyle{gray rect ddot}=[draw=black,fill=gray!40!white,doubled,minimum size=6mm,font=\footnotesize,rectangle,rounded corners=3mm]

\tikzstyle{gray phase dot}=[dot,fill=gray!40!white,phase dimensions]
\tikzstyle{gray phase ddot}=[ddot,fill=gray!40!white,dphase dimensions]
\tikzstyle{grey phase dot}=[gray phase dot]
\tikzstyle{grey phase ddot}=[gray phase ddot]

\tikzstyle{small phase dimensions}=[minimum size=4mm,font=\tiny,rectangle,rounded corners=2mm,inner sep=0.2mm,outer sep=-2mm]
\tikzstyle{small dphase dimensions}=[minimum size=4mm,font=\tiny,rectangle,rounded corners=2mm,inner sep=0.2mm,outer sep=-2mm]

\tikzstyle{small gray phase dot}=[dot,fill=gray!40!white,small phase dimensions]
\tikzstyle{small gray phase ddot}=[ddot,fill=gray!40!white,small dphase dimensions]

\tikzstyle{small map}=[draw,shape=rectangle,minimum height=4mm,minimum width=4mm,fill=white]

\tikzstyle{cnot}=[fill=white,shape=circle,inner sep=-1.4pt]

\tikzstyle{asym hadamard}=[fill=white,draw,shape=NEbox,inner sep=0.6mm,font=\footnotesize,minimum height=4mm]
\tikzstyle{asym hadamard conj}=[fill=white,draw,shape=NWbox,inner sep=0.6mm,font=\footnotesize,minimum height=4mm]
\tikzstyle{asym hadamard dag}=[fill=white,draw,shape=SEbox,inner sep=0.6mm,font=\footnotesize,minimum height=4mm]

\tikzstyle{hadamard}=[fill=white,draw,inner sep=0.6mm,font=\footnotesize,minimum height=4mm,minimum width=4mm]
\tikzstyle{small hadamard}=[fill=white,draw,inner sep=0.6mm,minimum height=1.5mm,minimum width=1.5mm]
\tikzstyle{small hadamard rotate}=[small hadamard,rotate=45]
\tikzstyle{dhadamard}=[hadamard,doubled]
\tikzstyle{small dhadamard}=[small hadamard,doubled]
\tikzstyle{small dhadamard rotate}=[small hadamard rotate,doubled]
\tikzstyle{antipode}=[white dot,inner sep=0.3mm,font=\footnotesize]

\tikzstyle{scalar}=[diamond,draw,inner sep=0.5pt,font=\small]
\tikzstyle{dscalar}=[diamond,doubled, draw,inner sep=0.5pt,font=\small]

\tikzstyle{small box}=[rectangle,inline text,fill=white,draw,minimum height=5mm,yshift=-0.5mm,minimum width=5mm,font=\small]
\tikzstyle{small gray box}=[small box,fill=gray!30]
\tikzstyle{medium box}=[rectangle,inline text,fill=white,draw,minimum height=5mm,yshift=-0.5mm,minimum width=10mm,font=\small]
\tikzstyle{square box}=[small box] 
\tikzstyle{medium gray box}=[small box,fill=gray!30]
\tikzstyle{semilarge box}=[rectangle,inline text,fill=white,draw,minimum height=5mm,yshift=-0.5mm,minimum width=12.5mm,font=\small]
\tikzstyle{large box}=[rectangle,inline text,fill=white,draw,minimum height=5mm,yshift=-0.5mm,minimum width=15mm,font=\small]
\tikzstyle{large gray box}=[small box,fill=gray!30]

\tikzstyle{Bayes box}=[rectangle,fill=black,draw, minimum height=3mm, minimum width=3mm]

\tikzstyle{gray square point}=[small box,fill=gray!50]

\tikzstyle{dphase box white}=[dhadamard]
\tikzstyle{dphase box gray}=[dhadamard,fill=gray!50!white]
\tikzstyle{phase box white}=[hadamard]
\tikzstyle{phase box gray}=[hadamard,fill=gray!50!white]

\tikzstyle{point}=[regular polygon,regular polygon sides=3,draw,scale=0.75,inner sep=-0.5pt,minimum width=9mm,fill=white,regular polygon rotate=180]
\tikzstyle{point nosep}=[regular polygon,regular polygon sides=3,draw,scale=0.75,inner sep=-2pt,minimum width=9mm,fill=white,regular polygon rotate=180]
\tikzstyle{copoint}=[regular polygon,regular polygon sides=3,draw,scale=0.75,inner sep=-0.5pt,minimum width=9mm,fill=white]
\tikzstyle{dpoint}=[point,doubled]
\tikzstyle{dcopoint}=[copoint,doubled]

\tikzstyle{pointgrow}=[shape=cornerpoint,kpoint common,scale=0.75,inner sep=3pt]
\tikzstyle{pointgrow dag}=[shape=cornercopoint,kpoint common,scale=0.75,inner sep=3pt]

\tikzstyle{wide copoint}=[fill=white,draw,shape=isosceles triangle,shape border rotate=90,isosceles triangle stretches=true,inner sep=0pt,minimum width=1.5cm,minimum height=6.12mm]
\tikzstyle{wide point}=[fill=white,draw,shape=isosceles triangle,shape border rotate=-90,isosceles triangle stretches=true,inner sep=0pt,minimum width=1.5cm,minimum height=6.12mm,yshift=-0.0mm]
\tikzstyle{wide point plus}=[fill=white,draw,shape=isosceles triangle,shape border rotate=-90,isosceles triangle stretches=true,inner sep=0pt,minimum width=1.74cm,minimum height=7mm,yshift=-0.0mm]

\tikzstyle{wide dpoint}=[fill=white,doubled,draw,shape=isosceles triangle,shape border rotate=-90,isosceles triangle stretches=true,inner sep=0pt,minimum width=1.5cm,minimum height=6.12mm,yshift=-0.0mm]

\tikzstyle{tinypoint}=[regular polygon,regular polygon sides=3,draw,scale=0.55,inner sep=-0.15pt,minimum width=6mm,fill=white,regular polygon rotate=180]

\tikzstyle{white point}=[point]
\tikzstyle{white dpoint}=[dpoint]
\tikzstyle{green point}=[white point] 
\tikzstyle{white copoint}=[copoint]
\tikzstyle{gray point}=[point,fill=gray!40!white]
\tikzstyle{gray dpoint}=[gray point,doubled]
\tikzstyle{red point}=[gray point] 
\tikzstyle{gray copoint}=[copoint,fill=gray!40!white]
\tikzstyle{gray dcopoint}=[gray copoint,doubled]

\tikzstyle{white point guide}=[regular polygon,regular polygon sides=3,font=\scriptsize,draw,scale=0.65,inner sep=-0.5pt,minimum width=9mm,fill=white,regular polygon rotate=180]

\tikzstyle{black point}=[point,fill=black,font=\color{white}]
\tikzstyle{black copoint}=[copoint,fill=black,font=\color{white}]

\tikzstyle{tiny gray point}=[tinypoint,fill=gray!40!white]

\tikzstyle{diredge}=[->]
\tikzstyle{ddiredge}=[<->]
\tikzstyle{rdiredge}=[<-]
\tikzstyle{thickdiredge}=[->, very thick]
\tikzstyle{pointer edge}=[->,very thick,gray]
\tikzstyle{pointer edge part}=[very thick,gray]
\tikzstyle{dashed edge}=[dashed]
\tikzstyle{thick dashed edge}=[very thick,dashed]
\tikzstyle{thick gray dashed edge}=[thick dashed edge,gray!40]
\tikzstyle{thick map edge}=[very thick,|->]

\makeatletter
\newcommand{\boxshape}[3]{%
\pgfdeclareshape{#1}{
\inheritsavedanchors[from=rectangle] 
\inheritanchorborder[from=rectangle]
\inheritanchor[from=rectangle]{center}
\inheritanchor[from=rectangle]{north}
\inheritanchor[from=rectangle]{south}
\inheritanchor[from=rectangle]{west}
\inheritanchor[from=rectangle]{east}
\backgroundpath{
\southwest \pgf@xa=\pgf@x \pgf@ya=\pgf@y
\northeast \pgf@xb=\pgf@x \pgf@yb=\pgf@y

\@tempdima=#2
\@tempdimb=#3

\pgfpathmoveto{\pgfpoint{\pgf@xa - 5pt + \@tempdima}{\pgf@ya}}
\pgfpathlineto{\pgfpoint{\pgf@xa - 5pt - \@tempdima}{\pgf@yb}}
\pgfpathlineto{\pgfpoint{\pgf@xb + 5pt + \@tempdimb}{\pgf@yb}}
\pgfpathlineto{\pgfpoint{\pgf@xb + 5pt - \@tempdimb}{\pgf@ya}}
\pgfpathlineto{\pgfpoint{\pgf@xa - 5pt + \@tempdima}{\pgf@ya}}
\pgfpathclose
}
}}

\boxshape{NEbox}{0pt}{3pt}
\boxshape{SEbox}{0pt}{-3pt}
\boxshape{NWbox}{5pt}{0pt}
\boxshape{SWbox}{-5pt}{0pt}
\boxshape{EBox}{-3pt}{3pt}
\boxshape{WBox}{3pt}{-3pt}
\makeatother

\tikzstyle{cloud}=[shape=cloud,draw,minimum width=1.5cm,minimum height=1.5cm]

\tikzstyle{map}=[draw,shape=NEbox,inner sep=1pt,minimum height=4mm,fill=white]
\tikzstyle{dashedmap}=[draw,dashed,shape=NEbox,inner sep=2pt,minimum height=6mm,fill=white]
\tikzstyle{mapdag}=[draw,shape=SEbox,inner sep=1pt,minimum height=4mm,fill=white]
\tikzstyle{mapadj}=[draw,shape=SEbox,inner sep=2pt,minimum height=6mm,fill=white]
\tikzstyle{maptrans}=[draw,shape=SWbox,inner sep=2pt,minimum height=6mm,fill=white]
\tikzstyle{mapconj}=[draw,shape=NWbox,inner sep=2pt,minimum height=6mm,fill=white]

\tikzstyle{medium map}=[draw,shape=NEbox,inner sep=2pt,minimum height=6mm,fill=white,minimum width=7mm]
\tikzstyle{medium map dag}=[draw,shape=SEbox,inner sep=2pt,minimum height=6mm,fill=white,minimum width=7mm]
\tikzstyle{medium map adj}=[draw,shape=SEbox,inner sep=2pt,minimum height=6mm,fill=white,minimum width=7mm]
\tikzstyle{medium map trans}=[draw,shape=SWbox,inner sep=2pt,minimum height=6mm,fill=white,minimum width=7mm]
\tikzstyle{medium map conj}=[draw,shape=NWbox,inner sep=2pt,minimum height=6mm,fill=white,minimum width=7mm]
\tikzstyle{semilarge map}=[draw,shape=NEbox,inner sep=2pt,minimum height=6mm,fill=white,minimum width=9.5mm]
\tikzstyle{semilarge map trans}=[draw,shape=SWbox,inner sep=2pt,minimum height=6mm,fill=white,minimum width=9.5mm]
\tikzstyle{semilarge map adj}=[draw,shape=SEbox,inner sep=2pt,minimum height=6mm,fill=white,minimum width=9.5mm]
\tikzstyle{semilarge map dag}=[draw,shape=SEbox,inner sep=2pt,minimum height=6mm,fill=white,minimum width=9.5mm]
\tikzstyle{semilarge map conj}=[draw,shape=NWbox,inner sep=2pt,minimum height=6mm,fill=white,minimum width=9.5mm]
\tikzstyle{large map}=[draw,shape=NEbox,inner sep=2pt,minimum height=6mm,fill=white,minimum width=12mm]
\tikzstyle{large map conj}=[draw,shape=NWbox,inner sep=2pt,minimum height=6mm,fill=white,minimum width=12mm]
\tikzstyle{very large map}=[draw,shape=NEbox,inner sep=2pt,minimum height=6mm,fill=white,minimum width=17mm]

\tikzstyle{medium dmap}=[draw,doubled,shape=NEbox,inner sep=2pt,minimum height=6mm,fill=white,minimum width=7mm]
\tikzstyle{medium dmap dag}=[draw,doubled,shape=SEbox,inner sep=2pt,minimum height=6mm,fill=white,minimum width=7mm]
\tikzstyle{medium dmap adj}=[draw,doubled,shape=SEbox,inner sep=2pt,minimum height=6mm,fill=white,minimum width=7mm]
\tikzstyle{medium dmap trans}=[draw,doubled,shape=SWbox,inner sep=2pt,minimum height=6mm,fill=white,minimum width=7mm]
\tikzstyle{medium dmap conj}=[draw,doubled,shape=NWbox,inner sep=2pt,minimum height=6mm,fill=white,minimum width=7mm]
\tikzstyle{semilarge dmap}=[draw,doubled,shape=NEbox,inner sep=2pt,minimum height=6mm,fill=white,minimum width=9.5mm]
\tikzstyle{semilarge dmap trans}=[draw,doubled,shape=SWbox,inner sep=2pt,minimum height=6mm,fill=white,minimum width=9.5mm]
\tikzstyle{semilarge dmap adj}=[draw,doubled,shape=SEbox,inner sep=2pt,minimum height=6mm,fill=white,minimum width=9.5mm]
\tikzstyle{semilarge dmap dag}=[draw,doubled,shape=SEbox,inner sep=2pt,minimum height=6mm,fill=white,minimum width=9.5mm]
\tikzstyle{semilarge dmap conj}=[draw,doubled,shape=NWbox,inner sep=2pt,minimum height=6mm,fill=white,minimum width=9.5mm]
\tikzstyle{large dmap}=[draw,doubled,shape=NEbox,inner sep=2pt,minimum height=6mm,fill=white,minimum width=12mm]
\tikzstyle{large dmap conj}=[draw,doubled,shape=NWbox,inner sep=2pt,minimum height=6mm,fill=white,minimum width=12mm]
\tikzstyle{large dmap trans}=[draw,doubled,shape=SWbox,inner sep=2pt,minimum height=6mm,fill=white,minimum width=12mm]
\tikzstyle{large dmap adj}=[draw,doubled,shape=SEbox,inner sep=2pt,minimum height=6mm,fill=white,minimum width=12mm]
\tikzstyle{large dmap dag}=[draw,doubled,shape=SEbox,inner sep=2pt,minimum height=6mm,fill=white,minimum width=12mm]
\tikzstyle{very large dmap}=[draw,doubled,shape=NEbox,inner sep=2pt,minimum height=6mm,fill=white,minimum width=19.5mm]

\tikzstyle{muxbox}=[draw,shape=rectangle,minimum height=3mm,minimum width=3mm,fill=white]
\tikzstyle{dmuxbox}=[muxbox,doubled]

\tikzstyle{box}=[draw,shape=rectangle,inner sep=2pt,minimum height=6mm,minimum width=6mm,fill=white]
\tikzstyle{dbox}=[draw,doubled,shape=rectangle,inner sep=2pt,minimum height=6mm,minimum width=6mm,fill=white]
\tikzstyle{dmap}=[draw,doubled,shape=NEbox,inner sep=2pt,minimum height=6mm,fill=white]
\tikzstyle{dmapdag}=[draw,doubled,shape=SEbox,inner sep=2pt,minimum height=6mm,fill=white]
\tikzstyle{dmapadj}=[draw,doubled,shape=SEbox,inner sep=2pt,minimum height=6mm,fill=white]
\tikzstyle{dmaptrans}=[draw,doubled,shape=SWbox,inner sep=2pt,minimum height=6mm,fill=white]
\tikzstyle{dmapconj}=[draw,doubled,shape=NWbox,inner sep=2pt,minimum height=6mm,fill=white]

\tikzstyle{ddmap}=[draw,doubled,dashed,shape=NEbox,inner sep=2pt,minimum height=6mm,fill=white]
\tikzstyle{ddmapdag}=[draw,doubled,dashed,shape=SEbox,inner sep=2pt,minimum height=6mm,fill=white]
\tikzstyle{ddmapadj}=[draw,doubled,dashed,shape=SEbox,inner sep=2pt,minimum height=6mm,fill=white]
\tikzstyle{ddmaptrans}=[draw,doubled,dashed,shape=SWbox,inner sep=2pt,minimum height=6mm,fill=white]
\tikzstyle{ddmapconj}=[draw,doubled,dashed,shape=NWbox,inner sep=2pt,minimum height=6mm,fill=white]

\boxshape{sNEbox}{0pt}{3pt}
\boxshape{sSEbox}{0pt}{-3pt}
\boxshape{sNWbox}{3pt}{0pt}
\boxshape{sSWbox}{-3pt}{0pt}
\tikzstyle{smap}=[draw,shape=sNEbox,fill=white]
\tikzstyle{smapdag}=[draw,shape=sSEbox,fill=white]
\tikzstyle{smapadj}=[draw,shape=sSEbox,fill=white]
\tikzstyle{smaptrans}=[draw,shape=sSWbox,fill=white]
\tikzstyle{smapconj}=[draw,shape=sNWbox,fill=white]

\tikzstyle{dsmap}=[draw,dashed,shape=sNEbox,fill=white]
\tikzstyle{dsmapdag}=[draw,dashed,shape=sSEbox,fill=white]
\tikzstyle{dsmaptrans}=[draw,dashed,shape=sSWbox,fill=white]
\tikzstyle{dsmapconj}=[draw,dashed,shape=sNWbox,fill=white]

\boxshape{mNEbox}{0pt}{10pt}
\boxshape{mSEbox}{0pt}{-10pt}
\boxshape{mNWbox}{10pt}{0pt}
\boxshape{mSWbox}{-10pt}{0pt}
\tikzstyle{mmap}=[draw,shape=mNEbox]
\tikzstyle{mmapdag}=[draw,shape=mSEbox]
\tikzstyle{mmaptrans}=[draw,shape=mSWbox]
\tikzstyle{mmapconj}=[draw,shape=mNWbox]

\tikzstyle{mmapgray}=[draw,fill=gray!40!white,shape=mNEbox]
\tikzstyle{smapgray}=[draw,fill=gray!40!white,shape=sNEbox]

\makeatletter

\pgfdeclareshape{cornerpoint}{
\inheritsavedanchors[from=rectangle] 
\inheritanchorborder[from=rectangle]
\inheritanchor[from=rectangle]{center}
\inheritanchor[from=rectangle]{north}
\inheritanchor[from=rectangle]{south}
\inheritanchor[from=rectangle]{west}
\inheritanchor[from=rectangle]{east}
\backgroundpath{
\southwest \pgf@xa=\pgf@x \pgf@ya=\pgf@y
\northeast \pgf@xb=\pgf@x \pgf@yb=\pgf@y

\pgfmathsetmacro{\pgf@shorten@left}{\pgfkeysvalueof{/tikz/shorten left}}
\pgfmathsetmacro{\pgf@shorten@right}{\pgfkeysvalueof{/tikz/shorten right}}

\pgfpathmoveto{\pgfpoint{0.5 * (\pgf@xa + \pgf@xb)}{\pgf@ya - 5pt}}
\pgfpathlineto{\pgfpoint{\pgf@xa - 8pt + \pgf@shorten@left}{\pgf@yb - 1.5 * \pgf@shorten@left}}
\pgfpathlineto{\pgfpoint{\pgf@xa - 8pt + \pgf@shorten@left}{\pgf@yb}}
\pgfpathlineto{\pgfpoint{\pgf@xb + 8pt - \pgf@shorten@right}{\pgf@yb}}
\pgfpathlineto{\pgfpoint{\pgf@xb + 8pt - \pgf@shorten@right}{\pgf@yb - 1.5 * \pgf@shorten@right}}
\pgfpathclose
}
}

\pgfdeclareshape{cornercopoint}{
\inheritsavedanchors[from=rectangle] 
\inheritanchorborder[from=rectangle]
\inheritanchor[from=rectangle]{center}
\inheritanchor[from=rectangle]{north}
\inheritanchor[from=rectangle]{south}
\inheritanchor[from=rectangle]{west}
\inheritanchor[from=rectangle]{east}
\backgroundpath{
\southwest \pgf@xa=\pgf@x \pgf@ya=\pgf@y
\northeast \pgf@xb=\pgf@x \pgf@yb=\pgf@y

\pgfmathsetmacro{\pgf@shorten@left}{\pgfkeysvalueof{/tikz/shorten left}}
\pgfmathsetmacro{\pgf@shorten@right}{\pgfkeysvalueof{/tikz/shorten right}}

\pgfpathmoveto{\pgfpoint{0.5 * (\pgf@xa + \pgf@xb)}{\pgf@yb + 5pt}}
\pgfpathlineto{\pgfpoint{\pgf@xa - 8pt + \pgf@shorten@left}{\pgf@ya + 1.5 * \pgf@shorten@left}}
\pgfpathlineto{\pgfpoint{\pgf@xa - 8pt + \pgf@shorten@left}{\pgf@ya}}
\pgfpathlineto{\pgfpoint{\pgf@xb + 8pt - \pgf@shorten@right}{\pgf@ya}}
\pgfpathlineto{\pgfpoint{\pgf@xb + 8pt - \pgf@shorten@right}{\pgf@ya + 1.5 * \pgf@shorten@right}}
\pgfpathclose
}
}

\makeatother

\pgfkeyssetvalue{/tikz/shorten left}{0pt}
\pgfkeyssetvalue{/tikz/shorten right}{0pt}

\tikzstyle{kpoint common}=[draw,fill=white,inner sep=1pt,minimum height=4mm]
\tikzstyle{kpoint sc}=[shape=cornerpoint,kpoint common]
\tikzstyle{kpoint adjoint sc}=[shape=cornercopoint,kpoint common]
\tikzstyle{kpoint}=[shape=cornerpoint,shorten left=5pt,kpoint common]
\tikzstyle{kpoint adjoint}=[shape=cornercopoint,shorten left=5pt,kpoint common]
\tikzstyle{kpoint conjugate}=[shape=cornerpoint,shorten right=5pt,kpoint common]
\tikzstyle{kpoint transpose}=[shape=cornercopoint,shorten right=5pt,kpoint common]
\tikzstyle{kpoint symm}=[shape=cornerpoint,shorten left=5pt,shorten right=5pt,kpoint common]

\tikzstyle{wide kpoint sc}=[shape=cornerpoint,kpoint common, minimum width=1 cm]
\tikzstyle{wide kpointdag sc}=[shape=cornercopoint,kpoint common, minimum width=1 cm]

\tikzstyle{black kpoint}=[shape=cornerpoint,shorten left=5pt,kpoint common,fill=black,font=\color{white}]

\tikzstyle{black kpoint sm}=[shape=cornerpoint,shorten left=5pt,kpoint common,fill=black,font=\color{white},scale=0.75]

\tikzstyle{black kpoint adjoint}=[shape=cornercopoint,shorten left=5pt,kpoint common,fill=black,font=\color{white}]
\tikzstyle{black kpointadj}=[shape=cornercopoint,shorten left=5pt,kpoint common,fill=black,font=\color{white}]

\tikzstyle{black kpointadj sm}=[shape=cornercopoint,shorten left=5pt,kpoint common,fill=black,font=\color{white},scale=0.75]

\tikzstyle{black dkpoint}=[shape=cornerpoint,shorten left=5pt,kpoint common,fill=black, doubled,font=\color{white}]
\tikzstyle{black dkpoint adjoint}=[shape=cornercopoint,shorten left=5pt,kpoint common,fill=black, doubled,font=\color{white}]
\tikzstyle{black dkpointadj}=[shape=cornercopoint,shorten left=5pt,kpoint common,fill=black, doubled,font=\color{white}]

\tikzstyle{black dkpoint sm}=[shape=cornerpoint,shorten left=5pt,kpoint common,fill=black, doubled,font=\color{white},scale=0.75]
\tikzstyle{black dkpointadj sm}=[shape=cornercopoint,shorten left=5pt,kpoint common,fill=black, doubled,font=\color{white},scale=0.75]

\tikzstyle{kpointdag}=[kpoint adjoint]
\tikzstyle{kpointadj}=[kpoint adjoint]
\tikzstyle{kpointconj}=[kpoint conjugate]
\tikzstyle{kpointtrans}=[kpoint transpose]

\tikzstyle{big kpoint}=[kpoint, minimum width=1.2 cm, minimum height=8mm, inner sep=4pt, text depth=3mm]

\tikzstyle{wide kpoint}=[kpoint, minimum width=1 cm, inner sep=2pt]
\tikzstyle{wide kpointdag}=[kpointdag, minimum width=1 cm, inner sep=2pt]
\tikzstyle{wide kpointconj}=[kpointconj, minimum width=1 cm, inner sep=2pt]
\tikzstyle{wide kpointtrans}=[kpointtrans, minimum width=1 cm, inner sep=2pt]

\tikzstyle{wider kpoint}=[kpoint, minimum width=1.25 cm, inner sep=2pt]
\tikzstyle{wider kpointdag}=[kpointdag, minimum width=1.25 cm, inner sep=2pt]
\tikzstyle{wider kpointconj}=[kpointconj, minimum width=1.25 cm, inner sep=2pt]
\tikzstyle{wider kpointtrans}=[kpointtrans, minimum width=1.25 cm, inner sep=2pt]

\tikzstyle{gray kpoint}=[kpoint,fill=gray!50!white]
\tikzstyle{gray kpointdag}=[kpointdag,fill=gray!50!white]
\tikzstyle{gray kpointadj}=[kpointadj,fill=gray!50!white]
\tikzstyle{gray kpointconj}=[kpointconj,fill=gray!50!white]
\tikzstyle{gray kpointtrans}=[kpointtrans,fill=gray!50!white]

\tikzstyle{gray dkpoint}=[kpoint,fill=gray!50!white,doubled]
\tikzstyle{gray dkpointdag}=[kpointdag,fill=gray!50!white,doubled]
\tikzstyle{gray dkpointadj}=[kpointadj,fill=gray!50!white,doubled]
\tikzstyle{gray dkpointconj}=[kpointconj,fill=gray!50!white,doubled]
\tikzstyle{gray dkpointtrans}=[kpointtrans,fill=gray!50!white,doubled]

\tikzstyle{white label}=[draw,fill=white,rectangle,inner sep=0.7 mm]
\tikzstyle{gray label}=[draw,fill=gray!50!white,rectangle,inner sep=0.7 mm]
\tikzstyle{black label}=[draw,fill=black,rectangle,inner sep=0.7 mm]

\tikzstyle{dkpoint}=[kpoint,doubled]
\tikzstyle{wide dkpoint}=[wide kpoint,doubled]
\tikzstyle{dkpointdag}=[kpoint adjoint,doubled]
\tikzstyle{wide dkpointdag}=[wide kpointdag,doubled]
\tikzstyle{dkcopoint}=[kpoint adjoint,doubled]
\tikzstyle{dkpointadj}=[kpoint adjoint,doubled]
\tikzstyle{dkpointconj}=[kpoint conjugate,doubled]
\tikzstyle{dkpointtrans}=[kpoint transpose,doubled]

\tikzstyle{kscalar}=[kpoint common, shape=EBox, inner xsep=-1pt, inner ysep=3pt,font=\small]
\tikzstyle{kscalarconj}=[kpoint common, shape=WBox, inner xsep=-1pt, inner ysep=3pt,font=\small]

\tikzstyle{spekpoint}=[kpoint sc,minimum height=5mm,inner sep=3pt]
\tikzstyle{spekcopoint}=[kpoint adjoint sc,minimum height=5mm,inner sep=3pt]

\tikzstyle{dspekpoint}=[spekpoint,doubled]
\tikzstyle{dspekcopoint}=[spekcopoint,doubled]

 \tikzstyle{upground}=[circuit ee IEC,thick,ground,rotate=90,scale=2.5]
 \tikzstyle{downground}=[circuit ee IEC,thick,ground,rotate=-90,scale=2.5]
 \tikzstyle{bigground}=[regular polygon,regular polygon sides=3,draw=gray,scale=0.50,inner sep=-0.5pt,minimum width=10mm,fill=gray]

\tikzstyle{arrs}=[-latex,font=\small,auto]
\tikzstyle{arrow plain}=[arrs]
\tikzstyle{arrow dashed}=[dashed,arrs]
\tikzstyle{arrow bold}=[very thick,arrs]
\tikzstyle{arrow hide}=[draw=white!0,-]
\tikzstyle{arrow reverse}=[latex-]
\tikzstyle{cdnode}=[]

\newtheorem{theo}{Theorem}
\newtheorem{thm}[theo]{Theorem}
 
\newtheorem{lem}[theo]{Lemma}

\newtheorem{defn}[theo]{Definition}

\newcommand{\cH}{\mathcal{H}}

\newcommand{\I}{\mathcal{I}}
\newcommand{\id}{\mathbb{I}}
\newcommand{\In}{\textbf{I}}
\newcommand{\tr}[2]{\mathrm{tr}_{#2} \left\{ #1 \right\}}

\newcommand{\X}{\mathbb{X}}
\newcommand{\Y}{\mathbb{Y}}
\newcommand{\A}{\mathbb{A}}
\newcommand{\D}{\mathbb{D}}
\newcommand{\W}{\mathbb{W}}
\newcommand{\C}{\mathbb{C}}
\newcommand{\Z}{\mathbb{Z}}
\newcommand{\U}{\mathbb{U}}
\newcommand{\As}{\boldsymbol{\Sigma}}

\newcommand{\B}{\mathbb{B}}
\newcommand{\N}{\textbf{N}}
\newcommand{\Ne}{\mathcal{N}}

\newcommand{\cc}{\textbf{c}}
\newcommand{\w}{\textbf{w}}
\newcommand{\z}{\textbf{z}}

\newcommand{\M}{\mathrm{M}}

\begin{document}

\title{Activation of post-quantumness in bipartite generalised EPR scenarios}

\author{Beata Zjawin}
\email{beata.zjawin@phdstud.edu.ug.pl}
\affiliation{International Centre for Theory of Quantum Technologies, University of  Gda{\'n}sk, 80-308 Gda{\'n}sk, Poland}

\author{Matty J. Hoban}
\affiliation{Quantum Group, Department of Computer Science, University of Oxford, United Kingdom}

\author{Paul Skrzypczyk}
\affiliation{H. H. Wills Physics Laboratory, University of Bristol, Tyndall Avenue, Bristol, BS8 1TL, UK}
\affiliation{CIFAR Azrieli Global Scholars program, CIFAR, Toronto, Canada}

\author{Ana Bel\'en Sainz}
\affiliation{International Centre for Theory of Quantum Technologies, University of  Gda{\'n}sk, 80-308 Gda{\'n}sk, Poland}

\date{\today}

\begin{abstract}
In a standard bipartite Einstein-Podolsky-Rosen (EPR) scenario, Alice and Bob share a system prepared in an entangled state and Alice performs local measurements. One possible generalisation of this set-up is to allow Bob to also locally process his subsystem. Then, correlations generated in such generalised EPR scenarios are examples of non-signalling bipartite resources, called assemblages, that can exhibit post-quantum behavior,  i.e., cannot be generated using solely quantum systems. There exist assemblages that, despite being post-quantum resources, can only generate quantum correlations in bipartite Bell-type scenarios. Here, we present a protocol for activation of post-quantumness in bipartite generalised EPR scenarios such as the so-called Bob-with-input, measurement-device-independent, and channel EPR scenarios. By designing a protocol that involves distributing the assemblages in a larger network, we derive tailored Bell inequalities which can be violated beyond their quantum bound in this new set-up. Our results show that in all of the above generalised scenarios, the post-quantumness of the assemblages can be witnessed at the level of correlations they produce. 
\end{abstract}

\maketitle

\tableofcontents

\section{Introduction}

In the bipartite Einstein-Podolsky-Rosen (EPR) scenario~\cite{einstein1935can,schrodinger1935discussion,cavalcanti2009experimental}, nonclassical correlations are generated between two distant parties, Alice and Bob, upon local measurements on half of a system prepared in an entangled state. Along with Bell-nonclassicality~\cite{Bell64,clauser1969proposed,brunner2014bell} and entanglement~\cite{LOCCentang}, it is one of the notable resources studied in quantum theory. Correlations generated in the EPR scenario are crucial for various information-theoretic tasks~\cite{wiseman2007steering,uola2020quantum}, such as quantum cryptography~\cite{branciard2012one,gianissecretsharing}, entanglement certification~\cite{cavalcanti2015detection,mattar2017experimental} and randomness certification~\cite{passaro2015optimal,law2014quantum}. From the quantum foundations perspective, the EPR scenario allows us to study the structure and limitations of quantum theory~\cite{reid2009colloquium,sainz2018formalism,cavalcanti2022post,EPRLOSR,rossi2022characterising}. 

Nonclassical resources generated in generalised EPR scenarios~\cite{sainz2020bipartite,cavalcanti2013entanglement,piani2015channel,schmid2020type,rosset2020type,channelEPR}, whose formal definitions we recall below, exhibit a richer structure than those in the standard EPR scenario. Notably, these resources do not require the shared system between Alice and Bob to be quantum; instead, it can be a post-quantum system from a more general theory. Therefore, examining generalised EPR scenarios allows us to study quantum theory from the outside, helping us explore its limitations. A significant example of the advantage of post-quantum resources over quantum ones is the task of remote state preparation, where post-quantum assemblages outperform quantum resources~\cite{cavalcanti2022post}. Additionally, post-quantum resources are relevant to the quantum causal models program~\cite{tucci1997quantum,costa2016quantum,chaves2015information,cavalcanti2014modifications,leifer2013towards,pienaar2020quantum,allen2017quantum}. When two parties share a system and perform local measurements on it, determining the nature of this system (whether it is classical, quantum or post-quantum) is a non-trivial task. Investigating post-quantum resources allows us to rule out the existence of a quantum common cause in such scenarios.

The standard EPR scenario~\cite{einstein1935can,schrodinger1935discussion,cavalcanti2009experimental} is illustrated in Fig.~\subref{fig:steering-scenario}, where the classical and quantum systems are depicted with single and double lines, respectively. Conventionally, Alice and Bob share a quantum system denoted by a density operator $\rho_{AB}$. Alice performs a local measurement on her subsystem (associated with a generalised measurement, i.e., a positive operator-valued measure (POVM) $\{M_{a|x}\}_{a,x}$): she chooses a measurement setting labeled by $x\in\X$ and registers a classical outcome $a\in\A$ with probability $p(a|x)$. Upon this measurement, the subsystem in Bob's laboratory can be described by a conditional state $\rho_{a|x}$. The relevant object of study in the EPR scenario is an \emph{assemblage} \cite{pusey2013negativity} defined as $\As_{\A|\X}= \{\sigma_{a|x}\}_{a,x}$, where each unnormalised state $\sigma_{a|x}=\tr{(M_{a|x}\otimes \id)\rho_{AB}}{A}$ is given by $\sigma_{a|x} := p(a|x) \rho_{a|x}$. In this paper, we do not limit Alice and Bob's shared resources to be quantum systems nor Alice to perform quantum measurements, but we do limit Bob to having a quantum system. That is, we take the space of all resources to be the non-signalling ones (meaning that the parties cannot utilise the shared system for instantaneous communication). Notably, in the standard bipartite EPR scenario, this does not make a difference. Due to the GHJW theorem, proven by Gisin \cite{gisin1989stochastic} and Hughston, Jozsa, and Wootters \cite{hughston1993complete}, all possible non-signalling assemblages can be reproduced by Alice making a measurement on a quantum system. However, this is not true for non-bipartite scenarios, as shown in Ref. \cite{sainz2015postquantum}; in this work we will not consider such multipartite settings.

Different generalisations of this
standard EPR scenario have been introduced in recent years, wherein Bob may also probe his system
in various ways. Then, it is possible to relax the condition of the shared physical system being quantum, and consider more general theories in the set-up of this experiment (\textit{post-quantum} systems). Distinction between these generalisations is made by different processing types on Bob's side, which are characterised by different input and output types in Fig.~\ref{fig:small-resources}. In every type of EPR scenario, the relevant assemblage is then given by a collection of ensembles of Bob's conditional processes labeled by Alice's measurement input and output variables. EPR scenarios considered in this paper are the following:
\begin{compactitem}
\item[($b$)] Bob-with-input EPR scenario~\cite{sainz2020bipartite}: a bipartite EPR scenario, where Bob has a classical input which allows him to locally influence the state preparation of his quantum output system.
\item[($c$)] Measurement-device-independent EPR scenario~\cite{cavalcanti2013entanglement,channelEPR}: a bipartite EPR scenario where Bob has a measurement channel with a quantum input and a classical output, i.e., a quantum instrument with a trivial output Hilbert space. 
\item[($d$)] Channel EPR scenario~\cite{piani2015channel,channelEPR}: a bipartite EPR scenario where Bob has quantum input and output systems.
\end{compactitem}
These scenarios are all closely related, and a unified framework for understanding them was introduced in Refs.~\cite{schmid2020type,rosset2020type,channelEPR}. In particular, all these scenarios have a \textit{common-cause} structure (Alice and Bob share a physical system) and Alice is measuring her local subsystem. Therefore, the correlations generated between Alice and Bob depend on the setting and output of Alice's measurement.

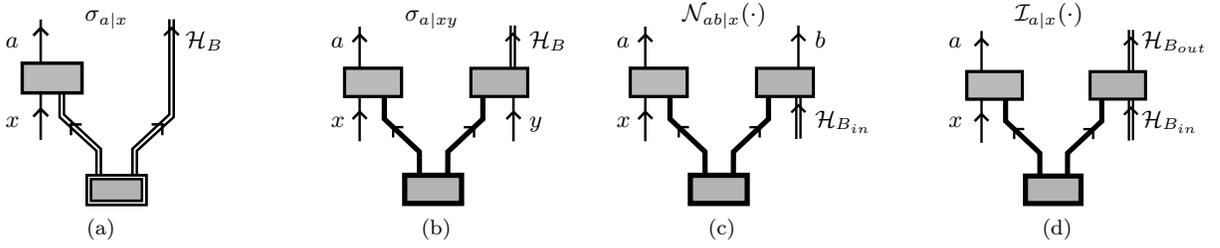
\begin{figure}[h!]

\subfloat[\label{fig:steering-scenario}]{%
  \includegraphics[width=0.23\textwidth]{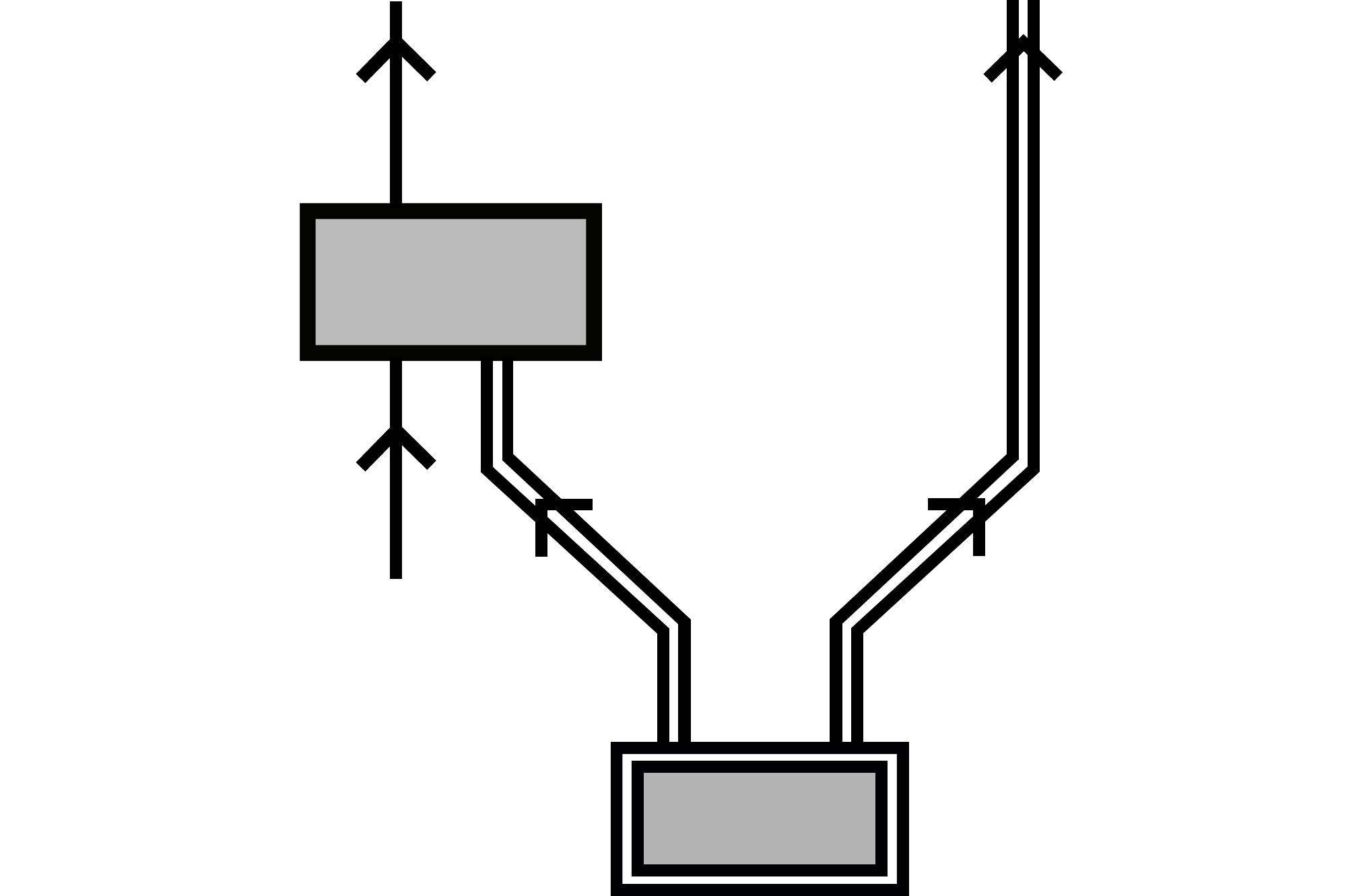}%
    \put(-90,30){$x$}
\put(-90,60){$a$}
\put(-60,70){$\sigma_{a|x}$ }
\put(-22,60){$\cH_{B}$}
}\hfill
\subfloat[\label{fig:BWI-scenario-small}]{%
  \includegraphics[width=0.15\textwidth]{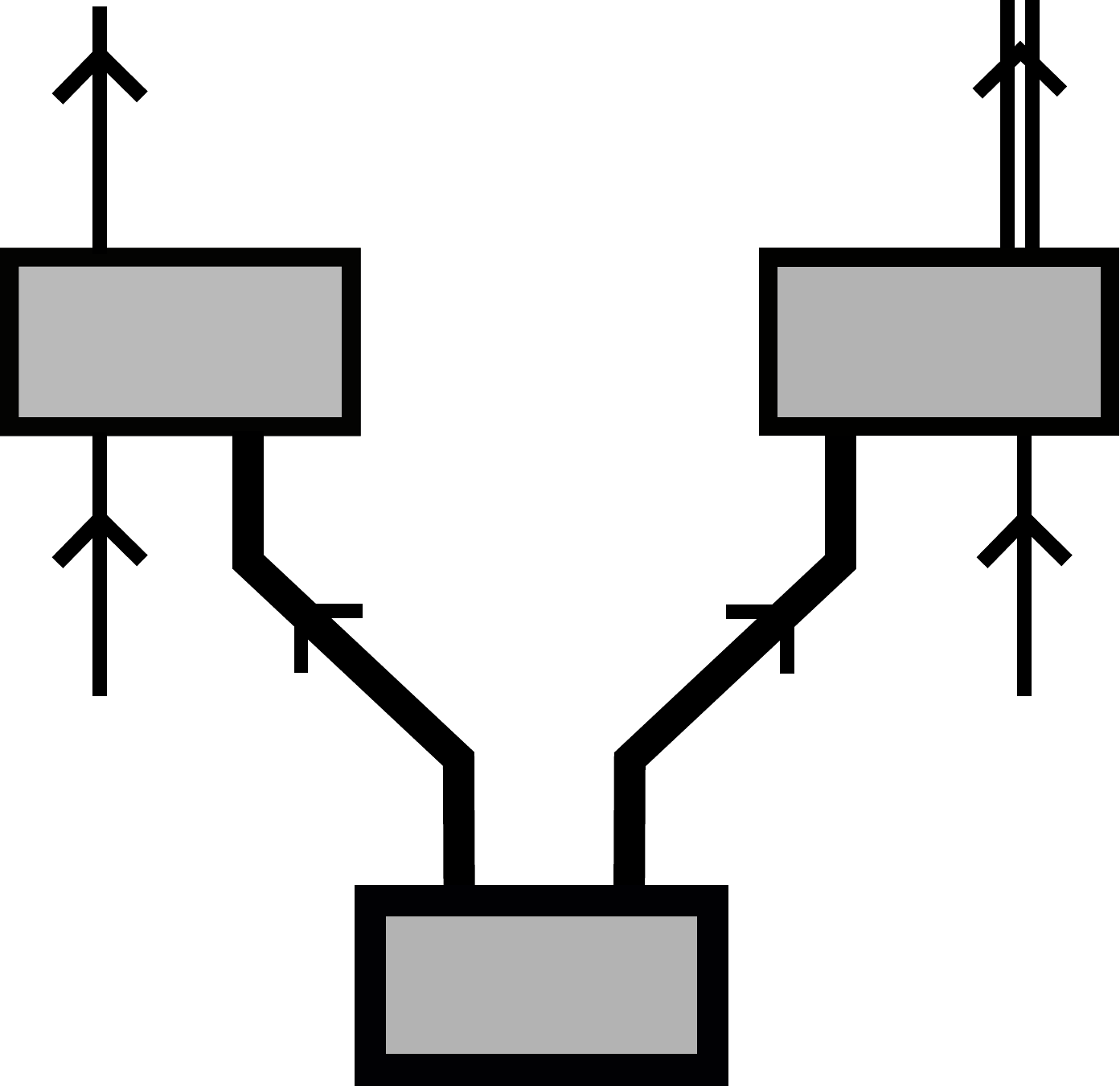}%
    \put(-75,30){$x$}
\put(-75,60){$a$}
\put(-47,70){$\sigma_{a|xy}$ }
\put(0,60){$\cH_{B}$}
\put(0,30){$y$}
  
}\hfill
\subfloat[\label{fig:MDI-scenario}]{%
  \includegraphics[width=0.15\textwidth]{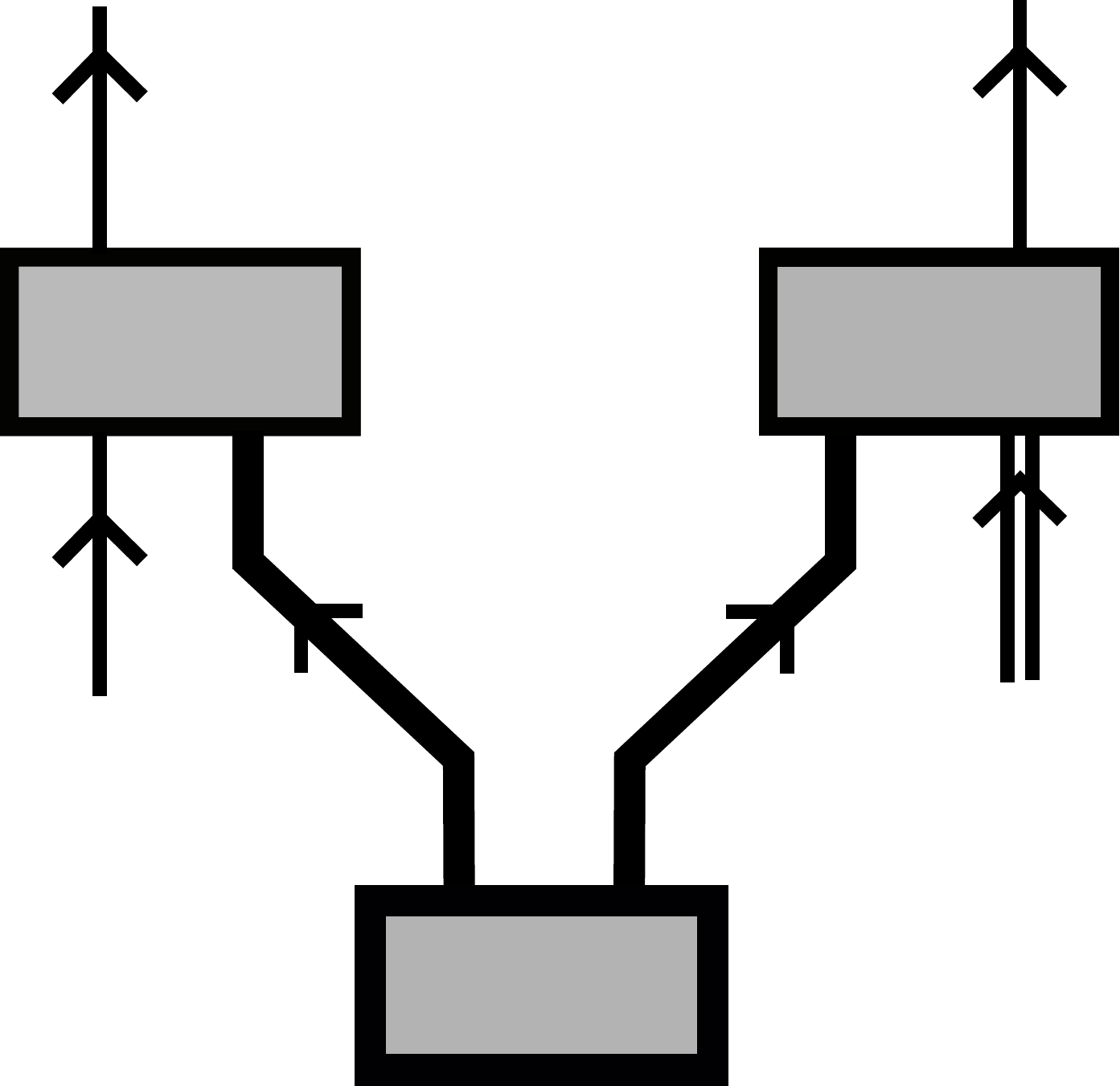}%
    \put(-75,30){$x$}
\put(-75,60){$a$}
\put(-50,70){$\Ne_{ab|x}(\cdot)$ }
\put(0,60){$b$}
\put(0,30){$\cH_{B_{in}}$}
}\hfill
\subfloat[\label{fig:channel-scenario}]{%
\includegraphics[width=0.23\textwidth]{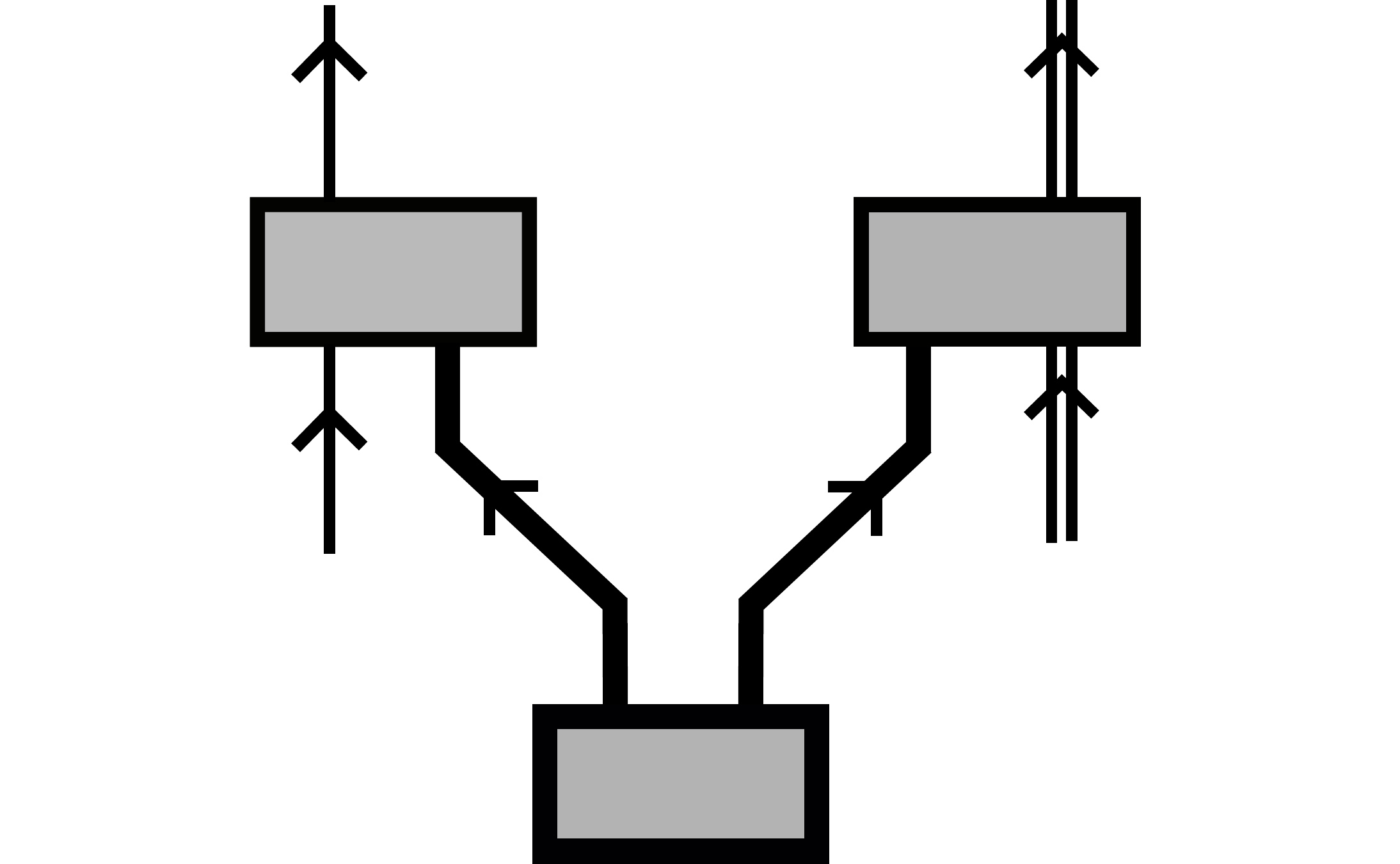}%
  \put(-95,30){$x$}
\put(-95,60){$a$}
\put(-70,70){$\I_{a|x}(\cdot)$ }
\put(-22,60){$\cH_{B_{out}}$}
\put(-22,30){$\cH_{B_{in}}$}

}

  \caption{\textbf{generalisations of the standard EPR scenario:} (a) Standard bipartite EPR scenario. (b)~Bob-with-input EPR scenario. (c) Measurement-device-independent EPR scenario. (d) Channel EPR scenario. Quantum and classical systems are depicted by double and single lines, respectively. Thick black lines depict the possibility that the shared systems may be classical, quantum, or even post-quantum. Formal definitions of the assemblage elements for each EPR scenario ($\sigma_{a|x}$ for the Bob-with-input scenario, $\Ne_{ab|x}(\cdot)$ for the measurement-device-independent scenario and $\I_{a|x}(\cdot)$ for the channel scenario) are introduced in the following sections. 
  }\label{fig:small-resources}
\end{figure}

In contrast to the standard EPR scenario, its generalisations can exhibit post-quantum resources. Moreover, for each of the generalisations introduced above, there exist assemblages such that their post-quantumness cannot be tested in a bipartite device-independent set-up\footnote{Here, the concept of a bipartite device-independent set-up is quite expansive. Essentially, device-independent set-ups are configurations involving devices with classical inputs and outputs whose characteristics are not specified. Further elaboration on this will be provided for each EPR scenario later on.}. For such assemblages, testing their post-quantumness requires either evaluating EPR functionals or using hierarchies of semi-definite programs~\cite{sainz2020bipartite, channelEPR,hoban2022hierarchy}. 

In this paper, we design a protocol for activating post-quantumness of bipartite assemblages in device-independent multipartite settings. We show how to distribute and process a post-quantum assemblage such that a device-independent test in this new setting can certify post-quantumness of the assemblage. The protocol has a similar structure for all types of generalised EPR scenarios. First, an additional quantum resource (a standard bipartite EPR assemblage) shared between Bob and an additional third party is self-tested. In this step of the protocol, Bob measures his share of this additional resource to learn about its structure in a device-independent way. Second, this additional quantum resource is used to promote EPR functionals to Bell functionals on correlations generated in this new larger network (with Alice, Bob, and the additional third party). Then, these Bell functionals can be used to certify post-quantumness of the correlations generated in the protocol set-up.

Given that the protocols slightly differ for various types of generalised EPR scenarios, below we dedicate a section to each scenario to introduce and analyze its protocol. In the main body of the paper, we restrict our considerations to qubit assemblages; generalisation to higher-dimensional assemblages is discussed in the Appendix~\ref{app:higher}. The method that we use to show activation of post-quantumness is analogous to the one introduced in Ref.~\cite{activation} for multipartite assemblages, which adapts the results of Refs.~\cite{bowles2018device,bowles2018self} for the self-testing stage and the general idea for the protocol. The most significant difference between our set-ups is that we focus on bipartite scenarios in which now Bob can apply a local processing on the system shared with Alice. 

\section{Bob-with-input scenario}
\subsection{Description of the scenario}

In the Bob-with-input (BwI) EPR scenario, Bob can choose the value of a classical input $y \in \Y$ which influences the local transformation of his subsystem (not necessarily a quantum one) into a new quantum system. Alice's local operations are the same as in the standard EPR scenario. Let $p(a|x)$ denote the probability that Alice obtains outcome $a$ when performing measurement $x$. Then, if in addition Bob chooses an input $y$, the normalised state in his laboratory is given by $\rho_{a|xy}$. The relevant object of study in this scenario is a \textit{Bob-with-input assemblage} which is given by $\As_{\A|\X\Y}=\{\sigma_{a|xy}\}_{a,x,y}$, with $\sigma_{a|xy}=p(a|x)\rho_{a|xy}$. The most general constraints on the possible unnormalised states $\sigma_{a|xy}$ are the following.
\begin{defn}[Non-signalling Bob-with-input assemblages]
An assemblage $\As_{\A|\X\Y}$ with elements $\{\sigma_{a|xy}\}_{a,x,y}$ is a valid non-signalling assemblage in a Bob-with-input EPR scenario if the following constraints are satisfied
\begin{align}
& \sigma_{a|xy}  \geq 0 \quad \forall \, a, x, y\,, \\
& \tr{\sum_a \sigma_{a|xy}}{}=1 \quad \forall \, x, y\,, \\
 & \tr{\sigma_{a|xy}}{}=p(a|x) \quad \forall \, a, x, y\,, \label{eq:bwins1}\\
 & \sum_a \sigma_{a|xy}=\sum_a \sigma_{a|x'y} \quad \forall \, x, y, x'\label{eq:bwins2}.
\end{align}
\end{defn}
\noindent
Conditions \eqref{eq:bwins1} and \eqref{eq:bwins2} are imposed by relativistic causality. They simply imply that Alice and Bob cannot signal to each other.

It was shown in Ref.~\cite{sainz2020bipartite} that there exist non-signalling BwI assemblages that do not admit a quantum realisation of the form $\sigma_{a|xy}=\tr{(M_{a|x}\otimes \mathcal{E}_y)\rho_{AB}}{A}$, i.e., cannot be generated by Alice and Bob sharing a quantum state $\rho_{AB}$, Bob applying channels (completely positive and trace preserving (CPTP) maps~\cite{NielsenChuang,schmidinitial}) $\{\mathcal{E}_y\}_y$ on his subsystem and Alice measuring with POVMs $\{M_{a|x}\}_{a,x}$. We will refer to such assemblages as \textit{post-quantum} assemblages. 

To certify post-quantumness of a BwI assemblage, one can define an EPR inequality and show that its quantum bound can be violated by this assemblage. Take an arbitrary fixed post-quantum BwI assemblage $\As_{\A|\X\Y}^{*}$ with elements $\{\sigma_{a|xy}^{*}\}_{a,x,y}$. There always exist Hermitian operators $\{F_{axy}\}_{a,x,y}$ such that the EPR functional
\begin{align}\label{eq:general_EPR_functional_bwi}
    \tr{\sum_{a,x,y}F_{axy}\sigma_{a|xy}^{*}}{}<0, 
\end{align}
while $\tr{\sum_{a,x,y}F_{axy}\sigma_{a|xy}}{}\geq0$ when evaluated on any quantum assemblage $\{\sigma_{a|xy}\}_{a,x,y}$.

One particularly interesting set of assemblages in the BwI scenario considered in Ref.~\cite{sainz2015postquantum} consists of post-quantum assemblages that can only lead to quantum correlations if Bob decides to measure his quantum state. Formally, such assemblages with elements $\{\sigma_{a|xy}\}_{a,x,y}$ are such that when Bob performs a measurement $\{N_b\}_b$, the observed correlations $p(ab|xy)=\tr{N_b \sigma_{a|xy}}{}$ belong to the quantum set, i.e., admit a quantum realisation of the form $p(ab|xy)=\tr{N_{a|x} \otimes N_{b|y} \rho}{}$ with $\{N_{a|x}\}_{a,x}$ and $\{N_{b|y}\}_{b,y}$ representing POVMs and $\rho$ being a valid quantum state. This set of assemblages is interesting, because it shows that post-quantum assemblages are a fundamentally different resource than post-quantum correlations in Bell-type scenarios. We denote this set $\As^{QC}$.

In this paper, we propose a protocol in which a bipartite Bob-with-input post-quantum assemblage is distributed in a tripartite network such that post-quantum correlations can be generated in this new network. If the input assemblage belongs to the set $\As^{QC}$, we refer to this phenomenon as activation of post-quantumness, as it enables one to reveal the post-quantum nature of the BwI assemblage in a tripartite Bell-type setting. This result shows that all assemblages from the set $\As^{QC}$ can be utilised to generate post-quantum correlations. As there is no known example of a post-quantum assemblage which does not produce post-quantum correlations, it raises the question of whether there is a fundamental difference between these two resources. 

\subsection{The activation protocol}\label{sec:protocol-bwi}

The activation protocol is illustrated in Fig.~\ref{fig:protocol-bwi}. The experiment consists of three parties: Alice, Bob and Charlie. Alice and Bob share a bipartite Bob-with-input assemblage $\As_{\A|\X\Y}$ whose post-quantumness we want to certify. Moreover, Bob and Charlie share a standard quantum assemblage $\As_{\C|\W}$, with $\C\coloneqq\{0,1\}$ and $\W\coloneqq\{1,2,3\}$. Bob’s quantum subsystem arising in this bipartition is defined on $\cH_{B’}$; here we consider $\cH_{B’}$ of dimension 2 -- generalisation of this case is discussed in the Appendix~\ref{app:higher}. Finally, Bob has a measurement device $\{\M^{BB'}_{b|z}\}_{b,z}$ which measures the system defined on $\cH_{B} \otimes \cH_{B’}$. The input and output sets of this device are the following: $b \in \B \coloneqq \{0,1\}$ and $z \in \Z \coloneqq \{1,2,3,4,\star \}$. Hereon, we assume that the assemblages $\As_{\A|\X\Y}$ and $\As_{\C|\W}$ are independent. Additionally, Alice, Bob and Charlie must be space-like separated.

\begin{figure}[h!]
  \begin{center}
\includegraphics[width=0.4\textwidth]{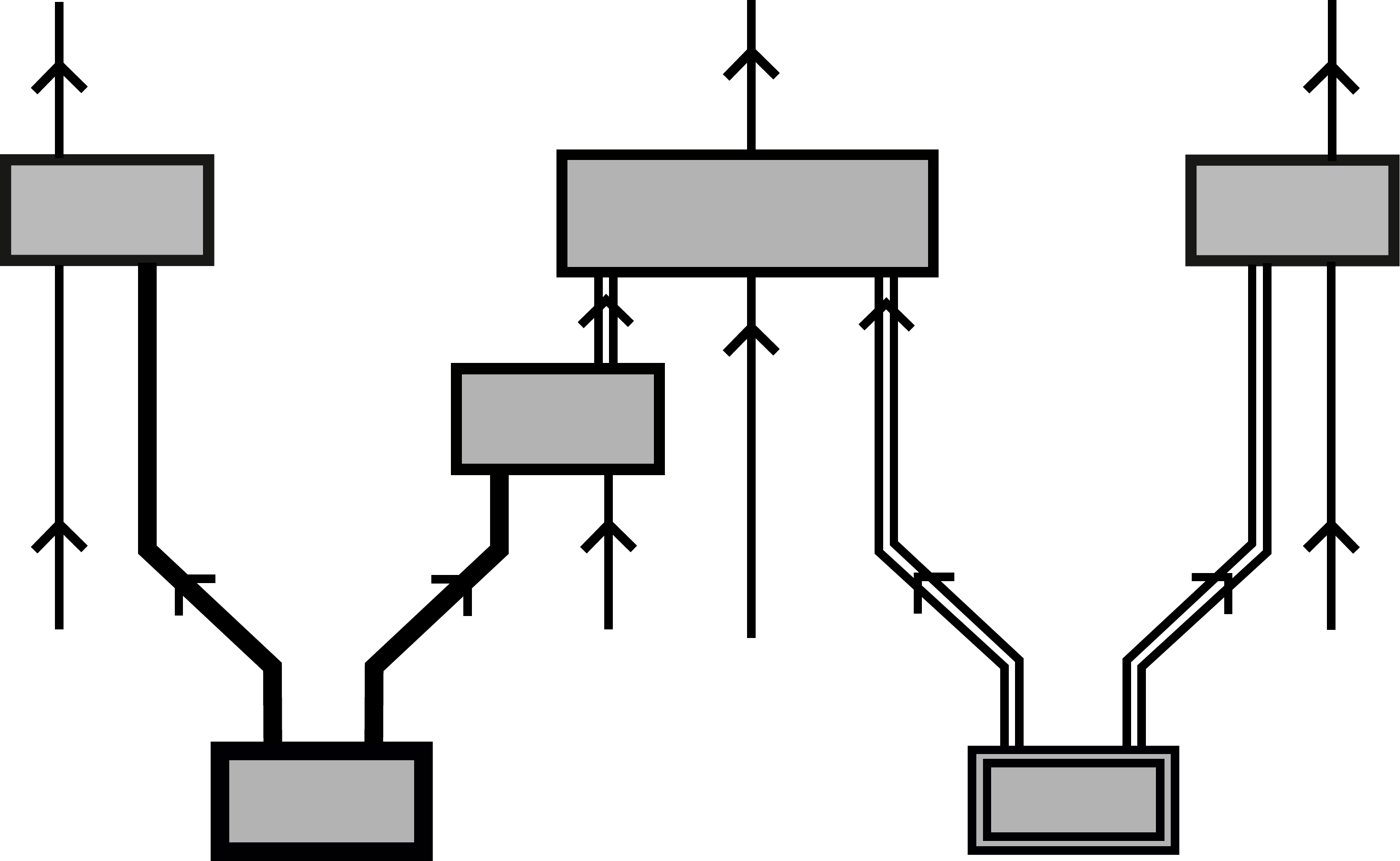}
\put(-200,55){$x$}
\put(-200,110){$a$}
\put(-116,35){$y$}
\put(-127,70){$\cH_{B}$}
\put(-100,110){$b$}
\put(-95,50){$z$}
\put(-60,70){$\cH_{B'}$}
\put(0,110){$c$}
\put(0,55){$w$}
  \end{center}
  \caption{Depiction of the activation protocol for the Bob-with-input scenario. Arbitrary systems (which may be post-quantum) are represented by thick lines, quantum systems are represented by double lines and classical systems are depicted as single lines}\label{fig:protocol-bwi}
\end{figure}
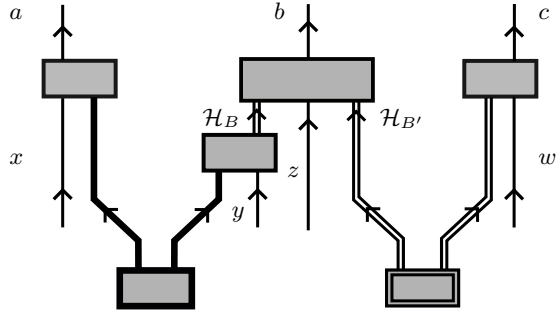
\noindent

If the BwI assemblage $\As_{\A|\X\Y}$ belongs to the set $\As^{QC}$, this protocol can activate its post-quantumness by utilizing the additional quantum resource $\As_{\C|\W}$ to produce post-quantum correlations in the new tripartite network. In the set-up described above, the protocol consists of the following two steps.
\begin{enumerate}[Step 1:]
    \item \textit{Self-testing of $\As_{\C|\W}$.}
\end{enumerate}
The first step of the protocol is to certify that the state of Bob's quantum system $\cH_{B’}$ is prepared in the assemblage $\As^{(r)}_{\C|\W}$ with elements $\sigma^{(r)}_{c|w}=r\widetilde{\sigma}_{c|w}+(1-r)(\widetilde{\sigma}_{c|w})^{T}$. Here, $0 \leq r \leq 1$ is an unknown parameter, and $\widetilde{\sigma}_{c|1}=(\id + (-1)^{c} Z)/4$, $\widetilde{\sigma}_{c|2}=(\id + (-1)^{c} X)/4$, $\widetilde{\sigma}_{c|3}=(\id - (-1)^{c} Y)/4$, with $X,Y,Z$ being the Pauli operators. Self-testing\footnote{Self-testing is a well-defined concept if the states and possible operations are specified within quantum theory. One may wonder if self-testing results are still relevant in our set-up, where we consider more general theories (post-quantum resources). In the specific case of self-testing of a standard bipartite EPR assemblage, post-quantum non-signalling resources are not more powerful than the quantum ones. Due to GHJW theorem~\cite{gisin1989stochastic,hughston1993complete}, all non-signalling assemblages in this standard bipartite scenario admit of a quantum realisation. Hence, using self-testing in this scenario is justified.} of this assemblage is possible due to the result proven in Refs.~\cite{bowles2018self,Chen2021robustselftestingof}. We recall the relevant Bell functional and the choice of measurements $\{M^{B'}_{b|z}\}_{b,z}$ for $z\in\{1,2,3,4\}$ which maximises it in Appendix~\ref{app:ST}.

\begin{enumerate}[Step 2:]
    \item \textit{Certification of post-quantumness.}
\end{enumerate}
The second step of the protocol is to measure the system on $\cH_{B} \otimes \cH_{B’}$ and evaluate an appropriate Bell functional to certify post-quantumness of the observed correlations\footnote{Here, even though the Bell functional is a multipartite one, the corresponding ``quantum bound" that we compute assumes a particular causal structure between the parties: the maximum quantum value is computed by letting the assemblage $\As_{\A|\X\Y}$ run over quantum assemblages and setting the assemblage $\As_{\C|\W}$ to be the one self-tested in the protocol. Notice that this is not what one would usually refer to as the ``quantum bound" of a multipartite Bell inequality, which assumes a common cause is shared among all parties. This is not a problem for our result, however, since the goal of the protocol is to certify that the assemblage $\As_{\A|\X\Y}$ is post-quantum. It is important to notice that $\As_{\C|\W}$ is a standard bipartite assemblage; hence, due to GHJW theorem, it always has a quantum realisation. Therefore, if the correlations observed in the protocol are post-quantum, it is clear that this post-quantumness arises from the post-quantum nature of $\As_{\A|\X\Y}$.}. For now, we only consider the case when Bob's measurement setting and outcome is fixed such that $M^{BB'}_{b=0|z=\star}=\ket{\phi^{+}}\bra{\phi^{+}}$, with $\ket{\phi^{+}}=(\ket{00}+\ket{11})/\sqrt{2}$. Moreover, for simplicity we assume that the elements of the assemblage $\As_{\C|\W}$ are given by $\{\widetilde{\sigma}_{c|w}\}$, not by the mixture $\{\sigma^{(r)}_{c|w}\}$. We denote such assemblage by $\widetilde{\As}_{\C|\W}$. In the Appendix~\ref{app:mixture} and~\ref{app:thm}, we show that these assumptions do not limit our result and the protocol for activation of post-quantumness works when they are relaxed. Under these assumptions, the observed correlations are given by
\begin{align}\label{eq:corr}
    p(a,b=0,c|x,y,z=\star,w)=\tr{M^{BB'}_{b=0|z=\star} (\sigma_{a|xy} \otimes \widetilde{\sigma}_{c|w})}{}.
\end{align}
Using the fact that $\tr{A_1 \otimes \id_2 \ket{\phi^{+}}\bra{\phi^{+}}}{1}=\frac{1}{2}A^{T}$ for any operator $A$, we can write Eq.~\eqref{eq:corr} as
\begin{align}\label{eq:corr2}
    p(a,b=0,c|x,y,z=\star,w)=\frac{1}{2}\tr{(\widetilde{\sigma}_{c|w})^{T} \sigma_{a|xy}}{}.
\end{align}
Given observed correlations, one way to certify their post-quantumness is to evaluate a suitable Bell functional, which we will now derive from the EPR functional specified in Eq.~\eqref{eq:general_EPR_functional_bwi}. For $w \in \{1,2,3\}$, let $\pi_{c|w}$ denote the projector onto the eigenspace of the eigenstates of Pauli $Z$, $X$ and $Y$ operators with eigenvalue $(-1)^{c}$. Notice that $\{ \pi_{c|w}\}_{c,w}$ form a basis of the set of Hermitian matrices in a two-dimensional Hilbert space. Therefore, the operators $\{F_{axy}\}_{a,x,y}$ can be decomposed as $F_{axy}=\sum_{c,w} \xi^{axy}_{cw}  \pi_{c|w}$ for some numbers $\{\xi^{axy}_{cw}\}_{a,x,y,c,w}$. We can use these coefficients to construct the following Bell functional on the correlation $\textbf{\textrm{p}}=\{p(a,b,c|x,y,z,w)\}$:
\begin{align}\label{eq:Bell-functional-bwi}
    I_{BwI}^{*}[\textbf{\textrm{p}}] \equiv \sum_{a,x,y,c,w} \xi^{axy}_{cw} p(a,b=0,c|x,y,z=\star,w).
\end{align}
Here, the coefficients corresponding to $(b,z)$ different than $(0,\star)$ are equal to zero. To calculate the value of this functional, we use the form of correlations given in Eq.~\eqref{eq:corr2}:
\begin{align}
    I_{BwI}^{*}[\textbf{\textrm{p}}] &= \sum_{a,x,y,c,w} \xi^{axy}_{cw} \frac{1}{2}\tr{(\widetilde{\sigma}_{c|w})^{T} \sigma_{a|xy}}{}, \\ \nonumber
&=\frac{1}{4} \sum_{a,x,y} \tr{\left(\sum_{c,w}\xi^{axy}_{cw}\pi_{c|w}\right) \sigma_{a|xy}}{}, \\ \nonumber
&=\frac{1}{4} \tr{\sum_{a,x,y}F_{axy}\sigma_{a|xy}}{}.
\end{align}
Here, we used the fact that the elements $(\widetilde{\sigma}_{c|w})^{T}$ are proportional to the projectors $\pi_{c|w}$. Therefore, $I_{BwI}^{*}[\textbf{\textrm{p}}] \geq 0$ for any correlation $\textbf{\textrm{p}}$ arising from a quantum assemblage $\{\sigma_{a|xy}\}_{a,x,y}$\footnote{Although in principle it is difficult to find operators $\{F_{axy}\}_{a,x,y}$ such that the quantum bound of the EPR functional is exactly $\beta^{Q}=0$, one can use numerical calculations to find a lower bound on $\beta^{Q}$. One possible approximation is an almost-quantum bound, which we denote $\beta^{AQ}$~\cite{sainz2020bipartite}. We show how to use this approximation for the activation protocol in the Appendix~\ref{app:example}, where we present how the protocol works for a specific fixed post-quantum assemblage.}. For the particular choice of measurement $M^{BB'}_{b=0|z=\star}=\ket{\phi^{+}}\bra{\phi^{+}}$, the Bell functional evaluated on the fixed assemblage $\As_{\A|\X\Y}^{*}=\{\sigma_{a|xy}^{*}\}_{a,x,y}$ gives $I_{BwI}^{*}[\textbf{\textrm{p*}}] < 0$. In the Appendix~\ref{app:thm}, we give a detailed analysis of the case when $M^{BB'}_{b=0|z=\star}\neq\ket{\phi^{+}}\bra{\phi^{+}}$ and $I_{BwI}^{*}$ needs to be optimised over all possible measurements. We show that the quantum bound is also given by 0 in this case; hence, if one observes $I_{BwI}^{*}[\textbf{\textrm{p}}] < 0$, it necessarily means that the correlations $\textbf{\textrm{p}}$ were generated from a post-quantum BwI assemblage.

If the first step of the protocol is successful, we can use the construction presented above to certify post-quantumness of the correlations in the tripartite set-up. For each BwI assemblage $\As^{*}_{\A|\X\Y}$, Theorem~\ref{thm-bwi} can be used to derive a tailored Bell functional to be used in the activation protocol. 
\begin{thm}\label{thm-bwi}
There always exists a Bell functional of the form
    \begin{align}
    I_{BwI}^{*}[\textbf{\textrm{p}}]=\frac{1}{2} \sum_{a,x,y,c,w} \xi^{axy}_{cw} p(a,b=0,c|x,y,z=\star,w),
    \end{align}
    for which $I_{BwI}^{*}[\textbf{\textrm{p}}] \geq 0$ when evaluated on any correlations $\textbf{\textrm{p}}=\{p(a,b,c|x,y,z,w)\}$ arising from quantum Bob-with-input assemblage with elements $\{\sigma_{a|xy}\}_{a,x,y}$ as $p(a,b=0,c|x,y,z=\star,w)=\tr{M^{BB'}_{b=0|z=\star} (\sigma_{a|xy} \otimes \widetilde{\sigma}_{c|w})}{}$.
\end{thm}
\noindent
As an explicit example of this method, in the Appendix~\ref{app:example}, we use the construction from Theorem~\ref{thm-bwi} to derive a Bell functional that activates post-quantumness of the assemblage $\As^{PTP}_{\A|\X\Y}$ introduced in Ref.~\cite[Eq.~(6)]{sainz2020bipartite}, which belongs to the set $\As^{QC}$.

The protocol designed in this section enables one to generate post-quantum Bell-type correlations from bipartite BwI assemblages. In particular, it can be used to activate post-quantumness of assemblages from the set $\As^{QC}$. Therefore, although the assemblages from the set $\As^{QC}$ do not exhibit bipartite Bell-type nonclassicality, they can generate post-quantum correlations in a larger network. This result suggests that the fundamental differences between the post-quantumness embedded in Bob-with-input EPR assemblages and in correlations generated in Bell scenarios are subtle, as the former can always generate latter.

\section{Measurement-device-independent scenario}
\subsection{Description of the scenario}

We will now consider a measurement-device-independent (MDI) EPR scenario in which Bob holds a collection of measurement channels that may be correlated with the physical system in Alice's lab~\cite{cavalcanti2013entanglement, channelEPR}. Bob's processing, which we denote by $\Theta_b^{BB_{in} \rightarrow B_{out}}$, takes as inputs a quantum system $\cH_{B_{in}}$ and the subsystem $B$. Bob's output system $B_{out}$ is just a classical variable that may take values within $\B$. Formally, the relevant \textit{MDI assemblages} in this scenario are given by $\N_{\A\B|\X}=~\{\Ne_{ab|x}(\cdot)\}_{a,b,x}$, where $\Ne_{ab|x}(\cdot)$ is a measurement channel -- a quantum instrument with a trivial output Hilbert space -- associated to the POVM element corresponding to outcome $b$ of Bob's measurement device (when Alice's measurement event is $a|x$). 

\begin{defn}[Non-signalling measurement-device-independent assemblages]
An assemblage $\N_{\A\B|\X}$ with elements $\{\Ne_{ab|x}(\cdot)\}_{a,b,x}$ is a valid non-signalling assemblage in a measurement-device-independent EPR scenario if the following constraints are satisfied
\begin{align} 
 & \sum_b \Ne_{ab|x}(\rho) = p(a|x) \quad \forall \, a, x, \rho, \label{eq:mdins1}\\
 & \sum_a \Ne_{ab|x}(\cdot) = \Omega_b^{B_{in} \rightarrow B_{out}}(\cdot) \quad \forall \, b,x \label{eq:mdins2},
\end{align}
where $\rho$ represents any normalised state of the quantum system $B_{in}$ and $\{\Omega_b^{B_{in} \rightarrow B_{out}}(\cdot)\}_{b}$ is a collection of measurement channels with $(\cdot)$ representing the input space $B_{in}$.
\end{defn}
\noindent
The no-signalling conditions \eqref{eq:mdins1} and \eqref{eq:mdins2} imply that Alice's output does not depend on Bob's input state and Alice's classical input cannot influence Bob's output. The elements of a quantumly-realisable MDI assemblage can be written as
\begin{equation}\label{eq:MDI-ass0}
\Ne_{ab|x}(\cdot)=\tr{(M_{a|x} \otimes \id_{B_{out}}) (\id_A \otimes \Theta_b^{BB_{in} \rightarrow B_{out}})[\rho_{AB} \otimes (\cdot)]}{AB_{out}}.
\end{equation}
Notably, there exist non-signalling MDI assemblages that do not admit of a decomposition of the form~\eqref{eq:MDI-ass0}~\cite{schmid2021postquantum,channelEPR}, which we refer to as \textit{post-quantum} MDI assemblages. 

One method for determining whether a measurement-device-independent EPR assemblage is post-quantum involves defining an EPR functional that it violates with respect to a quantum bound. Take an arbitrary fixed post-quantum MDI assemblage $\N_{\A\B|\X}^{*}$ with elements $\{\Ne^{*}_{ab|x}(\cdot)\}_{a,b,x}$. There always exist Hermitian operators $\{F_{abx}\}_{a,b,x}$ such that the EPR functional
\begin{align}\label{eq:general_EPR_functional_mdi}
    \tr{\sum_{a,b,x}F_{abx}\, J(\Ne^{*}_{ab|x})}{}<0, 
\end{align}
while $ \tr{\sum_{a,b,x}F_{abx} \, J(\Ne_{ab|x})}{}\geq0$ when evaluated on any quantum assemblage $\N_{\A\B|\X}$.

Moreover, there exist MDI assemblages which can only generate quantum correlations $p(ab|xy) = \Ne_{ab|x}(\rho_y)$ given a set of
fixed input states $\{\rho_y\}_y$. We denote the set of such assemblages by $\N^{QC}$. We give an example of such behaviour in the Appendix~\ref{app:mdiass}, where we show that a post-quantum assemblage introduced in Ref.~\cite{channelEPR} cannot generate post-quantum correlations in a bipartite Bell-type scenario. 

We now propose a protocol which certifies post-quantumness of MDI assemblages. Similarly to the previous section, we show that although existence of the set $\N^{QC}$ seems to imply that post-quantum assemblages are fundamentally different resources than post-quantum Bell nonclassicality, it is not the case. If the MDI assemblage in the protocol belongs to the set $\N^{QC}$, its post-quantumness can be activated, as it can generate post-quantum correlations in a tripartite Bell-like network.

\subsection{The activation protocol}

The experiment for the activation protocol consists of three parties – Alice, Bob and Charlie, as illustrated in Fig.~\ref{fig:protocol-mdi}. Alice and Bob share a bipartite measurement-device-independent assemblage $\N_{\A\B|\X}$. Bob's input state, which is on the Hilbert space $\cH_{B_{in}}$, is prepared by Charlie through a standard quantum assemblage $\As_{\C|\Z}$, with $\C\coloneqq\{0,1\}$ and $\W\coloneqq\{1,2,3\}$. That is, Bob and Charlie share a system $B_{in}C$ and Charlie performs measurements $\{M^{C}_{c|z}\}_{c,z}$, which result in assemblage elements $\{\sigma_{c|z}\}_{c,z}$. Finally, Bob has a measurement device $\{\M^{BB_{in}}_{b|y}\}_{b,y}$ with input and output sets $b \in \B$ and $y \in \Y \coloneqq \{1,2,3,4,\star \}$.

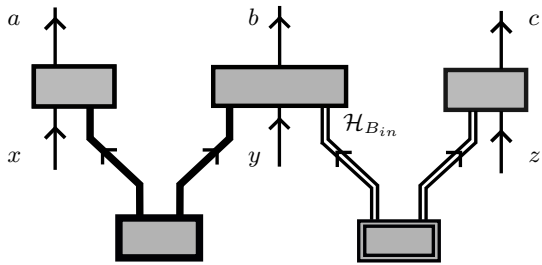
\begin{figure}[h!]
  \begin{center}
\includegraphics[width=0.4\textwidth]{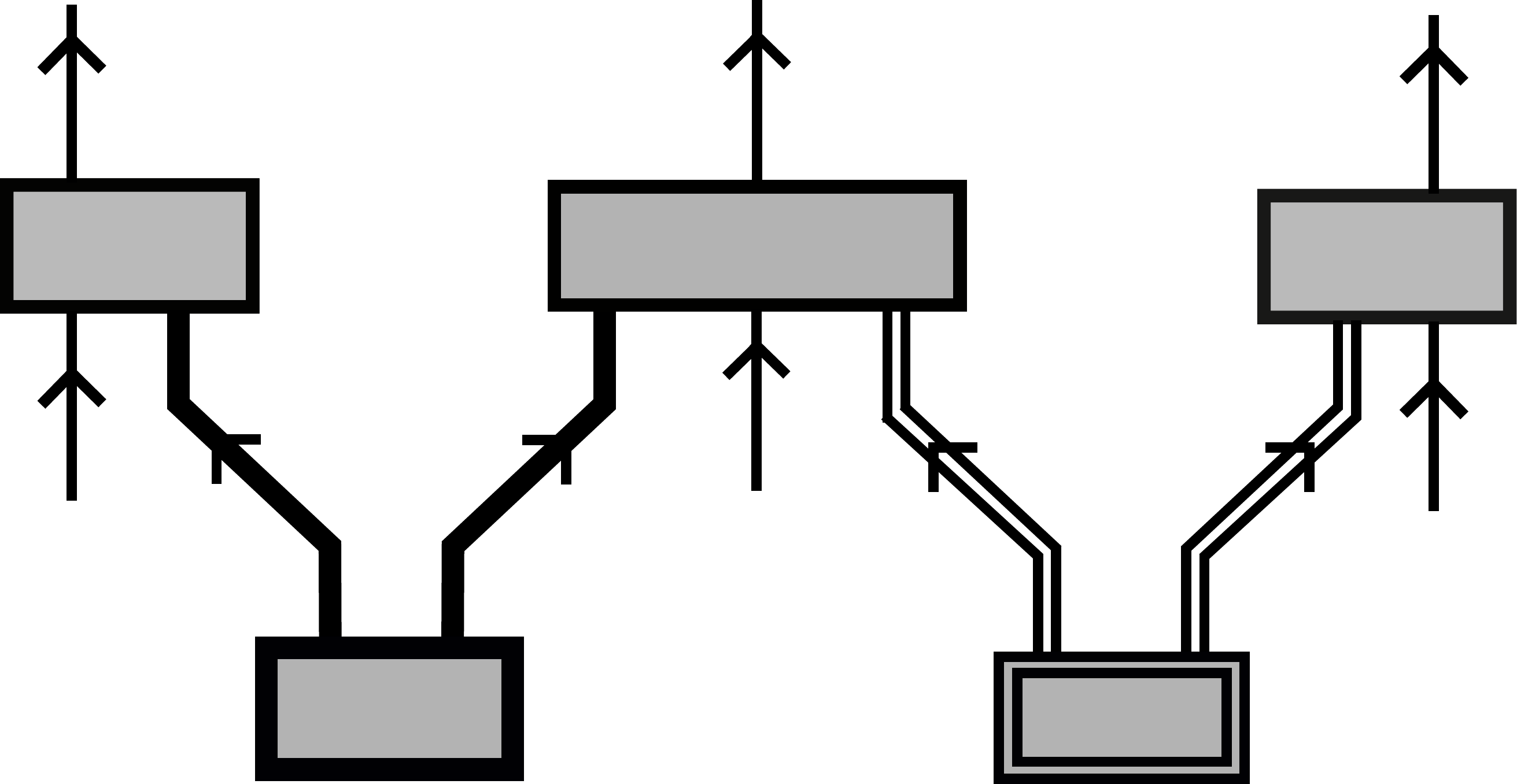}
\put(-197,38){$x$}
\put(-197,90){$a$}
\put(-106,90){$b$}
\put(-106,38){$y$}
\put(-70,50){$\cH_{B_{in}}$}
\put(0,90){$c$}
\put(0,38){$z$}
  \end{center}
  \caption{Depiction of the activation protocol for the measurement-device-independent scenario. Quantum systems are represented by double lines, while classical systems are depicted as single lines. Thick lines depict the possibility that the shared systems may be classical, quantum, or post-quantum.}\label{fig:protocol-mdi}
\end{figure}
\noindent
In the set-up of Fig.~\ref{fig:protocol-mdi}, $\N_{\A\B|\X}$ and $\As_{\C|\Z}$ are two independent assemblages. This can be guaranteed by preparing the systems $AB$ and
$B_{in}C$ separately. Additionally, Alice, Bob and Charlie must
be space-like separated.

The protocol for the MDI scenario has the same structure as the protocol for the BwI scenario introduced in the previous section. In the set-up described above, the protocol consists of the following two steps.
\begin{enumerate}[Step 1:]
    \item \textit{Self-testing of $\As_{\C|\Z}$.}
\end{enumerate}
For $y \in \{1,2,3,4\}$, Bob measures the quantum system $\cH_{B’}$ for the self-testing procedure that is equivalent to Step 1 in the Bob-with-input activation protocol. The goal of this step is to certify that Bob and Charlie share an assemblage $\As^{(r)}_{\C|\Z}$ with elements $\sigma^{(r)}_{c|z}=r\widetilde{\sigma}_{c|z}+(1-r)(\widetilde{\sigma}_{c|z})^{T}$. Details of the self-testing procedure are specified in the Appendix~\ref{app:ST}.

\begin{enumerate}[Step 2:]
    \item \textit{Certification of post-quantumness.}
\end{enumerate}
In the second step of the protocol, for $y=\star$, Bob applies the process $\Theta_b^{BB_{in} \rightarrow B_{out}}$ that defines the assemblage $\N_{\A\B|\X}$. Then, it is possible to certify the post-quantumness of the correlations $p(a,b,c|x,y=\star,z)$
via an appropriate Bell functional.

Assume that the first step of the protocol was successful. Here, we again assume that the elements of the assemblage Bob and Charlie share are $\{\widetilde{\sigma}_{c|z}\}_{c,z}$, not $\{\sigma^{(r)}_{c|z}\}_{c,z}$ (see the Appendix~\ref{app:mixture} for discussion about this assumption). Then, the observed correlations can be written as
\begin{align}\label{eq:corr-MDI-1}
p(a,b,c|x,y=\star,z)=\Ne_{ab|x}(\widetilde{\sigma}_{c|z}).
\end{align}
Now, let us use Choi-Jamiołkowski isomorphism to transform the channel $\Ne_{ab|x}(\cdot)$ to an operator. Define $J(\Ne_{ab|x})=(\Ne_{ab|x} \otimes \id^{C})(\ket{\phi^{+}}\bra{\phi^{+}}) _{B_{in}C}$. Moreover, notice that the elements $\widetilde{\sigma}_{c|z}$ can be written as $\widetilde{\sigma}_{c|z}=\tr{(\id^{B_{in}} \otimes \widetilde{M}^{C}_{c|z}) \ket{\phi^{+}}\bra{\phi^{+}} _{B_{in}C}}{C}$, with $\widetilde{M}^{C}_{c|1}=(\id + (-1)^{c} Z)/2$, $\widetilde{M}^{C}_{c|2}=(\id + (-1)^{c} X)/2$ and $\widetilde{M}^{C}_{c|3}=(\id + (-1)^{c} Y)/2$. Then, correlation~\eqref{eq:corr-MDI-1} can be written as
\begin{align}\label{eq:corr-MDI-2}
    p(a,b,c|x,y=\star,z)=\tr{\widetilde{M}^{C}_{c|z} \, J(\Ne_{ab|x})}{}.
\end{align}
To derive a suitable Bell functional for the considered set-up, we will use the EPR functional defined in Eq.~\eqref{eq:general_EPR_functional_mdi}. Let $\pi_{c|z}$ denote the projector onto the eigenspace of the eigenstates of Pauli $Z$,$X$ and $Y$ operators with eigenvalue $(-1)^{c}$. Then, the operators $\{F_{abx}\}_{a,b,x}$ can be decomposed as $F_{abx}=\sum_{c,z} \xi^{abx}_{cz}  \pi_{c|z}$ for some real numbers $\{\xi^{abx}_{cz}\}_{a,b,x,c,z}$. To certify post-quantum correlations in this protocol, it suffices to consider the following Bell functional on the correlation $\textbf{\textrm{p}}=\{p(a,b,c|x,y=\star,z)\}$:
\begin{align}\label{eq:Bell-functional-MDI}
    I_{MDI}^{*}[\textbf{\textrm{p}}] \equiv \sum_{a,b,x,c,z} \xi^{abx}_{cz} p(a,b,c|x,y=\star,z).
\end{align}
To calculate the value of this functional, we use the form of correlations given in Eq.~\eqref{eq:corr-MDI-2}:
\begin{align}
    I_{MDI}^{*}[\textbf{\textrm{p}}] &= \sum_{a,b,x,c,z} \xi^{abx}_{cz} \tr{ \widetilde{M}^{C}_{c|z} \, J(\Ne_{ab|x})}{}, \\ \nonumber
&=\sum_{a,b,x} \tr{\left(\sum_{c,z} \xi^{abx}_{cz}  \pi_{c|z}\right) \, J(\Ne_{ab|x})}{}, \\ \nonumber
&=\tr{\sum_{a,b,x}F_{abx}\, J(\Ne_{ab|x})}{}.
\end{align}
Here, we used the fact that the measurement elements $\widetilde{M}^{C}_{c|z}$ correspond to the projectors $\pi_{c|z}$. If $\textbf{\textrm{p}}$ was generated from a quantum MDI assemblage, then $I_{MDI}^{*}[\textbf{\textrm{p}}] \geq 0$. For the fixed post-quantum assemblage $\{\Ne^{*}_{ab|x}(\cdot)\}_{a,b,x}$, the functional evaluates to $I_{MDI}^{*}[\textbf{\textrm{p*}}] < 0$. This result is summarised in the theorem below.
\begin{thm}\label{thm-mdi}
There always exists a Bell functional of the form
    \begin{align}
    I_{MDI}^{*}[\textbf{\textrm{p}}] &= \sum_{a,b,x,c,z} \xi^{abx}_{cz} p(a,b,c|x,y=\star,z),
    \end{align}
    for which $I_{MDI}^{*}[\textbf{\textrm{p}}] \geq 0$ when evaluated on any correlations $\textbf{\textrm{p}}=\{p(a,b,c|x,y,z)\}$ arising from quantum measurement-device-independent assemblage with elements $\{\Ne_{ab|x}(\cdot)\}_{a,b,x}$ as $p(a,b,c|x,y=\star,z) = \Ne_{ab|x}(\widetilde{\sigma}_{c|z})$.
\end{thm}

To summarise, here we introduced a protocol in which a post-quantum MDI assemblage is processed to generate post-quantum Bell-type correlations. In analogy to the protocol for the BwI scenario, even assemblages from the set $\N^{QC}$ can be activated to generate post-quantum correlations when distributed in a larger network. It shows that there exists a relation between post-quantumness in measurement-device-independent EPR scenarios and device-independent post-quantum correlations.

\section{Channel EPR scenario}
\subsection{Description of the scenario}
In the channel EPR scenario~\cite{piani2015channel,channelEPR}, Bob has access to a channel with a quantum input defined on $\mathcal{H}_{B_{in}}$ and a quantum output defined on $\mathcal{H}_{B_{out}}$. Bob's local process, which we denote by $\Gamma^{BB_{in} \rightarrow B_{out}}$, might be correlated with Alice's system through the system $B$. Denote by $\I_{a|x}(\cdot)$ the instrument that is effectively applied on Bob's quantum system $B_{in}$ to output a quantum system $B_{out}$, given that Alice has performed a measurement $x$ on $A$ and obtained the outcome $a$. Then, the object of study in a channel EPR scenario is the \emph{channel assemblage} of instruments $\In_{\A|\X}= \{\I_{a|x}(\cdot)\}_{a,x}$.

\begin{defn}[Non-signalling channel assemblages]
An assemblage $\In_{\A|\X}$ with elements $\{\I_{a|x}(\cdot)\}_{a,x}$ is a valid non-signalling assemblage in a channel EPR scenario if the following constraints are satisfied
\begin{align}
 &\tr{\I_{a|x}(\rho)}{} = p(a|x) \quad \forall \, a, x, \rho, \label{eq:chns1}\\
 &\sum_{a} \I_{a|x}(\cdot) = \Lambda^{B_{in} \rightarrow B_{out}} (\cdot)\quad \forall \, x \label{eq:chns2},
\end{align}
where $\rho$ represents any normalised state of the quantum system $B_{in}$ and $\Lambda^{B_{in} \rightarrow B_{out}}(\cdot)$ is a quantum channel with $(\cdot)$ representing an input defined on $B_{in}$.
\end{defn}
\noindent
The no-signalling conditions \eqref{eq:chns1} and \eqref{eq:chns2} ensure that Alice's output does not depend on Bob's input state and Alice's classical input cannot influence Bob's output state.

When Alice and Bob share a bipartite quantum system prepared on a state $\rho_{AB}$ and Alice performs POVM measurements $\{M_{a|x}\}_{a,x}$, the elements of a channel assemblage admit quantum realisation of the form:
\begin{equation}\label{eq:channel-ass}
\I_{a|x}(\cdot)=\tr{(M_{a|x} \otimes \id_{B_{out}}) (\id_A \otimes \Gamma^{BB_{in} \rightarrow B_{out}})[\rho_{AB} \otimes (\cdot)_{B_{in}}]}{A}.
\end{equation}
For each measurement input $x$, the instruments $\{\I_{a|x}\}_{a,x}$ form a channel which does not depend on Alice's measurement choice, i.e., $\sum_{a\in\A} \I_{a|x}$ is a CPTP map. In general, however, there exists post-quantum non-signalling channel assemblages. 

To certify post-quantumness of an arbitrary but fixed post-quantum channel assemblage $\In_{\A|\X}^{*}$ with elements $\{\I_{a|x}^{*}\}$ one can use a channel EPR functional similar to the functional introduced for the BwI and MDI scenarios. The channel EPR functional is given by
\begin{align}\label{eq:general_EPR_functional_channel}
    \tr{\sum_{a,x}F_{ax}J(\I_{a|x}^{*})}{}<0, 
\end{align}
where$\tr{\sum_{a,x,y}F_{ax}J(\I_{a|x})}{}\geq0$ when evaluated on any quantum assemblage with elements $\{\I_{a|x}(\cdot)\}_{a,x}$. 

Moreover, there exist post-quantum assemblages $\In_{\A|\X}$ that only generate quantum correlations $\{p(ab|xy)\}$ when the input systems are fixed to be $\{\rho_y\}_y$ and Bob measures his output with a measurement $\{N_b\}_b$, i.e., $p(ab|xy) = \tr{N_b(\I_{a|x}(\rho_y))}{}$ has a quantum realisation. An example of such assemblage is given in Ref.~\cite[Eq.~(12)]{channelEPR}\footnote{A post-quantum channel assemblages introduced in Ref.~\cite[Eq.~(12)]{channelEPR} has a very similar structure to the MDI assemblage which we discuss in the Appendix~\ref{app:mdiass}; the proof that this assemblage can only generate quantum correlations $\{p(ab|xy)\}$ is analogous to the one presented there.}. We denote this set of assemblages $\In^{QC}$. 

Here we show how to embed a bipartite channel assemblage belonging to the set $\In^{QC}$ in a larger network such that post-quantum correlations are generated in this new set-up. The protocol combines the techniques introduced for the BwI and MDI scenarios in the previous sections. In particular, Bob's quantum input is generated in the same way as Bob's input in the protocol for the MDI scenario and Bob's quantum output is measured in the same manner as Bob's output in the protocol for the BwI scenario. 

\subsection{The activation protocol}

The experiment consists of four parties: Alice, Bob, Charlie and Dani, as illustrated in Fig~\ref{fig:protocol-channel}. Alice and Bob share a channel assemblage $\In_{\A|\X}$. Bob's input to his local channel is defined on the Hilbert space $\cH_{B_{in}}$. It is prepared by Charlie through a standard quantum assemblage $\As_{\C|\Z}$, with $\C\coloneqq\{0,1\}$ and $\W\coloneqq\{1,2,3\}$. Additionally, Bob has a classical input $y \in \Y \coloneqq \{1,2,3,4,\star,\lozenge,\blacklozenge \}$ which controls this stage of the protocol. Moreover, Bob and Dani share a standard quantum assemblage $\As_{\D|\U}$, with $\D\coloneqq\{0,1\}$ and $\U\coloneqq\{1,2,3\}$, where Bob's system is defined on $\cH_{B'}$. Finally, Bob has a measurement device $\{\M^{B_{out}B'}_{b|z}\}_{b,z}$ which measures the system defined on $\cH_{B_{out}} \otimes \cH_{B’}$. The input and output sets of this device are the following: $b \in \B \coloneqq \{0,1\}$ and $z \in \Z \coloneqq \{1,2,3,4,\star, \lozenge,\blacklozenge \}$. Throughout this section, we assume that the assemblages $\In_{\A|\X}$, $\As_{\C|\Z}$ and $\As_{\D|\U}$ are independent, i.e., the systems $AB$, $B_{in}C$ and $B'D$ are prepared separately and the four parties are space-like separated. Then, the activation protocol is the following.

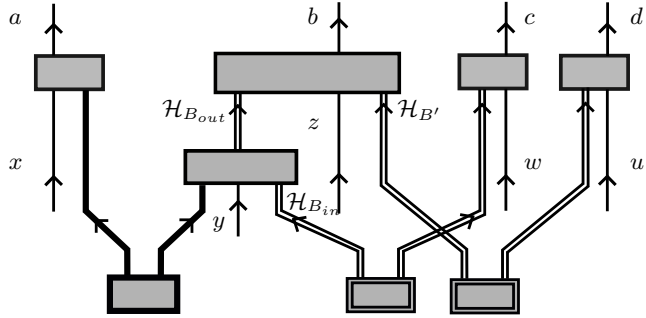
\begin{figure}[h!]
  \begin{center}
\includegraphics[width=0.48\textwidth]{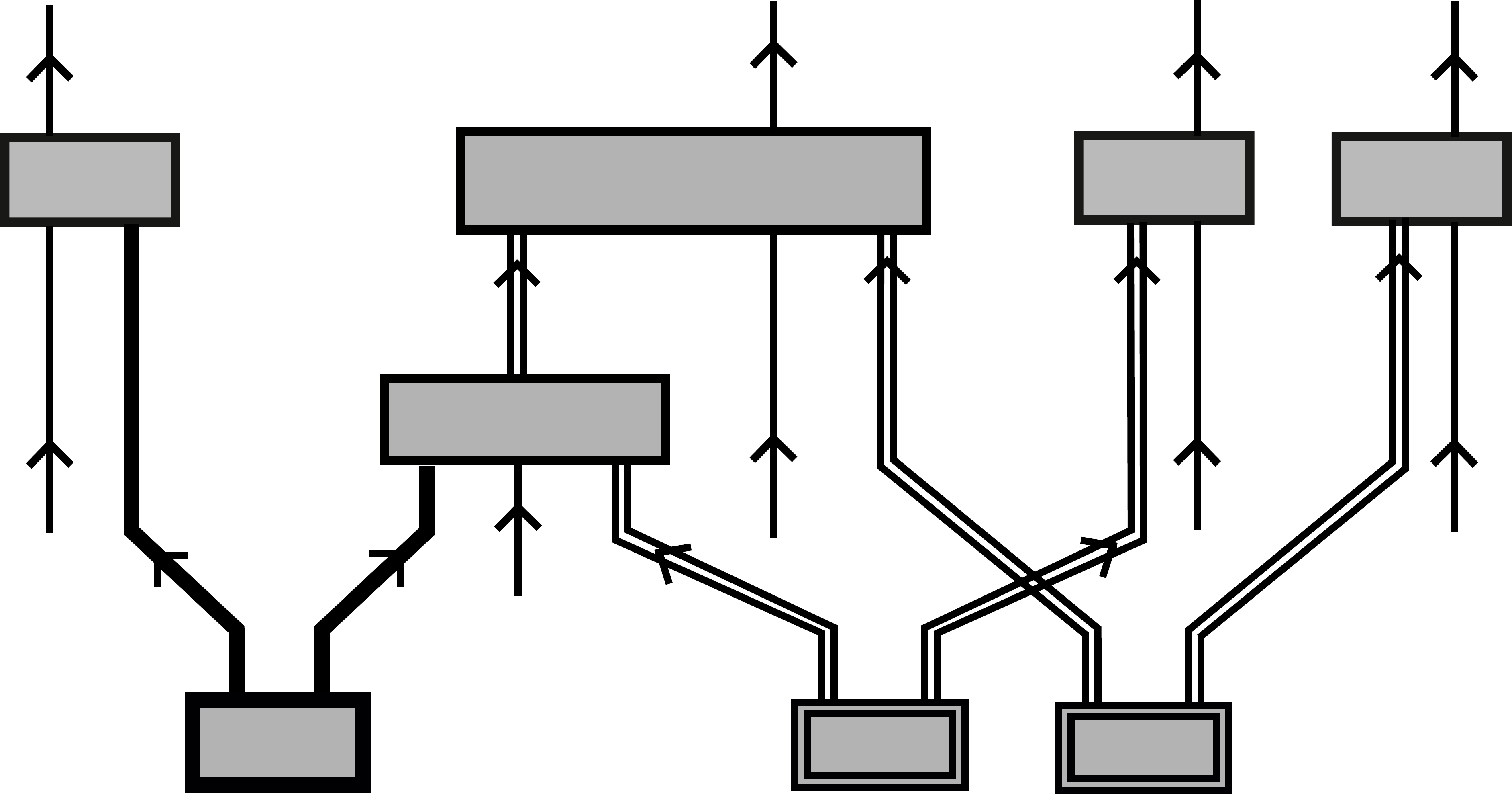}
\put(-235,55){$x$}
\put(-235,110){$a$}
\put(-158,33){$y$}
\put(-177,75){$\cH_{B_{out}}$}
\put(-122,110){$b$}
\put(-122,70){$z$}
\put(-130,40){$\cH_{B_{in}}$}
\put(-88,75){$\cH_{B'}$}
\put(-40,110){$c$}
\put(-40,55){$w$}
\put(0,110){$d$}
\put(0,55){$u$}
  \end{center}
  \caption{Depiction of the activation protocol for the channel EPR scenario. Systems that may be classical, quantum, or even
post-quantum, are represented by thick lines. Quantum systems are represented by double lines and classical
systems are depicted as single lines.}\label{fig:protocol-channel}
\end{figure}
\begin{enumerate}[Step 1:]
    \item \textit{Self-testing of $\As_{\C|\W}$ and $\As_{\D|\U}$.}
\end{enumerate}
The first step of the protocol is to certify that the state of Bob's quantum systems $\cH_{B_{in}}$ and $\cH_{B’}$ are both independently prepared in the assemblages $\As^{(r_1)}_{\C|\W}$ and $\As^{(r_2)}_{\D|\U}$. Using Bob's inputs $y \in \{1,2,3,4 \}$ and $z \in \{1,2,3,4 \}$, the elements of the assemblages need to be self-tested to be $\sigma^{(r_1)}_{c|w}=r_1\widetilde{\sigma}_{c|w}+(1-r_1)(\widetilde{\sigma}_{c|w})^{T}$ and $\sigma^{(r_2)}_{d|u}=r_2\widetilde{\sigma}_{d|u}+(1-r_2)(\widetilde{\sigma}_{d|u})^{T}$. Then, these assemblages must be aligned such that their tensor product can be written as $\{r (\widetilde{\sigma}_{c|w} \otimes \widetilde{\sigma}_{d|u})+(1-r)(\widetilde{\sigma}_{c|w} \otimes \widetilde{\sigma}_{d|u})^{T}\}$. The alignment procedure, which relies on the measurement settings $y \in \{\lozenge,\blacklozenge \}, z \in \{\lozenge,\blacklozenge \}$, is discussed in the Appendix~\ref{app:ST-channel}.

\begin{enumerate}[Step 2:]
    \item \textit{Certification of post-quantumness.}
\end{enumerate}
In the second step of the protocol, when $y=z=\star$, Bob applies the process $\Gamma^{BB_{in} \rightarrow B_{out}}$ and measures the system on $\cH_{B_{out}} \otimes \cH_{B’}$. Then, the post-quantumness of the correlations $p(a,b=0,c,d|x,y=\star,z=\star,w,u)$ can be tested with a tailored Bell functional. For now, we only consider the case when Bob's measurement is given by $M^{B_{out}B'}_{b=0|z=\star}=\ket{\phi^{+}}\bra{\phi^{+}}$ and the elements of the assemblages $\As_{\C|\W}$ and $\As_{\D|\U}$ are given by $\{\widetilde{\sigma}_{c|w}\}$ and $\{\widetilde{\sigma}_{d|u}\}$, respectively. The activation protocol also works without these assumptions, as we discuss in the Appendix~\ref{app:mixture} and ~\ref{app:thm}. Under these assumptions, the observed correlations are given by
\begin{align}\label{eq:corr-channel}
    p(a,b=0,c,d|x,y=\star,z=\star,w,u)=\tr{M^{B_{out}B'}_{b=0|z=\star} (\I_{a|x}(\widetilde{\sigma}_{c|w}) \otimes \widetilde{\sigma}_{d|u})}{}.
\end{align}
Due to Choi-Jamiołkowski isomorphism, the instrument $\I_{a|x}(\cdot)$ can be expressed as an operator $J(\I_{a|x})=(\I_{a|x} \otimes \id^{C})(\ket{\phi^{+}}\bra{\phi^{+}}) _{B_{in}C}$. Using the techniques introduced for the BwI and MDI protocols, we can rewrite Eq.~\eqref{eq:corr-channel} as
\begin{align}\label{eq:corr2-channel}
    p(a,b=0,c,d|x,y=\star,z=\star,w,u)=\frac{1}{2}\tr{((\widetilde{\sigma}_{d|u})^{T} \otimes \widetilde{M}_{c|w} ) J(\I_{a|x})}{}.
\end{align}
We will now transform the EPR functional given in Eq.~\eqref{eq:general_EPR_functional_channel} into a Bell functional. For $u \in \{1,2,3\}$, let $\pi_{d|u}$ denote the projector onto the eigenspace of the eigenstates of Pauli $Z$, $X$ and $Y$ operators with eigenvalue $(-1)^{d}$. The projectors $\{ \pi_{c|w} \otimes \pi_{d|u}\}$ form a basis of the set of Hermitian matrices in a four-dimensional Hilbert space. Then, the operators $\{F_{ax}\}_{a,x}$ can be decomposed as $F_{ax}=\sum_{c,d,w,u} \xi^{ax}_{cdwu}  ( \pi_{c|w} \otimes \pi_{d|u})$ for some real coefficients $\{\xi^{ax}_{cdwu}\}_{a,x,c,d,w,u}$. Construct the following Bell functional on the correlation $\textbf{\textrm{p}}=\{p(a,b,c,d|x,y,z,w,u)\}$:
\begin{align}\label{eq:Bell-functional}
    I^{*}[\textbf{\textrm{p}}] \equiv \sum_{a,x,c,d,w,u} \xi^{ax}_{cdwu} p(a,b=0,c,d|x,y=\star,z=\star,w,u).
\end{align}
To calculate the value of this functional, take the correlations given in Eq.~\eqref{eq:corr2-channel}:
\begin{align}
    I^{*}[\textbf{\textrm{p}}] &= \sum_{a,x,c,d,w,u} \xi^{ax}_{cdwu} \frac{1}{2}\tr{((\widetilde{\sigma}_{d|u})^{T} \otimes \widetilde{M}_{c|w} ) J(\I_{a|x})}{}, \\ \nonumber
&=\frac{1}{4} \sum_{a,x} \tr{\left(\sum_{c,d,w,u}\xi^{ax}_{cdwu}( \pi_{c|w} \otimes \pi_{d|u}) \right) J(\I_{a|x})}{}, \\ \nonumber
&=\frac{1}{4} \tr{\sum_{a,x}F_{ax}J(\I_{a|x})}{}.
\end{align}
Similarly to the BwI and MDI scenarios, $I^{*}[\textbf{\textrm{p}}] \geq 0$ for any $\textbf{\textrm{p}}$ arising from a quantum channel assemblage. For the measurement $M^{BB'}_{b=0|z=\star}=\ket{\phi^{+}}\bra{\phi^{+}}$ and the fixed assemblage $\In_{\A|\X}^{*}=\{\I_{a|x}^{*}(\cdot)\}$, the Bell functional evaluates to $I^{*}[\textbf{\textrm{p*}}] < 0$. The analysis of the case $M^{BB'}_{b=0|z=\star}\neq\ket{\phi^{+}}\bra{\phi^{+}}$ is the same as in the Bob-with-input scenario. Therefore, if one observes $I^{*}[\textbf{\textrm{p}}] < 0$, it certifies that the correlations $\textbf{\textrm{p}}$ were generated from a post-quantum channel assemblage.
\begin{thm}\label{thm-channel}
There always exists a Bell functional of the form
    \begin{align}\label{eq:functional-channel}
    I^{*}[\textbf{\textrm{p}}] &= \sum_{a,x,c,d,w,u} \xi^{ax}_{cdwu} p(a,b=0,c,d|x,y=\star,z=\star,w,u) 
    \end{align}
    for which $I^{*}[\textbf{\textrm{p}}] \geq 0$ when evaluated on any correlations $\textbf{\textrm{p}}=\{p(a,b,c,d|x,y,z,w,u)\}$ arising from quantum channel assemblage with elements $\{\I_{a|x}(\cdot)\}$ as $p(a,b=0,c,d|x,y=\star,z=\star,w,u)= \tr{M^{B_{out}B'}_{b=0|z=\star} (\I_{a|x}(\widetilde{\sigma}_{c|w}) \otimes \widetilde{\sigma}_{d|u})}{}$.
\end{thm}

In summary, based on the protocols introduced for the BwI and MDI scenarios, we showed that all channel assemblages that belong to the set $\In^{QC}$ can generate post-quantum Bell-type correlations when embedded in a larger network. This result completes our analysis and motivates the following question: is it possible to construct a generalised post-quantum assemblage that cannot generate post-quantum correlations? Such an assemblage would need to be either a generalised multipartite assemblage or have a structure different from those discussed in this paper and in Ref.~\cite{activation}. 

\section{Discussion}

This paper focused on generalised EPR scenarios. In these scenarios, Alice and Bob share a system. Alice performs measurements on her subsystem, while Bob locally processes his subsystem using either a local CPTP map on the input state, a measurement channel, or a general quantum channel.

We have shown how to embed a bipartite assemblage in a multipartite set-up such that post-quantum correlations can be generated in this larger network. For post-quantum assemblages that can only generate quantum correlations in a bipartite device-independent setting, this protocol activates their post-quantumness, i.e., it allows one to witness the post-quantum nature of the assemblage in a Bell-type scenario. We introduced and analyzed protocols for three generalised EPR scenarios: the Bob-with-input, measurement-device-independent, and channel EPR scenarios. 

The existence of post-quantum bipartite generalised EPR assemblages, which generate only quantum correlations when Bob's processing is somehow transformed to be device-independent, may suggest that there is a fundamental difference between post-quantum assemblages and post-quantum Bell-type correlations. However, our results show a relation between these two resources. For such sets of bipartite assemblages introduced in this work, we showed that they can generate post-quantum correlations when embedded in a larger network. As of now, there are no examples of post-quantum assemblages in EPR scenarios that cannot generate post-quantum correlations in a device-independent setting. This result was known for the standard multipartite EPR scenarios~\cite{activation} and here we proved it for the three generalisations considered in the paper. 

The construction of the protocols for all scenarios relies on the same idea. In each protocol, an additional quantum resource is self-tested and later used to transform EPR functionals to Bell functionals. This idea was first introduced as a method of device-independent entanglement certification~\cite{bowles2018device,bowles2018self}. However, it appears to have broad and powerful applications beyond this problem. Exploring other research areas where this concept could be applied seems interesting.

One interesting application of our results is to the quantum causal models program. Consider Alice and Bob sharing a generalised EPR assemblage with the objective of determining whether their shared common cause is classical, quantum, or post-quantum. If the assemblage is post-quantum but of the type discussed in this paper—capable of generating only quantum correlations in a bipartite device-independent setting—Alice and Bob might mistakenly conclude that their common cause is quantum if they only perform a bipartite Bell-like test. However, our results demonstrate that the post-quantum nature of the common cause can always be revealed if they employ the protocol we proposed to generate multipartite post-quantum device-independent correlations.

Looking forward, it would be interesting to see how the activation protocol changes for multipartite generalised EPR scenarios. The activation protocol for assemblages generated in experiments with multiple Alice's or Bob's should be similar to the one we introduced in this paper; however, a detailed study of this problem could reveal some differences.

\begin{acknowledgments}
BZ acknowledges support by the National Science Centre, Poland no. 2021/41/N/ST2/02242. MJH acknowledges the FQXi large grant ``The Emergence of Agents from Causal Order'' (FQXi FFF Grant number FQXi-RFP-1803B).
PS is a CIFAR Azrieli Global Scholar in the Quantum Information Science Program, and also gratefully acknowledges support from a Royal Society University Research Fellowship (UHQT/NFQI). ABS acknowledges support by the Foundation for Polish Science (IRAP project, ICTQT, contract no. 2018/MAB/5, co-financed by EU within Smart Growth Operational Programme), and by the IRA Programme, project no. FENG.02.01-IP.05-0006/23, financed by the FENG program 2021-2027, Priority FENG.02, Measure FENG.02.01., with the support of the FNP.
\end{acknowledgments}

\appendix

\section{Self-testing stage of the protocol}\label{app:ST}
\subsection{Bob-with-input EPR scenario: self-testing of \texorpdfstring{$\As^{(r)}_{\C|\W}$}{}}\label{app:ST-BwI}

The first step of the activation protocol for the Bob-with-input EPR scenario relies on the self-testing of the assemblage $\As^{(r)}_{\C|\W}$. The Bell inequality used to self-tests the assemblage $\As^{(r)}_{\C|\W}$ was analysed in Refs.~\cite{bowles2018self,Chen2021robustselftestingof}. In this appendix, we follow the result of Ref.~\cite{Chen2021robustselftestingof}.

Let us focus on the assemblage shared between Bob and Charlie in Fig.~\ref{fig:protocol-bwi}. From the statistics $\textbf{\textrm{p}}=\{p(abc|xyzw)\}$ obtained in the activation protocol, compute the marginal $\textbf{\textrm{p}}^{B'C}=\{p(bc|zw)\}$. Then, the following result holds.

\begin{lem}\label{lemma-ST}
    Given a marginal correlation $\textbf{\textrm{p}}=\{p(bc|zw)\}$ between Bob and Charlie, compute the following Bell functional
    \begin{align}
       \I_E =  \ev{C_1B_1} +\ev{C_1B_2} - \ev{C_1B_3}-\ev{C_1B_4}\\ \nonumber
       +\ev{C_2B_1}-\ev{C_2B_2}+\ev{C_2B_3}-\ev{C_2B_4} \\ \nonumber
       +\ev{C_3B_1}-\ev{C_3B_2}-\ev{C_3B_3}+\ev{C_3B_4}.
    \end{align}
If $\I_E =4\sqrt{3}$, then the assemblage shared between Bob and Charlie is given by $\As^{(r)}_{\C|\W}$ with elements $\sigma^{(r)}_{c|w}=r\widetilde{\sigma}_{c|w}+(1-r)(\widetilde{\sigma}_{c|w})^{T}$. Here, $\widetilde{\sigma}_{c|1}=(\id + (-1)^{c} Z)/4$, $\widetilde{\sigma}_{c|2}=(\id + (-1)^{c} X)/4$ and $\widetilde{\sigma}_{c|3}=(\id - (-1)^{c} Y)/4$.
\end{lem}
Bob's observables that saturate the functional $\I_E$ are given by $\{ \widetilde{B_1},\widetilde{B_2},\widetilde{B_3},\widetilde{B_4}\} = \{ (Z+X-Y)/\sqrt{3}, (-Z+X+Y)/\sqrt{3}, (-Z-X-Y)/\sqrt{3} \}$. Therefore, in the self-testing stage of the protocol, Bob's measurement $M^{BB'}_{b|z}$ on his two qubits corresponds to the observables $\id^{B} \otimes \widetilde{B_z}^{B'}$ for $z \in\{1,2,3,4\}$.

Like all self-testing statements, the self-tested assemblage has the desired form up to local isometries. For discussion about this result, including the possible isometries in the EPR scenario, we refer the reader to Ref.~\cite{Chen2021robustselftestingof}.

\subsection{Measurement-device-independent EPR scenario: self-testing of \texorpdfstring{$\As^{(r)}_{\C|\Z}$}{}}\label{app:ST-MDI}

The same self-testing result as in the Bob-with-input scenario is used in the formulation of the protocol for the measurement-device-independent scenario. In the activation of post-quantumness of MDI assemblages, the assemblage to be self-tested in the first step of the protocol is the standard EPR assemblage $\As^{(r)}_{\C|\Z}$ shared between Bob and Charlie in Fig.~\ref{fig:protocol-mdi}. Then, the self-testing result of Lemma~\ref{lemma-ST} must be used for the marginal $\textbf{\textrm{p}}^{B_{in}C}=\{p(bc|yz)\}$ computed from $\textbf{\textrm{p}}=\{p(abc|xyz)\}$.

\subsection{Channel EPR scenario: self-testing of \texorpdfstring{$\As^{(r_1)}_{\C|\W}$}{} and \texorpdfstring{$\As^{(r_2)}_{\D|\U}$}{}}\label{app:ST-channel}

In the protocol for the channel EPR scenario, two assemblages are self-tested: $\As_{\C|\W}$ shared between Bob and Charlie and $\As_{\D|\U}$ shared between Bob and Dani in Fig.~\ref{fig:protocol-channel}. The goal of the first step of the protocol is to show that the elements of the self-tested assemblages correspond to the basis of the operators $\{F_{ax}\}$ that define the Bell functional used for post-quantumness certification. In the channel EPR scenario, the operators $\{F_{ax}\}$ can be decomposed in the basis $\{ \pi_{c|w} \otimes \pi_{d|u}\}$ as $F_{ax}=\sum_{c,d,w,u} \xi^{ax}_{cdwu}  ( \pi_{c|w} \otimes \pi_{d|u})$ for some real numbers $\{\xi^{ax}_{cdwu}\}$. 

The self-testing stage relies on the marginals $\textbf{\textrm{p}}^{B_{in}C}=\{p(b_{out}c|yw)\}$ and $\textbf{\textrm{p}}^{B'D}=\{p(bd|zu)\}$, both calculated from the overall statistics $\textbf{\textrm{p}}=\{p(abcd|xyzwu)\}$. Here, $b_{out}$ represents a classical outcome obtained in the first stage of the protocol, where the Hilbert space $\cH_{B_{out}}$ encodes merely a classical outcome. If one simply uses the result of Lemma~\ref{lemma-ST}, the relevant assemblages will be certified to have elements $\sigma^{(r_1)}_{c|w}=r_1\widetilde{\sigma}_{c|w}+(1-r_1)(\widetilde{\sigma}_{c|w})^{T}$ and $\sigma^{(r_2)}_{d|u}=r_2\widetilde{\sigma}_{d|u}+(1-r_2)(\widetilde{\sigma}_{d|u})^{T}$, with unknown parameters $r_1$ and $r_2$. To avoid a false-positive detection in the activation protocol (which we discuss further in the Appendix~\ref{app:mixture-channel}), this self-test need to be accompanied by a complementary procedure to certify that the two assemblages are aligned, as we outline below. 

Let $\As^{(r)}_{\C\D|\W\U}$ denote the tensor product of the self-tested assemblages. The assemblage $\As^{(r)}_{\C\D|\W\U}$ has elements $\{\sigma^{(r_1)}_{c|w} \otimes \sigma^{(r_2)}_{d|u}\}$. The goal of the alignment procedure is to ensure that these elements have the form $\{r (\widetilde{\sigma}_{c|w} \otimes \widetilde{\sigma}_{d|u})+(1-r)(\widetilde{\sigma}_{c|w} \otimes \widetilde{\sigma}_{d|u})^{T}\}$ for some parameter $r$. To align the assemblages, we need to introduce new measurements on the system $\cH_{B_{out}} \otimes \cH_{B'}$ labeled by settings $\{\lozenge,\blacklozenge\}$. The outcome set of the measurement $\lozenge$ is denoted by $\B^{\lozenge}$ and has 4 outcomes, while the outcome set of the measurement $\blacklozenge$ is denoted by $\B^{\blacklozenge}$ and has 2 outcomes. Upon these measurements, the system $\cH_{B_{in}}$ is simply passed to the Bob's measuring device without going through the channel assemblage and measured with the system $\cH_{B'}$. In short, this additional alignment procedure self-tests the measurements corresponding to $\lozenge$ and $\blacklozenge$ to ensure that the relevant assemblage has the desired form. This result relies on Ref.~\cite[Lemma~3]{bowles2018self} and is also discussed in detail in Ref.~\cite[Appendix~B]{activation}.

\section{No false-positives}\label{app:mixture}

\subsection{Bob-with-input EPR scenario}\label{app:mixture-bwi}

Let us start with analyzing the Bob-with-input protocol. In the second step of the protocol, we assumed that the assemblage shared between Bob and Charlie is given by $\As^{(r=1)}_{\C|\W}=\widetilde{\As}_{\C|\W}$. However, the success of first step of the protocol can only guarantee that that the assemblage $\As^{(r)}_{\C|\W}$ that they share has elements of the form $\sigma^{(r)}_{c|w}=r\widetilde{\sigma}_{c|w}+(1-r)(\widetilde{\sigma}_{c|w})^{T}$, with an unknown parameter $r$. Here we show that if the Bob-with-input assemblage $\As_{\A|\X\Y}$ belongs to the quantum set, it is impossible to observe a value of $I_{BwI}^{*}[\textrm{\textbf{p}}]$ which violates the quantum bound regardless of the parameter $r$ that defines $\As^{(r)}_{\C|\W}$. In other words, if the self-testing stage of the protocol is successful and a violation of the quantum bound is observed in the second step, it must be the case that the BwI assemblage shared between Alice and Bob is post-quantum.

Consider the case when $r=0$ and $\As^{(r=0)}_{\C|\W}=\{(\widetilde{\sigma}_{c|w})^{T}\}$. Then, the observed correlations are given by
\begin{align}
    p(a,b=0,c|x,y,z=\star,w)_{T}&=\tr{M^{BB'}_{b=0|z=\star} (\sigma_{a|xy} \otimes (\widetilde{\sigma}_{c|w})^{T})}{\cH_{B},\cH_{B'}}  \\ \nonumber
    &=\frac{1}{2}\tr{((\widetilde{\sigma}_{c|w})^{T})^{T} \sigma_{a|xy}}{},
\end{align}
and the Bell functional can be written as
\begin{align}
    I_{BwI}^{*}[\textbf{\textrm{p}}_{T}] &= \sum_{a,x,y,c,w} \xi^{axy}_{cw} \frac{1}{2}\tr{((\widetilde{\sigma}_{c|w})^{T})^{T} \sigma_{a|xy}}{}, \\ \nonumber
&=\frac{1}{2} \sum_{a,x,y} \tr{\left(\sum_{c,w}\xi^{axy}_{cw}(\widetilde{\sigma}_{c|w})^{T}\right) (\sigma_{a|xy})^{T}}{}, \\ \nonumber
&=\frac{1}{4} \tr{\sum_{a,x,y}F_{axy}(\sigma_{a|xy})^{T}}{}.
\end{align}
If the assemblage $\{\sigma_{a|xy}\}$ is quantum, the assemblage $\{(\sigma_{a|xy})^{T}\}$ is quantum as well (due to the fact that the transpose operation on a qubit is a decomposable positive and trace-preserving map) and $I_{BwI}^{*}[\textbf{\textrm{p}}_{T}] \geq 0$. It follows that if $\{\sigma_{a|xy}\}$ has a quantum realisation,  $I_{BwI}^{*}[\textbf{\textrm{p}}_{T}] \geq 0$ and $I_{BwI}^{*}[\textbf{\textrm{p}}] \geq 0$. Hence, any correlation generated by measurements on the convex combination $\{r\widetilde{\sigma}_{c|w}+(1-r)(\widetilde{\sigma}_{c|w})^{T}\}$ must satisfy the quantum bound as well.

\subsection{Measurement-device-independent EPR scenario}\label{app:mixture-mdi}
The same reasoning as in the protocol for the Bob-with-input scenario can be applied to the measurement-device-independent scenario. If Bob and Charlie share an assemblage defined by $r=0$, i.e., $\As^{(r=0)}_{\C|\Z}=\{(\widetilde{\sigma}_{c|z})^{T}\}$, the observed correlations are given by

\begin{align}\label{eq:corr-MDI}
p(a,b,c|x,y=\star,z)_{T} &=\Ne_{ab|x}((\widetilde{\sigma}_{c|z})^{T}), \\ \nonumber
&=\tr{ (M^{C}_{c|z})^{T} \, J(\Ne_{ab|x})}{}, \\ \nonumber
&=\tr{ M^{C}_{c|z} \, J(\Ne_{ab|x})^{T}}{}.
\end{align}
Here, we used the fact that $\widetilde{\sigma}_{c|z}^{T}=\tr{(\id^{B_{in}} \otimes (M^{C}_{c|z})^{T}) \ket{\phi^{+}}\bra{\phi^{+}} _{B_{in}C}}{C}$. Then, the Bell functional reads
\begin{align}
    I_{MDI}^{*}[\textbf{\textrm{p}}_{T}] &= \sum_{a,b,x,c,z} \xi^{abx}_{cz} \tr{ M^{C}_{c|z} \, J(\Ne_{ab|x})^{T}}{}, \\ \nonumber
&=\sum_{a,b,x} \tr{\left(\sum_{c,z} \xi^{abx}_{cz}  \pi_{c|z}\right) \, J(\Ne_{ab|x})^{T}}{}, \\ \nonumber
&=\tr{\sum_{a,b,x}F_{abx}\, J(\Ne_{ab|x})^{T}}{}.
\end{align}
Using the same arguments as for the Bob-with-input case, one can see that if $J(\Ne_{ab|x})$ corresponds to a quantum assemblage, then $J(\Ne_{ab|x})^{T}$ has a quantum realisation and any correlation generated by measurements on the convex combination $\{r\widetilde{\sigma}_{c|w}+(1-r)(\widetilde{\sigma}_{c|w})^{T}\}$ must satisfy the quantum bound.

\subsection{Channel EPR scenario}\label{app:mixture-channel}

In the set-up for the channel EPR scenario, the first step of the protocol certifies that the relevant assemblages have elements $\{r (\widetilde{\sigma}_{c|w} \otimes \widetilde{\sigma}_{d|u})+(1-r)(\widetilde{\sigma}_{c|w} \otimes \widetilde{\sigma}_{d|u})^{T}\}$ for some parameter $r$. In the main text, we showed that when we restrict the tensor product of these assemblages to have the form $\{\widetilde{\sigma}_{c|w} \otimes \widetilde{\sigma}_{d|u}\}$, the quantum bound of the Bell functional given in Eq.~\eqref{eq:functional-channel} is equal to zero. In this appendix, we show that a false positive detection of post-quantumness is not possible even if their tensor product is instead $\{(\widetilde{\sigma}_{c|w} \otimes \widetilde{\sigma}_{d|u})^{T}\}$ (the result of no false-positive detection with $\{r (\widetilde{\sigma}_{c|w} \otimes \widetilde{\sigma}_{d|u})+(1-r)(\widetilde{\sigma}_{c|w} \otimes \widetilde{\sigma}_{d|u})^{T}\}$ follows from a convexity argument, similar to the BwI and MDI protocols).

Let $r=0$, hence $\As^{(r=0)}_{\C|\W}=\{(\widetilde{\sigma}_{c|w})^{T}\}$ and $\As^{(r=0)}_{\D|\U}=\{(\widetilde{\sigma}_{d|u})^{T}\}$. Then, the correlations generated in the protocol set-up can be written as
\begin{align}
    p(a,b=0,c,d|x,y=\star,z=\star,w,u)_{T} &=\tr{M^{B_{out}B'}_{b=0|z=\star} (\I_{a|x}((\widetilde{\sigma}_{c|w})^T) \otimes (\widetilde{\sigma}_{d|u})^T)}{\cH_{B_{out}},\cH_{B'}}, \\ \nonumber
    &=\frac{1}{2}\tr{(((\widetilde{\sigma}_{d|u})^{T})^T \otimes (\widetilde{M}_{c|w})^T ) J(\I_{a|x})}{\cH_{B_{out}},\cH_{B'}}, \\ \nonumber
    &=\frac{1}{2}\tr{((\widetilde{\sigma}_{d|u})^{T} \otimes \widetilde{M}_{c|w} ) (J(\I_{a|x}))^T}{\cH_{B_{out}},\cH_{B'}},
\end{align}
where we used the same transformations as in the Appendix~\ref{app:mixture-bwi} for the BwI protocol and in the Appendix~\ref{app:mixture-mdi} for the MDI protocol. The only thing left to show is that if $J(\I_{a|x})$ is a quantum assemblage, $J(\I_{a|x})^T$ has a quantum realisation as well. Recall that the Choi matrix of a quantum channel assemblage has the following form
\begin{equation}
J(\I_{a|x})=\tr{(M_{a|x} \otimes \Gamma^{BB_{in} \rightarrow B_{out}} \otimes \id^{C})[\rho_{AB} \otimes (\ket{\phi^{+}}\bra{\phi^{+}}) _{B_{in}C}]}{A}.
\end{equation}
Therefore, we can write its transpose as
\begin{align}
J(\I_{a|x})^T &= \tr{(M_{a|x} \otimes \Gamma^{BB_{in} \rightarrow B_{out}} \otimes \id^{C})[\rho_{AB} \otimes (\ket{\phi^{+}}\bra{\phi^{+}}) _{B_{in}C}]}{A}^T, \\ \nonumber
&= \tr{(\id^{A} \otimes T^{B_{out}C})(M_{a|x} \otimes \Gamma^{BB_{in} \rightarrow B_{out}} \otimes \id^{C})[\rho_{AB} \otimes (\ket{\phi^{+}}\bra{\phi^{+}}) _{B_{in}C}]}{A}.
\end{align}
To show that $J(\I_{a|x})^T$ admits a quantum realisation, we will use the following lemma.
\begin{lem}\label{lemma}
    Let $\Phi(\cdot)$ and $\Psi(\cdot)$ be CPTP maps and $\rho$ be a quantum state. The application of a CPTP map on a quantum state followed by a transpose, $\left( \Phi (\rho) \right)^{T}$, is equivalent to first applying the transpose on the state, and then acting on it with a different CPTP map, $ \Psi (\rho^{T})$.
\end{lem}
\begin{proof}
Using Kraus decomposition of the map $\Phi(\cdot)$, we can write $\Phi (\rho) = \sum_{k} A_k \rho A_k^{\dagger}$, where $\sum_{k}A_k^{\dagger}A_k = \id$. The transpose of this new state can be written as
\begin{align}
  \left( \Phi (\rho) \right)^{T} &= \left( \sum_{k} A_k \rho A_k^{\dagger} \right)^{T}, \\ \nonumber
  &= \sum_{k} (A_k^{\dagger})^{T} \rho^{T}  (A_k)^{T} , \\ \nonumber
  &= \sum_{k} B_k \rho^{T}  B_k^{\dagger},
\end{align}
where we defined a new operator $B_k = (A_k^{\dagger})^{T}$. This new operator defines a CPTP map $ \Psi (\cdot) = \sum_{k} B_k (\cdot) B_k^{\dagger}$ as 
\begin{align}
    \sum_{k} B_k^{\dagger} B_k  &= \sum_{k} (A_k)^{T} (A_k^{\dagger})^{T}, \\  \nonumber
     &= \sum_{k} (A_k^{\dagger} A_k)^{T}, \\ \nonumber
     &= \left( \sum_{k} A_k^{\dagger} A_k\right)^{T}, \\ \nonumber
     &= \id.
\end{align}
Therefore, we showed that $ \left( \Phi (\rho) \right)^{T} = \Psi (\rho^{T})$.
\end{proof}
Due to Lemma~\ref{lemma}, we can move the transpose from the output system of the CPTP map $\Gamma^{BB_{in} \rightarrow B_{out}}$ to its input systems. Then, we can write
\begin{align}
J(\I_{a|x})^T &= \tr{(M_{a|x} \otimes \Gamma'^{BB_{in} \rightarrow B_{out}} \otimes \id^{C})(\id^{A} \otimes T^{BB_{in}C})[\rho_{AB} \otimes (\ket{\phi^{+}}\bra{\phi^{+}}) _{B_{in}C}]}{A}, \\ \nonumber
&= \tr{(M_{a|x} \otimes \Gamma'^{BB_{in} \rightarrow B_{out}} \otimes \id^{C})[\rho_{AB}^{T_B} \otimes (\ket{\phi^{+}}\bra{\phi^{+}}) _{B_{in}C}^T]}{A}, \\ \nonumber
&= \tr{(M_{a|x} \otimes \Gamma'^{BB_{in} \rightarrow B_{out}} \otimes \id^{C})^T_{A}[\rho_{AB}^{T} \otimes (\ket{\phi^{+}}\bra{\phi^{+}}) _{B_{in}C}^T]}{A}, \\ \nonumber
&= \tr{((M_{a|x})^T \otimes \Gamma'^{BB_{in} \rightarrow B_{out}} \otimes \id^{C})[\rho_{AB}^{T} \otimes (\ket{\phi^{+}}\bra{\phi^{+}}) _{B_{in}C}^T]}{A}.
\end{align}
Therefore, if $J(\I_{a|x})$ is a valid quantum assemblage, then its transpose $J(\I_{a|x})^T$ admits of a quantum realisation in terms of the measurements $\{(M_{a|x})^{T}\}$, a CPTP map $\Gamma'^{BB_{in} \rightarrow B_{out}}$ and the states $\rho_{AB}^T$ and $(\ket{\phi^{+}}\bra{\phi^{+}}) _{B_{in}C}^T$. This completes our argument. 

Finally, we can revisit the additional alignment step of the self-testing stage and discuss why is it necessary. Imagine the assemblages are only self-tested to have the form $\sigma^{(r_1)}_{c|w}=r_1\widetilde{\sigma}_{c|w}+(1-r_1)(\widetilde{\sigma}_{c|w})^{T}$ and $\sigma^{(r_2)}_{d|u}=r_2\widetilde{\sigma}_{d|u}+(1-r_2)(\widetilde{\sigma}_{d|u})^{T}$, with unknown parameters $r_1$ and $r_2$. The case of $r_1 = r_2 = 1$ was discussed in the main text. The contrary situation, i.e., when $r_1 = r_2 = 0$, is discussed in this appendix. However, if the self-tested assemblages would be defined by $r_1 = 1$ and $r_2 = 0$, this could create a false-positive detection in the activation protocol. Therefore, the alignment stage of the step 1 of the protocol is necessary to eliminate the possibility of a false-positive detection. 

\section{Optimizing the Bell functional over all measurements}\label{app:thm}

In the activation protocols for the Bob-with-input and channel EPR scenarios, we first assumed that $M_{b=0|z=\star}=\ket{\phi^{+}}\bra{\phi^{+}}$. Then, we have shown that under this assumption the quantum bound of the Bell functionals constructed for these scenarios is equal to zero. In this appendix, we show that the quantum bound of these functionals, when optimised over all possible measurements $M_{b=0|z=\star}$, is still equal to zero. Therefore, observing a negative value of the Bell functional in an activation set-up necessarily means that the underlying generalised EPR assemblage is post-quantum. Here we focus on the Bob-with-input scenario; the proof for the channel EPR scenario has the same structure.

In Section~\ref{sec:protocol-bwi}, we have shown that
\begin{align}\label{eq:app-diag}
     I_{BwI}^{*}[\textbf{\textrm{p}}]=\frac{1}{2} \sum_{a,x,y,c,w} \xi^{axy}_{cw} \tr{M_{0|\star} (\sigma_{a|xy} \otimes \widetilde{\sigma}_{c|w})}{\cH_{B},\cH_{B'}} \geq 0,
     \end{align}
for $M_{b=0|z=\star}=\ket{\phi^{+}}\bra{\phi^{+}}$. The proof that this inequality holds for arbitrary $M_{b=0|z=\star}$ relies on the observations made in Ref.~\cite[Appendix~A]{activation}. Here we first present it using diagrammatic notation introduced in Ref.~\cite{coecke2018picturing}. The standard mathematical proof can be found below the diagrammatic one. 

Let us
represent an arbitrary quantum Bob-with-input assemblage on a system $\cH_{B'}$ with the diagram
\begin{align*}
\begin{tikzpicture}
	\begin{pgfonlayer}{nodelayer}
		\node [style=none] (0) at (-2, 0) {};
		\node [style=none] (1) at (2, 0) {};
		\node [style=none] (2) at (0, -1) {};
		\node [style=none] (4) at (0, -0.5) {$\rho$};
		\node [style=copoint] (6) at (-1, 1.75) {$a|x$};
		\node [style=none] (7) at (1, 0.75) {};
		\node [style=none] (8) at (-1, 0) {};
		\node [style=none] (10) at (1, 0) {};
		\node [style=none] (12) at (1.75, 2.5) {$\cH_{B}$};
		\node [style=none] (13) at (1.75, 0.75) {};
		\node [style=none] (14) at (0.25, 0.75) {};
		\node [style=none] (15) at (1, 1.75) {};
		\node [style=none] (16) at (0.25, 1.75) {};
		\node [style=none] (17) at (1.75, 1.75) {};
		\node [style=none] (18) at (1, 2.5) {};
		\node [style=none] (19) at (1, 1.25) {$\varepsilon_y$};
	\end{pgfonlayer}
	\begin{pgfonlayer}{edgelayer}
		\draw (0.center) to (2.center);
		\draw (2.center) to (1.center);
		\draw (1.center) to (0.center);
		\draw [qWire] (8.center) to (6);
		\draw [qWire] (10.center) to (7.center);
		\draw (16.center) to (14.center);
		\draw (14.center) to (13.center);
		\draw (13.center) to (17.center);
		\draw (17.center) to (16.center);
		\draw [qWire] (15.center) to (18.center);
	\end{pgfonlayer}
\end{tikzpicture}
\end{align*}
for a quantum system $\cH_{B}$, an effect $\begin{tikzpicture}
	\begin{pgfonlayer}{nodelayer}
		\node [style=none] (0) at (0, 0.5) {};
		\node [style=copoint] (1) at (0, 0.5) {$a|x$};
		\node [style=none] (2) at (0, -1) {};
	\end{pgfonlayer}
	\begin{pgfonlayer}{edgelayer}
		\draw [style=qWire] (1) to (2.center);
	\end{pgfonlayer}
\end{tikzpicture}
$, a state $\begin{tikzpicture}
	\begin{pgfonlayer}{nodelayer}
		\node [style=none] (0) at (-1.5, 0) {};
		\node [style=none] (1) at (1.5, 0) {};
		\node [style=none] (2) at (0, -1) {};
		\node [style=none] (4) at (0, -0.5) {$\rho$};
		\node [style=none] (5) at (0, 1) {};
		\node [style=none] (6) at (-1, 1) {};
		\node [style=none] (7) at (1, 1) {};
		\node [style=none] (8) at (-1, 0) {};
		\node [style=none] (9) at (0, 0) {};
		\node [style=none] (10) at (1, 0) {};
	\end{pgfonlayer}
	\begin{pgfonlayer}{edgelayer}
		\draw (0.center) to (2.center);
		\draw (2.center) to (1.center);
		\draw (1.center) to (0.center);
		\draw [qWire] (8.center) to (6.center);
		\draw [qWire] (10.center) to (7.center);
	\end{pgfonlayer}
\end{tikzpicture}$ and a channel $\begin{tikzpicture}
	\begin{pgfonlayer}{nodelayer}
		\node [style=none] (7) at (0, -0.5) {};
		\node [style=none] (10) at (0, -1.25) {};
		\node [style=none] (13) at (0.75, -0.5) {};
		\node [style=none] (14) at (-0.75, -0.5) {};
		\node [style=none] (15) at (0, 0.5) {};
		\node [style=none] (16) at (-0.75, 0.5) {};
		\node [style=none] (17) at (0.75, 0.5) {};
		\node [style=none] (18) at (0, 1.25) {};
		\node [style=none] (19) at (0, 0) {$\varepsilon_y$};
	\end{pgfonlayer}
	\begin{pgfonlayer}{edgelayer}
		\draw [qWire] (10.center) to (7.center);
		\draw [in=90, out=-90] (16.center) to (14.center);
		\draw (14.center) to (13.center);
		\draw (13.center) to (17.center);
		\draw (17.center) to (16.center);
		\draw [qWire] (15.center) to (18.center);
	\end{pgfonlayer}
\end{tikzpicture}$. Now, we can represent Eq.~\eqref{eq:app-diag} as
\begin{align}\label{eq:diag2}
I_{BwI}^{*}[\mathbf{p}] = \frac{1}{4} \sum_{a,x,y}\,\begin{tikzpicture}
	\begin{pgfonlayer}{nodelayer}
		\node [style=none] (1) at (2.5, 0.75) {};
		\node [style=none] (7) at (2, 2.25) {};
		\node [style=none] (10) at (2, 0.75) {};
		\node [style=none] (15) at (3.5, -1) {};
		\node [style=none] (16) at (6.5, -1) {};
		\node [style=none] (17) at (5, -2.5) {};
		\node [style=none] (18) at (5, -1.5) {\color{white}$\mathcal{F}_{axy}$};
		\node [style=none] (19) at (5, 2.25) {};
		\node [style=none] (20) at (5, -1) {};
		\node [style=none] (21) at (6, 1.5) {$\cH_{B'}$};
		\node [style=none] (22) at (-1, -1) {};
		\node [style=none] (23) at (3, -1) {};
		\node [style=none] (24) at (1, -2) {};
		\node [style=none] (25) at (1, -1.5) {$\rho$};
		\node [style=copoint] (26) at (0, 0.75) {$a|x$};
		\node [style=none] (27) at (2, -0.25) {};
		\node [style=none] (28) at (0, -1) {};
		\node [style=none] (29) at (2, -1) {};
		\node [style=none] (30) at (2.75, 1.5) {$\cH_{B}$};
		\node [style=none] (31) at (2.75, -0.25) {};
		\node [style=none] (32) at (1.25, -0.25) {};
		\node [style=none] (33) at (2, -0.25) {};
		\node [style=none] (34) at (1.25, 0.75) {};
		\node [style=none] (35) at (2.75, 0.75) {};
		\node [style=none] (36) at (2, 1.5) {};
		\node [style=none] (37) at (2, 0.25) {$\varepsilon_{y}$};
	\end{pgfonlayer}
	\begin{pgfonlayer}{edgelayer}
		\draw [qWire] (10.center) to (7.center);
		\draw [fill=black] (15.center)
			 to (17.center)
			 to (16.center)
			 to cycle;
		\draw [qWire] (20.center) to (19.center);
		\draw [qWire, bend left=90] (7.center) to (19.center);
		\draw (22.center) to (24.center);
		\draw (24.center) to (23.center);
		\draw (23.center) to (22.center);
		\draw [qWire] (28.center) to (26);
		\draw [qWire] (29.center) to (27.center);
		\draw (34.center) to (32.center);
		\draw (32.center) to (31.center);
		\draw (31.center) to (35.center);
		\draw (35.center) to (34.center);
	\end{pgfonlayer}
\end{tikzpicture} \, \geq 0 \,.
\end{align}
Here, the representation of the operators $F_{axy}$ as black triangles implies that they are possibly nonphysical (post-quantum) processes.
The statement of Theorem~\ref{thm-bwi} can be expressed in the diagrammatic language as:
\begin{align}\label{eq:diag3}
I_{BwI}^{*}[\mathbf{p}] = \frac{1}{4} \sum_{a,x,y}\,\begin{tikzpicture}
	\begin{pgfonlayer}{nodelayer}
		\node [style=none] (1) at (2.5, 0.75) {};
		\node [style=none] (7) at (2, 2.25) {};
		\node [style=none] (10) at (2, 0.75) {};
		\node [style=none] (15) at (3.5, -1) {};
		\node [style=none] (16) at (6.5, -1) {};
		\node [style=none] (17) at (5, -2.5) {};
		\node [style=none] (18) at (5, -1.5) {\color{white}$\mathcal{F}_{axy}$};
		\node [style=none] (19) at (5, 2.25) {};
		\node [style=none] (20) at (5, -1) {};
		\node [style=none] (21) at (6, 1.5) {$\cH_{B'}$};
		\node [style=none] (22) at (-1, -1) {};
		\node [style=none] (23) at (3, -1) {};
		\node [style=none] (24) at (1, -2) {};
		\node [style=none] (25) at (1, -1.5) {$\rho$};
		\node [style=copoint] (26) at (0, 0.75) {$a|x$};
		\node [style=none] (27) at (2, -0.25) {};
		\node [style=none] (28) at (0, -1) {};
		\node [style=none] (29) at (2, -1) {};
		\node [style=none] (30) at (2.75, 1.5) {$\cH_{B}$};
		\node [style=none] (31) at (2.75, -0.25) {};
		\node [style=none] (32) at (1.25, -0.25) {};
		\node [style=none] (33) at (2, -0.25) {};
		\node [style=none] (34) at (1.25, 0.75) {};
		\node [style=none] (35) at (2.75, 0.75) {};
		\node [style=none] (36) at (2, 1.5) {};
		\node [style=none] (37) at (2, 0.25) {$\varepsilon_{y}$};
		\node [style=none] (38) at (1.5, 2.25) {};
		\node [style=none] (39) at (5.5, 2.25) {};
		\node [style=none] (40) at (3.5, 3.5) {};
		\node [style=none] (41) at (3.5, 2.75) {$M_{0|\star}$};
	\end{pgfonlayer}
	\begin{pgfonlayer}{edgelayer}
		\draw [qWire] (10.center) to (7.center);
		\draw [fill=black] (15.center)
			 to (17.center)
			 to (16.center)
			 to cycle;
		\draw [qWire] (20.center) to (19.center);
		\draw (22.center) to (24.center);
		\draw (24.center) to (23.center);
		\draw (23.center) to (22.center);
		\draw [qWire] (28.center) to (26);
		\draw [qWire] (29.center) to (27.center);
		\draw (34.center) to (32.center);
		\draw (32.center) to (31.center);
		\draw (31.center) to (35.center);
		\draw (35.center) to (34.center);
		\draw (40.center) to (39.center);
		\draw (39.center) to (38.center);
		\draw (38.center) to (40.center);
	\end{pgfonlayer}
\end{tikzpicture} \, \geq 0 \,,
\end{align}
where $M_{0|\star}$ now represents a general measurement. To prove this statement, first define a new quantum state $\rho'$ as
\begin{align}\label{eq:diag4}
\begin{tikzpicture}
	\begin{pgfonlayer}{nodelayer}
		\node [style=none] (19) at (4, 2) {};
		\node [style=none] (20) at (4, 0) {};
		\node [style=none] (21) at (4.75, 1) {$\cH_{B'}$};
		\node [style=none] (22) at (0.5, 2) {};
		\node [style=none] (23) at (4.5, 2) {};
		\node [style=none] (24) at (2.5, 3.25) {};
		\node [style=none] (25) at (4.25, 1.5) {};
		\node [style=none] (26) at (4.25, 1.5) {};
		\node [style=none] (27) at (2.5, 2.5) {$M_{0|\star}$};
		\node [style=none] (28) at (6, 1.75) {};
		\node [style=none] (29) at (6, 0) {};
		\node [style=none] (30) at (7, 1) {$\cH_{B}$};
		\node [style=none] (34) at (-6.5, -1.25) {$\rho^\prime$};
		\node [style=none] (44) at (-3, 0) {$=$};
		\node [style=none] (45) at (-8.5, -0.75) {};
		\node [style=none] (46) at (-4.5, -0.75) {};
		\node [style=none] (47) at (-6.5, -1.75) {};
		\node [style=none] (49) at (-5.5, 1.75) {};
		\node [style=none] (50) at (-5.5, -0.75) {};
		\node [style=none] (51) at (-4.75, 1.75) {$\cH_{B}$};
		\node [style=none] (59) at (-1.75, -0.75) {};
		\node [style=none] (60) at (2.25, -0.75) {};
		\node [style=none] (61) at (0.25, -1.75) {};
		\node [style=none] (62) at (0.25, -1.25) {$\rho$};
		\node [style=none] (63) at (1.25, 0) {};
		\node [style=none] (64) at (1.25, -0.75) {};
		\node [style=none] (65) at (2, 1.5) {$\cH_{B}$};
		\node [style=none] (66) at (2, 0) {};
		\node [style=none] (67) at (0.5, 0) {};
		\node [style=none] (68) at (1.25, 1) {};
		\node [style=none] (69) at (0.5, 1) {};
		\node [style=none] (70) at (2, 1) {};
		\node [style=none] (71) at (1.25, 2) {};
		\node [style=none] (72) at (1.25, 0.5) {$\varepsilon_y$};
		\node [style=none] (74) at (-7.5, -0.75) {};
		\node [style=none] (76) at (-0.75, -0.75) {};
		\node [style=none] (77) at (-7.5, 1.75) {};
		\node [style=none] (78) at (-0.75, 2) {};
		\node [style=none] (79) at (-1.5, 1.5) {$\cH_{A}$};
		\node [style=none] (80) at (-8.25, 1.75) {$\cH_{A}$};
	\end{pgfonlayer}
	\begin{pgfonlayer}{edgelayer}
		\draw [qWire] (20.center) to (19.center);
		\draw (24.center) to (23.center);
		\draw (23.center) to (22.center);
		\draw (22.center) to (24.center);
		\draw [qWire] (29.center) to (28.center);
		\draw [qWire, in=270, out=-90, looseness=3.00] (20.center) to (29.center);
		\draw (45.center) to (47.center);
		\draw (47.center) to (46.center);
		\draw (46.center) to (45.center);
		\draw [qWire] (50.center) to (49.center);
		\draw (59.center) to (61.center);
		\draw (61.center) to (60.center);
		\draw (60.center) to (59.center);
		\draw [qWire] (64.center) to (63.center);
		\draw (69.center) to (67.center);
		\draw (67.center) to (66.center);
		\draw (66.center) to (70.center);
		\draw (70.center) to (69.center);
		\draw [qWire] (68.center) to (71.center);
		\draw [qWire] (77.center) to (74.center);
		\draw [qWire] (76.center) to (78.center);
	\end{pgfonlayer}
\end{tikzpicture}
\end{align}
This is a valid quantum state, although it can be unnormalised. Then, we can use this construction to rewrite the Bell functional as
\begin{align}\label{eq:diag5}
I_{BwI}^{*}[\mathbf{p}] = \frac{1}{4} \sum_{a,x,y}\,\begin{tikzpicture}
	\begin{pgfonlayer}{nodelayer}
		\node [style=none] (1) at (2.5, 0.75) {};
		\node [style=none] (7) at (2, 2.25) {};
		\node [style=none] (10) at (2, 0.75) {};
		\node [style=none] (15) at (3.5, -1) {};
		\node [style=none] (16) at (6.5, -1) {};
		\node [style=none] (17) at (5, -2.5) {};
		\node [style=none] (18) at (5, -1.5) {\color{white}$\mathcal{F}_{axy}$};
		\node [style=none] (19) at (5, 2.25) {};
		\node [style=none] (20) at (5, -1) {};
		\node [style=none] (21) at (6, 1.5) {$\cH_{B'}$};
		\node [style=none] (22) at (-1, -1) {};
		\node [style=none] (23) at (3, -1) {};
		\node [style=none] (24) at (1, -2) {};
		\node [style=none] (25) at (1, -1.5) {$\rho$};
		\node [style=copoint] (26) at (0, 0.75) {$a|x$};
		\node [style=none] (27) at (2, -0.25) {};
		\node [style=none] (28) at (0, -1) {};
		\node [style=none] (29) at (2, -1) {};
		\node [style=none] (30) at (2.75, 1.5) {$\cH_{B}$};
		\node [style=none] (31) at (2.75, -0.25) {};
		\node [style=none] (32) at (1.25, -0.25) {};
		\node [style=none] (33) at (2, -0.25) {};
		\node [style=none] (34) at (1.25, 0.75) {};
		\node [style=none] (35) at (2.75, 0.75) {};
		\node [style=none] (36) at (2, 1.5) {};
		\node [style=none] (37) at (2, 0.25) {$\varepsilon_{y}$};
		\node [style=none] (38) at (1.5, 2.25) {};
		\node [style=none] (39) at (5.5, 2.25) {};
		\node [style=none] (40) at (3.5, 3.5) {};
		\node [style=none] (41) at (3.5, 2.75) {$M_{0|\star}$};
	\end{pgfonlayer}
	\begin{pgfonlayer}{edgelayer}
		\draw [qWire] (10.center) to (7.center);
		\draw [fill=black] (15.center)
			 to (17.center)
			 to (16.center)
			 to cycle;
		\draw [qWire] (20.center) to (19.center);
		\draw (22.center) to (24.center);
		\draw (24.center) to (23.center);
		\draw (23.center) to (22.center);
		\draw [qWire] (28.center) to (26);
		\draw [qWire] (29.center) to (27.center);
		\draw (34.center) to (32.center);
		\draw (32.center) to (31.center);
		\draw (31.center) to (35.center);
		\draw (35.center) to (34.center);
		\draw (40.center) to (39.center);
		\draw (39.center) to (38.center);
		\draw (38.center) to (40.center);
	\end{pgfonlayer}
\end{tikzpicture} 
=\, \frac{1}{4} \sum_{a,x,y}\,\begin{tikzpicture}
	\begin{pgfonlayer}{nodelayer}
		\node [style=none] (1) at (2.5, 0.75) {};
		\node [style=none] (7) at (2, 2.25) {};
		\node [style=none] (10) at (2, 0.75) {};
		\node [style=none] (15) at (3.5, -1) {};
		\node [style=none] (16) at (6.5, -1) {};
		\node [style=none] (17) at (5, -2.5) {};
		\node [style=none] (18) at (5, -1.5) {\color{white}$\mathcal{F}_{axy}$};
		\node [style=none] (19) at (5, 2.25) {};
		\node [style=none] (20) at (5, -1) {};
		\node [style=none] (21) at (6, 1.5) {$\cH_{B'}$};
		\node [style=none] (22) at (-1, -1) {};
		\node [style=none] (23) at (3, -1) {};
		\node [style=none] (24) at (1, -2) {};
		\node [style=none] (25) at (1, -1.5) {$\rho'$};
		\node [style=copoint] (26) at (0, 0.75) {$a|x$};
		\node [style=none] (27) at (2, -0.25) {};
		\node [style=none] (28) at (0, -1) {};
		\node [style=none] (29) at (2, -1) {};
		\node [style=none] (30) at (2.75, 1.5) {$\cH_{B}$};
		\node [style=none] (31) at (2.75, -0.25) {};
		\node [style=none] (32) at (1.25, -0.25) {};
		\node [style=none] (33) at (2, -0.25) {};
		\node [style=none] (34) at (1.25, 0.75) {};
		\node [style=none] (35) at (2.75, 0.75) {};
		\node [style=none] (36) at (2, 1.5) {};
	\end{pgfonlayer}
	\begin{pgfonlayer}{edgelayer}
		\draw [qWire] (29.center) to (7.center);
		\draw [fill=black] (15.center)
			 to (17.center)
			 to (16.center)
			 to cycle;
		\draw [qWire] (20.center) to (19.center);
		\draw [qWire, bend left=90] (7.center) to (19.center);
		\draw (22.center) to (24.center);
		\draw (24.center) to (23.center);
		\draw (23.center) to (22.center);
		\draw [qWire] (28.center) to (26);
	\end{pgfonlayer}
\end{tikzpicture} \, \geq 0 \,,
\end{align}
The inequality follows from Eq.~\eqref{eq:diag2}, which is valid for all quantum states and quantum channels.

Let us now present the proof in standard mathematical notation. Recall that any quantum Bob-with-input assemblage can be expressed as 
\begin{align}
\sigma_{a|xy}=\tr{(M_{a|x}\otimes \mathcal{E}^{B}_y)\rho_{AB}}{A}.
\end{align}
Define a new quantum state on the Hilbert space $\cH_A \otimes \cH_{B_c}$ as
\begin{align}\label{eq:app-newstate}
\rho' = \tr{M^{BB'}_{b=0|z=\star} (\id^{A}\otimes \mathcal{E}^{B}_y\rho_{AB} \otimes \ket{\phi^{+}}\bra{\phi^{+}}_{B'B_c}}{BB'}.
\end{align}
Notice that we introduced an additional system on $B_c$ which will be useful in the calculations later. The state $\rho'$ is a valid unnormalised quantum state. Now, recall that the Bell functional in Theorem~\ref{thm-bwi} is given by
\begin{align}
     I_{BwI}^{*}[\textbf{\textrm{p}}]&=\frac{1}{2} \sum_{a,x,y,c,w} \xi^{axy}_{cw} \tr{M_{0|\star} (\sigma_{a|xy} \otimes \widetilde{\sigma}_{c|w})}{BB'} \\
     &=\frac{1}{4} \sum_{a,x,y} \tr{M_{0|\star} (\sigma_{a|xy} \otimes F_{axy})}{BB'} \\
 &= \frac{1}{4} \sum_{a,x,y} \tr{M_{0|\star} (\tr{(M_{a|x}\otimes \mathcal{E}^{B}_y)\rho_{AB}}{A} \otimes F_{axy})}{BB'}.
     \end{align}
Next, we insert an identity map on the system $\cH_{B'}$ by introducing the additional system on $B_c'B_c$ in the following form: $\tr{(\id^{B'_c}\otimes \ket{\phi^{+}}\bra{\phi^{+}}_{B_cB'})(\ket{\phi^{+}}\bra{\phi^{+}}_{B'_cB_c}\otimes \id^{B'})}{B_c'B_c}$. Then, the Bell functional $I_{BwI}^{*}[\textbf{\textrm{p}}]$ can be written as
\begin{align}
     I_{BwI}^{*}[\textbf{\textrm{p}}]
 &= \frac{1}{4} \sum_{a,x,y} \tr{ (M_{a|x} \otimes M_{0|\star} \otimes \ket{\phi^{+}}\bra{\phi^{+}}_{B'_cB_c}) \,\, \mathcal{E}^{B}_y \,\, (\rho_{AB} \otimes F_{axy} \otimes \ket{\phi^{+}}\bra{\phi^{+}}_{B_cB'})}{ABB'B_c'B_c}.
     \end{align}     
We can now rewrite the Bell functional using the quantum state $\rho'$ defined in Eq.~\eqref{eq:app-newstate} as follows
\begin{align}\label{eq:app-mathfinal}
     I_{BwI}^{*}[\textbf{\textrm{p}}]
 &= \frac{1}{4} \sum_{a,x,y} \tr{ (M_{a|x} \otimes \ket{\phi^{+}}\bra{\phi^{+}}_{B'_cB_c}) (\rho'_{AB_c} \otimes F_{axy} )}{AB_cB_c'}.
     \end{align}
From Eq.~\eqref{eq:app-diag} it follows that for the form of $I_{BwI}^{*}[\textbf{\textrm{p}}]$ given in Eq.~\eqref{eq:app-mathfinal} , $I_{BwI}^{*}[\textbf{\textrm{p}}]\geq 0$ for all quantum states $\rho'$ (note that Eq.~\eqref{eq:app-diag} is valid for all CPTP maps $\mathcal{E}_y$; here, we simplified the expression such that the CPTP map on the system $B_c$ is the identity). This concludes the proof.

\section{Activation in the Bob-with-input scenario - example}\label{app:example}
In this appendix, we use the construction from Theorem~\ref{thm-bwi} to derive a Bell functional that certifies post-quantumness of the assemblage $\As^{PTP}_{\A|\X\Y}$ introduced in Ref.~\cite[Eq.~(6)]{sainz2020bipartite}, which belongs to the set $\As^{QC}$.

We focus on the assemblage $\As^{PTP}_{\A|\X\Y}$ in a Bob-with-input EPR scenario where $\X=\{1,2,3\}$, $\A=\{0,1\}$ and $\Y=\{0,1\}$. The elements of this assemblage, $\{\sigma^{PTP}_{a|xy}\}$, can be mathematically expressed as follows: 
\begin{align}\label{eq:PTP}
\sigma^{PTP}_{a|xy} =  \frac{1}{4} (\id + (-1)^{a+\delta_{x,2}+\delta_{y,1}}\sigma_x) \,, 
\end{align}
where the operators $\sigma_1$, $\sigma_2$, and $\sigma_3$ (since the choice of Alice's measurement $x \in \{1,2,3\}$) denote Pauli X, Y, and Z operators respectively. In Ref.~\cite[Appendix D]{sainz2020bipartite}, it was shown that $\As^{PTP}_{\A|\X\Y}$ is a post-quantum assemblage and that if Bob  measures his subsystem, the correlations that arise between him and Alice always admit a quantum realisation.

To show that the post-quantumness of $\As^{PTP}_{\A|\X\Y}$ can be activated, we will use the EPR functional constructed in Ref.~\cite[Eq.~(D3)]{sainz2020bipartite}. It is specified by the operators $F^{PTP}_{axy}=\frac{1}{2} (\id - (-1)^a \sigma_x)^{T^{y}}$,
where $T^y$ denotes that the transpose operation is applied when $y=1$. The classical bound of this functional is $\beta^{C}_{PTP}=1.2679$ and the no-signalling bound is $\beta^{NS}_{PTP} = 0$. The assemblage  $\As^{PTP}_{\A|\X\Y}$ saturates the minimum bound of this EPR functional. Although the quantum bound $\beta^{Q}_{PTP}$ of this functional is not known exactly, it can be lower bounded by the almost-quantum bound, i.e., a bound calculated for all assemblages lying in the outer approximation of the quantum set called the almost-quantum set, which is given by $\beta^{AQ}_{PTP}=0.4135$. On the other hand, $\beta^{Q}_{PTP}$ is upper bounded by the classical bound. Hence, we know that $1.2679 \geq \beta^{Q}_{PTP} \geq 0.4135$.
In order to use the activation protocol introduced in Section~\ref{sec:protocol-bwi}, we should normalise the operators $\{F^{PTP}_{axy}\}$ and construct an EPR functional with a quantum bound equal to zero\footnote{In Section~\ref{sec:protocol-bwi}, we assume that the quantum bound of the functional is equal to zero. This assumption is necessary for the proof that the protocol also works for measurements different than $M^{BB'}_{b=0|z=\star}=\ket{\phi^{+}}\bra{\phi^{+}}$, which is given in Appendix~\ref{app:thm}.}. However, since the value of $\beta^{Q}_{PTP}$ is unknown, this is impossible. Instead, we normalise the operators $\{F^{PTP}_{axy}\}$ using the lower bound on $\beta^{Q}_{PTP}$, namely $\beta^{AQ}_{PTP}$. Then, for any quantum or almost-quantum assemblage $\{\sigma^{AQ}_{a|xy}\}$, the following holds
\begin{align}\label{eq:EPR_functional}
    \tr{\sum_{a,x,y}\widetilde{F}^{PTP}_{axy}\sigma^{AQ}_{a|xy}}{}\geq 0, 
\end{align}
where the operators $\{\widetilde{F}^{PTP}\}$ are given by 
\begin{align}\label{eq:EPR-func}
\widetilde{F}^{PTP}_{axy}=F^{PTP}_{axy} - \frac{\beta^{AQ}_{PTP}}{|\X|\,|\Y|} \, \id \,.
\end{align}
To certify post-qauntumness of $\As^{PTP}_{\A|\X\Y}$ solely from correlations $\{p(a,b,c|x,y,z,w)\}$, we need to transform the EPR functional~\eqref{eq:EPR_functional} to a Bell functional of the form~\eqref{eq:Bell-functional-bwi}. The coefficients $(\xi^{PTP})^{axy}_{cw}$ can be read off Eq.~\eqref{eq:EPR-func} and are given by
\begin{align}\label{eq:Bell-coeff}
(\xi^{PTP})^{axy}_{cw}=\delta_{c,a \oplus 1  \oplus y\delta_{x,2}} \,\, \delta_{w,x\oplus_3 1} - \delta_{w,1} \frac{\beta^{AQ}_{PTP}}{|\X|\,|\Y|}, 
\end{align}
where $\oplus$ and $\oplus_3$ denote sum modulo 2 and modulo 3, respectively. Using the reasoning of Eq.~\eqref{eq:corr2}, we know that the Bell functional specified by coefficients given in Eq.~\eqref{eq:Bell-coeff}, when evaluated on a post-quantum assemblage (which is not almost-quantum), must satisfy $I_{BwI}^{PTP}[\textbf{\textrm{p}}] < 0$. It is easy to check that if $\As_{\C|\W}=\widetilde{\As}_{\C|\W}$, $M^{BB'}_{b=0|z=\star}=\ket{\phi^{+}}\bra{\phi^{+}}$ and $\As_{\A|\X\Y}=\As_{\A|\X\Y}^{PTP}$, the functional $I_{BwI}^{PTP}$ evaluates to $I_{BwI}^{PTP}[\textbf{\textrm{p}}^{PTP}]=-\beta^{AQ}$, which shows the activation of post-quantumness of $\As_{\A|\X\Y}^{PTP}$.

\section{MDI assemblage from the set $\N^{QC}$}\label{app:mdiass}

Here we show that the post-quantum measurement-device-independent assemblage $\N^{PTP}_{\A\B|\X}$ introduced in Ref.~\cite[Eq.~(28)]{channelEPR} belongs to the set $\N^{QC}$. $\N^{PTP}_{\A\B|\X}$ is illustrated in Fig.~\ref{fig:MDI-PTP} and it can be mathematically expressed as
\begin{align}
\N^{PTP}_{\A\B|\X} &= \left\{\Ne^{PTP}_{ab|x}\right\}_{a\in \A, \, b \in \B, \, x \in \X} \,,\\
\text{with} \quad & \begin{cases} \Ne^{PTP}_{ab|x} = \tr{({M}_{a|x} \otimes \id_{B_{out}}) (\id_A \otimes N_{b}) (\id_A \otimes \textrm{CT}^{BB_{in} \rightarrow B})[\rho_{AB} \otimes \rho_y]}{} \,, \nonumber \\
{M}_{a|1} = \frac{\id + (-1)^a \sigma_x}{2}\,,\quad {M}_{a|2} = \frac{\id + (-1)^a \sigma_y}{2}\,,\quad
{M}_{a|3} = \frac{\id + (-1)^a \sigma_z}{2}\,, \nonumber \end{cases}
\end{align}
where $\rho_{AB}=\ket{\phi^{+}}\bra{\phi^{+}}$ and $\sigma_x$, $\sigma_y$, and $\sigma_z$ are the Pauli operators. The processing CT$^{BB_{in} \rightarrow B}$ can be viewed as a controlled-transpose operation on $BB_{in}$, where $B_{in}$ acts as the control qubit and $B$ is the system that is being transposed. Then, the system $B_{in}$ is traced-out and the system $B$ is measured by Bob to generate a classical outcome $b$. The measurement elements are $N_0=\frac{1}{3}\id+\frac{1}{3}\sigma_y$ and $N_1=\frac{2}{3}\id-\frac{1}{3}\sigma_y$. In Ref.~\cite{channelEPR}, it was shown that $\N^{PTP}$ is a post-quantum assemblage. 
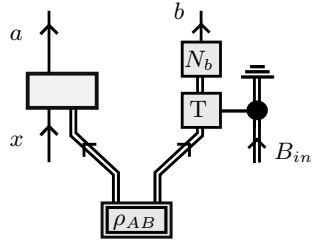
\begin{figure}[h!]
  \begin{center}
{\includegraphics[width=0.2\textwidth]{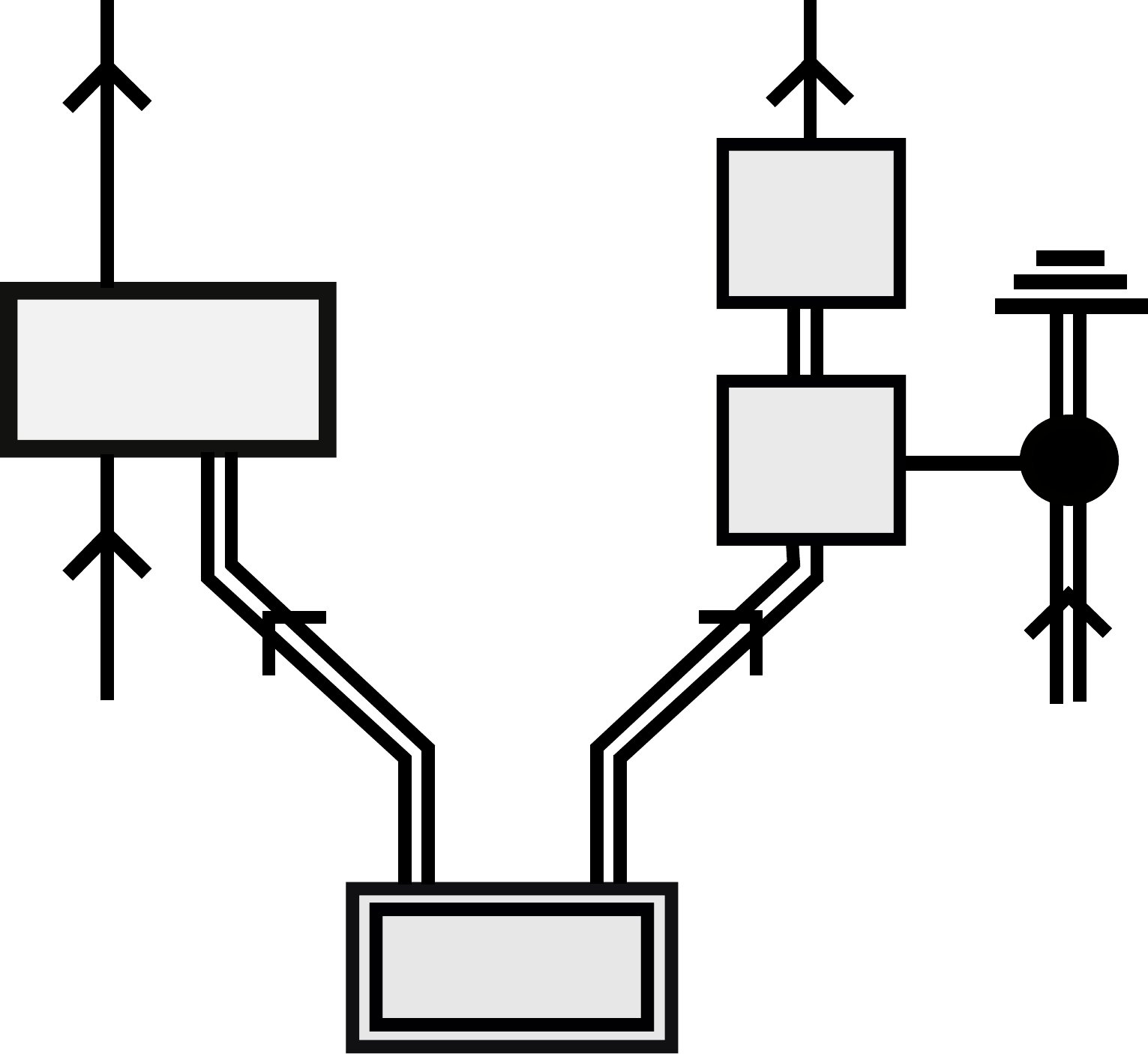}
\put(-100,35){$x$}
\put(-100,75){$a$}
\put(-38,83){$b$}
\put(0,30){$B_{in}$}
\put(-34,65){$N_b$}
\put(-32,46){T}
\put(-61,5){$\rho_{AB}$}
}

  \end{center}

  \caption{Mathematical depiction of the MDI assemblage $\N^{PTP}$: Alice and Bob share a quantum state $\rho_{AB}$; Alice performs measurements on her system, while Bob performs a controlled-transpose operation CT$^{BB_{in} \rightarrow B}$ followed by a measurement $N_b$. }\label{fig:MDI-PTP}
\end{figure}

Let us examine a situation when Bob's input state defined on $\cH_{B_{in}}$ is fixed and given by $\rho_y$. The correlations in such set-up are given by $p(ab|xy) = \Ne^{PTP}_{ab|x}(\rho_y)$ and can be written as
\begin{align}\label{eq:PTP-MDI}
 p(ab|xy) &= \tr{({M}_{a|x} \otimes \id_{B_{out}}) (\id_A \otimes N_{b}) (\id_A \otimes \textrm{CT}^{BB_{in} \rightarrow B})[\rho_{AB} \otimes \rho_y]}{} \,. 
\end{align}
Here, $\rho_y$ acts only as a control qubit. Hence, we can represent the controlled-transpose operation as labeled by the classical label $y$: $(\id_A \otimes \textrm{CT}^{BB_{in} \rightarrow B})[\rho_{AB} \otimes \rho_y] \equiv (\id_A \otimes \textrm{CT}^{B}_{y})\rho_{AB}$. Transpose is a positive and trace-preserving map, hence it has a dual map which is positive and unital. For the transpose map, its dual is also a transpose map. Therefore, instead of applying the transpose map on the system $B$, one can imagine the controlled transpose is applied on the measurement $N_b$ instead; hence, we can rewrite Eq.~\eqref{eq:PTP-MDI} as
\begin{align}\label{eq:PTP-MDI2}
 p(ab|xy) &= \tr{({M}_{a|x} \otimes  N_{b|y})\rho_{AB}}{} \,. 
\end{align}
The new general measurement $N_{b|y}$ depends on the control qubit and can be given by the original measurement $N_b$, its transpose, or superposition thereof. It follows that the correlations defined as in Eq.~\eqref{eq:PTP-MDI2} admit of a quantum realisation. 

\section{Activation protocol beyond qubit assemblages}\label{app:higher}

In this appendix, we show how to adapt the activation protocols such that post-quantumness of higher-dimensional assemblages can be certified. First, we use methods introduced in Ref.~\cite[Appendix B]{activation} to generalise the self-testing stage of the protocol. Then, we adapt the Bell inequalities used in the protocols to the higher-dimensional case for each scenario.

In all three scenarios, we consider assemblages with the dimension of Bob's output system being $d>2$. We embed these systems in a higher-dimensional Hilbert space, where each qudit lives on a Hilbert space arising from a parallel composition of $n$ qubits. The basis for this space is given by tensor products of Pauli operators. Following the intuition of the qubit-assemblages protocols, the first step of the protocol need to self-test the basis elements of the relevant Hilbert space. This step of the protocol is the same for the Bob-with-input, measurement-device-independent and channel EPR scenarios. 

The fact that one can only self-test a mixture of an assemblage and its transpose will again pose a problem here. This is analogous to the problem that occurs in the protocol for the channel EPR scenario, where if the self-tested assemblages would form a tensor product of the form $\widetilde{\sigma}_{c|w} \otimes (\widetilde{\sigma}_{d|u})^{T}$, a false-positive detection of post-quantumness could be detected (we discuss this problem in the Appendix~\ref{app:mixture-channel}). Therefore, instead of just self-testing tensor product of Pauli operators, the self-testing stage of the protocol must guarantee that the assemblage is of the form
\begin{align}\label{eq:app-tensor}
    \sigma^{(r)}_{\cc|\w}=r   \bigotimes_{i=1}^n \widetilde{\sigma}_{c_i|w_i}+(1-r)  \bigotimes_{i=1}^n (\widetilde{\sigma}_{c_i|w_i})^{T}. 
\end{align}
Here, $\cc=(c_1,...,c_n)$ with $c_i \in \{0,1\}$ and $\w=(w_1,...,w_n)$ with $w_i \in \{1,2,3\}$. The technique to self-test $\{\sigma^{(r)}_{\cc|\w}\}$ of the form~\eqref{eq:app-tensor} was first introduced in Ref.~\cite{bowles2018self} and later used in the activation protocol introduced in Ref.~\cite{activation}. We will not recall the exact self-testing statement here; all the technical details can be found in Refs.~\cite{activation,bowles2018self}.

If the first step of the activation protocol is successful, the probabilities generated in the protocol set-up can be used for post-quantumness certification using tailored Bell functionals. Below, we adapt the Bell functional for each scenario to the higher-dimensional case.

\subsection{Bob-with-input EPR scenario}

Consider a Bob-with-input assemblage with the dimension of the system $\cH_B$ being greater than 2. Analogous to the case of a qubit assemblage, define the Bell functional as
\begin{align}
    I_{BwI}^{*}[\textbf{\textrm{p}}] \equiv \sum_{a,x,y,c,w} \xi^{axy}_{cw} p(a,b=0,c|x,y,z=\star,w).
\end{align}
To compute the quantum bound of the functional, recall that
\begin{align}
    p(a,b=0,c|x,y,z=\star,w)=\tr{M^{BB'}_{b=0|z=\star} (\sigma_{a|xy} \otimes \sigma_{\cc|\w}^{(r)})}{},
\end{align}
where $\sigma_{\cc|\w}^{(r)}$ is given by Eq.~\eqref{eq:app-tensor}. For now, assume that $\sigma_{\cc|\w}^{(r=1)}=\bigotimes_{i=1}^n \widetilde{\sigma}_{c_i|w_i}$ and $M^{BB'}_{b=0|z=\star} = \ket{\phi^{n}}\bra{\phi^{n}}$, with $ \ket{\phi^{n}} = \frac{1}{\sqrt{2^n}}\sum_{k=0}^{2^n -1} \ket{kk}_{BB'}$. Then, the Bell functional can be written as
\begin{align}
    I_{BwI}^{*}[\textbf{\textrm{p}}] & = \sum_{a,x,y,c,w} \xi^{axy}_{cw} \frac{1}{2^n} \tr{ (\bigotimes_{i=1}^n \widetilde{\sigma}_{c_i|w_i})^T \sigma_{a|xy}}{}, \\ \nonumber
    & = \frac{1}{2^{n+1}} \sum_{a,x,y} \tr{ F_{axy} \sigma_{a|xy}}{}.
\end{align}
It follows that $I_{BwI}^{*}[\textbf{\textrm{p}}] \geq 0$ for any $\textbf{\textrm{p}}$ arising from a quantum assemblage $\As_{\A|\X\Y}$ with elements $\{\sigma_{a|xy}\}$.

To have a universal protocol for activation of post-quantumness in high-dimensional Bob-with-input assemblages, there are two things left to show. Firstly, it is necessary to show that the quantum bound of the functional $I_{BwI}^{*}$ does not change when it is optimised over measurements different than $M^{BB'}_{b=0|z=\star} = \ket{\phi^{n}}\bra{\phi^{n}}$. The proof of this claim is analogous to the one presented in Appendix~\ref{app:thm}. Secondly, the protocol cannot show false-positive detection when the self-tested assemblage is of the form given in Eq.~\eqref{eq:app-tensor}. Let us first focus on the case when $\sigma_{\cc|\w}^{(r=0)}=(\bigotimes_{i=1}^n \widetilde{\sigma}_{c_i|w_i})^T$. Then, the Bell functional is given by
\begin{align}
    I_{BwI}^{*}[\textbf{\textrm{p}}] = \frac{1}{2^{n+1}} \sum_{a,x,y} \tr{ F_{axy} (\sigma_{a|xy})^T}{}.
\end{align}
The transpose map applied on a quantum Bob-with-input assemblage transforms its elements as follows
\begin{align}\label{eq:app-BwIQ}
    (\sigma_{a|xy})^T &=  \left(\tr{(M_{a|x}\otimes \mathcal{E}_y)\rho_{AB}}{A}\right)^T, \\ \nonumber
    &= \tr{(\id \otimes T) (M_{a|x}\otimes \mathcal{E}_y)\rho_{AB}}{A}.
\end{align}
To show that $(\sigma_{a|xy})^T$ is a quantum assemblage if $(\sigma_{a|xy})$ has a quantum realisation, we will use Lemma~\ref{lemma}. Let $\mathcal{E}_y'$ be a CPTP map that is indexed by $y$. Due to Lemma~\ref{lemma}, we can shift the transpose in Eq.~\eqref{eq:app-BwIQ} to the quantum state as
\begin{align}
    (\sigma_{a|xy})^T &=  \tr{(\id \otimes T) (M_{a|x}\otimes \mathcal{E}_y)\rho_{AB}}{A}, \\ \nonumber
    & = \tr{(M_{a|x}\otimes \mathcal{E}_y')\rho_{AB}^{T_B}}{A}, \\ \nonumber
    &= \tr{(M_{a|x}\otimes \mathcal{E}_y')^{T_A} \rho_{AB}^{T}}{A}, \\ \nonumber
     &= \tr{((M_{a|x})^{T} \otimes \mathcal{E}_y') \rho_{AB}'}{A}, 
\end{align}
where $\{(M_{a|x})^{T}\}$ are valid POVMs and $\rho'$ is a valid quantum state. Therefore, we can see that if $\{\sigma_{a|xy}\}$ form a valid quantum assemblage, the elements $\{(\sigma_{a|xy})^T \}$ also have a quantum realisation in terms of the measurements $\{(M_{a|x})^{T}\}$, CPTP maps $\{ \mathcal{E}_y' \}$ and the state $\rho_{AB}'$. It follows that $I_{BwI}^{*}[\textbf{\textrm{p}}] \geq 0$ for any $\textbf{\textrm{p}}$ arising from a quantum assemblage, even when $\sigma_{\cc|\w}^{(r=0)}=(\bigotimes_{i=1}^n \widetilde{\sigma}_{c_i|w_i})^T$. Hence, by convexity, any correlation generated when the self-tested assemblage has the form~\eqref{eq:app-tensor} cannot create a false-positive detection of post-quantumness.

\subsection{Measurement-device-independent EPR scenario}

Recall that the Bell functional for the MDI scenario reads
\begin{align}
    I_{MDI}^{*}[\textbf{\textrm{p}}] \equiv \sum_{a,b,x,c,z} \xi^{abx}_{cz} p(a,b,c|x,y=\star,z).
\end{align}
If the first step of the protocol was successful, the relevant probabilities can be written as
\begin{align}
p(a,b,c|x,y=\star,z)=\Ne_{ab|x}(\sigma_{\cc|\z}^{(r)}),
\end{align}
where
\begin{align}\label{eq:app-tensor-MDI}
    \sigma^{(r)}_{\cc|\z}=r   \bigotimes_{i=1}^n \widetilde{\sigma}_{c_i|z_i}+(1-r)  \bigotimes_{i=1}^n (\widetilde{\sigma}_{c_i|z_i})^{T}. 
\end{align}
For now, assume that $\sigma_{\cc|\z}^{(r=1)}=\bigotimes_{i=1}^n \widetilde{\sigma}_{c_i|z_i}$ Then, the Bell functional can be written as
\begin{align}
    I_{MDI}^{*}[\textbf{\textrm{p}}] \equiv \sum_{a,b,x,c,z} \xi^{abx}_{cz} \Ne_{ab|x}(\bigotimes_{i=1}^n \widetilde{\sigma}_{c_i|z_i}).
\end{align}
Notice that a Choi matrix of a $d$-dimensional qudit assemblage $\{\Ne_{ab|x}(\cdot)\}$ is expressed in terms of a maximally entangled state $\ket{\phi^{n}}\bra{\phi^{n}}$, with $ \ket{\phi^{n}} = \frac{1}{\sqrt{2^n}}\sum_{k=0}^{2^n -1} \ket{kk}_{B'C}$ and $n=d$. In what follows, we encode this state in a space that arises from a parallel composition of many qubits. Then, similar to the qubit case, we can express the probabilities $p(a,b,c|x,y=\star,z)$ as
\begin{align}
    p(a,b,c|x,y=\star,z)=\tr{(\bigotimes_{i=1}^n\widetilde{M}^{C_i}_{c_i|z_i}) \, J(\Ne_{ab|x})}{},
\end{align}
and the functional is given by
\begin{align}
    I_{MDI}^{*}[\textbf{\textrm{p}}] &= \sum_{a,b,x,c,z} \xi^{abx}_{cz} \tr{(\bigotimes_{i=1}^n\widetilde{M}^{C_i}_{c_i|z_i}) \, J(\Ne_{ab|x})}{}, \\ \nonumber
&=\tr{\sum_{a,b,x}F_{abx}\, J(\Ne_{ab|x})}{}.
\end{align}
Therefore, $I_{MDI}^{*}[\textbf{\textrm{p}}] \geq 0$ for any $\textbf{\textrm{p}}$ arising from a quantum assemblage $\N_{\A\B|\X}$ with elements $\{\Ne_{ab|x}\}$.

Now, we will show that a false-positive detection is not possible even when the self-tested assemblage is of the form given in Eq.~\eqref{eq:app-tensor-MDI}. Let $\sigma_{\cc|\z}^{(r=0)}=(\bigotimes_{i=1}^n \widetilde{\sigma}_{c_i|z_i})^T$. Then, the Bell functional can be written as
\begin{align}
    I_{MDI}^{*}[\textbf{\textrm{p}}] =\tr{\sum_{a,b,x}F_{abx}\, \left(J(\Ne_{ab|x})\right)^{T}}{}.
\end{align}
To show that a transpose operation on the Choi matrix of an MDI assemblage cannot generate post-quantumness, we will again use Lemma~\ref{lemma}. Recall that any quantum MDI assemblage is defined by a set of measurement channels $\Theta_b^{BB' \rightarrow B_{out}}$. For each $b \in \B$, we can decompose $\Theta_b^{BB_{in} \rightarrow B_{out}}$ in terms of its Kraus operators to see that a transpose on the output space of this measurement channel $B_{out}$ can be passed onto the input of the channel defined on $BB_{in}$. It follows that
\begin{align}
\left(J(\Ne_{ab|x})\right)^{T} &= \tr{(M_{a|x} \otimes \Theta_b^{BB_{in} \rightarrow B_{out}} \otimes \id^C)[\rho_{AB} \otimes \ket{\phi^{n}}\bra{\phi^{n}}_{B_{in}C}]}{A}^{T}, \\ \nonumber 
&= \tr{(\id^{A} \otimes T^{B_{out}C})(M_{a|x} \otimes \Theta_b^{BB_{in} \rightarrow B_{out}} \otimes \id^C)[\rho_{AB} \otimes \ket{\phi^{n}}\bra{\phi^{n}}_{B_{in}C}]}{A}, \\ \nonumber 
&= \tr{(M_{a|x} \otimes \Theta_b^{'BB_{in} \rightarrow B_{out}} \otimes \id^C)[\rho_{AB}^{T_B} \otimes (\ket{\phi^{n}}\bra{\phi^{n}}_{B_{in}C})^{T}]}{A}, \\ 
\nonumber 
&= \tr{(M_{a|x} \otimes \Theta_b^{'BB_{in} \rightarrow B_{out}} \otimes \id^C)^{T_A}[\rho_{AB}^{T} \otimes (\ket{\phi^{n}}\bra{\phi^{n}}_{B_{in}C})^{T}]}{A}, \\ \nonumber 
&= \tr{((M_{a|x})^T \otimes \Theta_b^{'BB_{in} \rightarrow B_{out}} \otimes \id^C)[\rho_{AB}^{T} \otimes (\ket{\phi^{n}}\bra{\phi^{n}}_{B_{in}C})^{T}]}{A}.
\end{align}
Therefore, if $J(\Ne_{ab|x})$ corresponds to a valid quantum assemblage, its transpose $\left(J(\Ne_{ab|x})\right)^{T}$ also have a quantum realisation in terms of the measurements $\{(M_{a|x})^{T}\}$, measurement channels $\{ \Theta_b^{'BB_{in} \rightarrow B_{out}} \}$ and the states $\rho_{AB}^{T}$ and $(\ket{\phi^{n}}\bra{\phi^{n}}_{B_{in}C})^{T}$. Since this assemblage has a quantum realisation, it cannot create a false-positive detection in the activation protocol.

\subsection{Channel EPR scenario}

In the protocol for the channel EPR scenario, we used the self-test of a tensor product of Pauli operators (accompanied by an alignment procedure) even when dim($\cH_{B_{out}}$)$=2$. Moreover, the technique to show that there are no false-positive detections does not rely on any of the subsystems having a particular dimension. Therefore, the generalisation of this protocol to the qudit case is straight-forward. 

%

\end{document}